\documentclass[10pt, twocolumn]{article}
\usepackage{graphicx} 

\usepackage{styles/style}
\usepackage{styles/QIA-Network-Diagram-Style-Guide}
\usepackage[giveninits=true, url=false, maxnames=20, eprint=true]{biblatex}
\color{black}
\pagecolor{white}

\title{Arqon: A suite of control applications enabling a reliable quantum network}
\author{Scarlett Gauthier\footnote{These authors contributed equally.}, Thomas R. Beauchamp$^{*}$, Stephanie Wehner}
\date{\today}

\begin{document}

\maketitle

\begin{abstract}
A quantum network's purpose is to enable users to execute applications on end nodes.
This requires the network to provide the service of creating entangled links between those nodes.
Users of mature networks, such as the internet or the telephone network expect accepted service demands to be met reliably. 
We first define reliability requirements that extend classical computer network concepts to quantum network service delivery.  
We then introduce Arqon, a suite of control applications designed to deliver reliable service in centrally controlled quantum networks.
We demonstrate through both analytic and numerical evaluation that Arqon satisfies all reliability requirements for accepted demands. These evaluations consider static network topologies.
We provide a complete Python implementation and perform complexity analysis showing that admission control scales as $O(k^3)$ in the number of incoming demands~$k$ and schedule computation scales as ${O(N^3)}$ in the number of accepted demands to schedule~$N$.
\end{abstract}

\section{Introduction}\label{sec:Introduction}

A quantum network enables users to execute applications that enable new functionalities. These  
include secure remote computation \cite{bqc1, bqc2, childs_secure_2005}, secure communication which does not rely on computational assumptions \cite{BB84, e91}, fast coordination of decisions between remote parties without the real-time exchange of messages \cite{loadBalancing, EntDecisionCoordinationJiang, GraphRendezvous, Entanglementassistedorientationspace, GraphRendezvous2}, anonymous leader election \cite{leaderElectionPaper}, and improvements in the precision of metrology \cite{ClockSynch, TelescopesBaseline}.
These applications require quantum communication between quantum devices, supplementing the classical communication which underpins classical network applications.
 Entangled links between end nodes, known as end-to-end links, are key resources that can enable arbitrary quantum communication operations.
The service that a quantum network must provide to end nodes is to create these end-to-end entangled links according to the quality, quantity, and time based requirements of the end node applications.

To correctly schedule the quantum communication operations needed in local application execution, the operating systems of programmable end nodes may rely on a network schedule. This schedule dictates when the nodes can access internal resources of the network to attempt end-to-end entangled link generation.
In~\cite{NetArx}, Beauchamp~\textit{et~al.} proposed a centrally controlled quantum network architecture which is designed to provide this service in a way that is compatible with the state of the art operating systems for quantum processing nodes, QNodeOS~\cite{QNodeOS} and its upgrade Qoala~\cite{Qoala}. 

The network architecture in~\cite{NetArx} is modular, in the sense that it is defined to consist of a set of core functionalities, which may be implemented in separate control applications that can each be modified or upgraded independently of the others. A control application packages related network functionalities such as admission control and schedule computation, and provides specific algorithms implementing these functionalities. 
The network architecture is centered around the periodic computation and distribution of network schedules, enabling internal resources and end nodes to receive schedules in advance of the times at which they can be executed.
In a proof of principle implementation and evaluation~\cite{NetArx}, the proportion of satisfied demands decreased substantially as the number of submitted demands increased, exposing a need for improved control applications. 
More broadly, an outstanding challenge in the design of quantum network architectures is to develop control applications that can deliver reliable service to end nodes.

In computer networking, a system or protocol is reliable if it is able to consistently deliver services to their intended recipients without error or significant delay \cite{ReliabilityInternetStandardBook}. Adapting this principle to the service of a quantum network, we define the essential features of reliable service:
\begin{enumerate}[label=\textit{(R\arabic*)}]
    \item demands receive accept/reject responses;\label{req: 1: responses} 
    \item accepted demands are satisfied with a tolerably high probability;\label{req: 2: satisfied}
    \item accepted demands are satisfied before their deadlines; \label{req: 3: on time}
    \item and the amount or frequency of demands from other users does not disrupt service to accepted demands.\label{req: 4: independent}
\end{enumerate}

 Arqon is a suite of control applications for the centrally controlled quantum network architecture in~\cite{NetArx}, which is designed to overcome the challenge of delivering reliable service to end nodes. 
 Arqon comprises four control applications: A \textit{Network Manager}, a \textit{Demand Manager}, a \textit{Network Scheduler}, and a \textit{Schedule Manager}. 
 The Network Manager is responsible for maintaining an accurate overview of the network, including all components and the entanglement generation capabilities of the components. 
 The Demand Manager is responsible for all processes that explicitly relate to demands, such as registration of new demands, creating accept/reject responses to send to end nodes, and re-formatting demands into an internal representation.
 The Network Scheduler produces network schedules based on end node demands. In addition to a \textit{Compute Schedule} process, the Network Scheduler also includes an \textit{Admit Tasks} process that determines whether demands are accepted for service or rejected. 
 Finally, the Schedule Manager directs the distribution of network schedules to all network components.

In this work we make the following contributions: \begin{itemize}
    \item We introduce \textit{Arqon}, the first suite of control applications for a quantum network architecture capable of delivering reliable service to all accepted demands. To make this achievement precise, we propose a set of reliability requirements for service from a quantum network. We deliver a full performance analysis of our Network Scheduler control application, proving that it can meet the proposed requirements. Moreover, we validate in simulation Arqon's performance for a variety of network topologies and sets of demands.
    \item We provide a Python implementation of Arqon that enables researchers to simulate realistic multi-user quantum networks, investigate performance under diverse traffic patterns, and develop extensions to our control algorithms. In particular, the Arqon simulator allows configuration of a network topology, the entanglement generation capabilities (success rate and fidelity) between end nodes, and sets of end node application requirements.
    \item We prove that the operational complexity of our implementation of the Network Scheduler is $${O \bigg{(}k 2^{R} \big{(} (N + k)^2 +R \big{)} + (N+R)^2 + N^3 R \bigg{)}},$$ where $N$ is the current number of accepted demands to schedule, $k$ is the number of internal representations of new demands requiring an accept/reject decision, and $R$ is the number of internal resources in the network (network components that are not end nodes). 
    In our evaluations, we observe that the time to produce a network schedule ($t_{\text{compute}}$) is always significantly less than the allowed time, known as a scheduling interval ($T^{SI}$). That is, ${t_{\text{compute}}} \ll T^{SI}$. 
\end{itemize}

The notation used throughout this paper is summarized in Table~\ref{tab: notation summary}, included in Appendix~\ref{app: notation}. 
The rest of this paper is organized as follows: In Section~\ref{sec: Related Work} we discuss related work. In Section~\ref{sec: design considerations} we discuss design considerations for Arqon. In Section~\ref{sec: Arqon Suite of Control Applications} we define each control application in turn. In Section~\ref{sec: Performance Analysis} we state and prove performance results for the abstractly defined control applications. 
In Section~\ref{sec: Implementation} we discuss a specific implementation of Arqon in Python. In Section~\ref{sec: Complexity} we provide an operational complexity analysis of this implementation. In Section~\ref{sec: Evaluation} we discuss results from numeric simulations. Finally, we conclude in Section~\ref{sec: Conclusion} and discuss our outlook on future research goals.

\section{Related Work}
\label{sec: Related Work}
The Arqon suite of control applications is designed to integrate with the centrally controlled quantum network architecture defined in \cite{NetArx}.
This architecture is compatible with the network protocol stack proposed by Dahlberg~\textit{et~al.}~in~\cite{dahlberg_link_2019}, where the stack is decomposed into layers, each of which exposes a functionality to the layer above and exploits the functionality of the layer below, as in the classical OSI network model~\cite{OSIrefModel}.
We make use of the terminology of this network stack. The \textit{physical layer} of the network is responsible for attempting entanglement generation. Above the physical layer is the \textit{link layer}, which manages link level entanglement generation, transforming it into a robust service. The link layer service was experimentally realized in Pompili~\textit{et~al.}~in~\cite{pompili_experimental_2022}, demonstrating it's viability as a functional layer of a quantum network. Operations such as scheduling entanglement swaps in a repeater chain and aligning entanglement generation with the memory management service of a quantum node are functions of the link layer.
The network architecture~\cite{NetArx} defines a centrally controlled structure for the \textit{network layer}, which is responsible for coordinating end-to-end entanglement generation. 
The combination of this architecture with the Arqon suite of control applications constitutes an implementation of the network layer. 

Experiments on leading quantum network hardware platforms present a major technological challenge: the expected time to generate an entangled link between neighboring nodes at metropolitan distance scales is almost as long as the maximum storage time. In a state-of-the-art demonstration based on trapped ion nodes~\cite{IonsHighRateEnt} the average reported time to generate one entangled pair over an effective 10 km link was 450 ms, while the average maximum storage time was 550 ms with a standard deviation of 36 ms.
In networks where the maximum storage time of entangled links does not exceed by orders of magnitude the time required to generate such a link, the multiple neighboring entangled links required to construct an end-to-end link will rarely co-exist.
To overcome this technical challenge, Arqon produces network schedules that simultaneously allocate nodes along end-to-end routes for durations of time calculated based on end-to-end entanglement generation rates. This method of resource allocation is reminiscent of circuit switching in wired telephone networks~\cite{TelephoneCircuits, tanenbaum2003computer}, Frame Relay~\cite{FrameRelay1, FrameRelay2}, and Autonomous Transfer Mode (ATM) networks~\cite{JainATM, tanenbaum2003computer, steenhaut_scheduling_1997}.

In wired telephone networks, multiple calls share physical links through time division multiplexing (TDM), where each call receives a guaranteed time slot within a frame that repeats at fixed frequency for the entire call duration~\cite{TelephoneCircuits, tanenbaum2003computer}. The maximum number of simultaneous calls is determined by the number of time slots per frame on each link along a route.
In ATM networks, multiple virtual circuits share physical links through statistical multiplexing of 53-byte cells transmitted on demand rather than on fixed schedules~\cite{JainATM, tanenbaum2003computer}. Users requesting virtual circuits provide traffic contracts specifying bounds on data production: sustained cell rate (average long-term rate), peak cell rate (maximum rate), and maximum burst size (longest burst at peak rate). Admission control evaluates these contracts and, if accepted, calculates effective bandwidth to reserve along the entire route for the virtual circuit's duration~\cite{steenhaut_scheduling_1997}.
Frame Relay uses similar mechanisms to ATM but with variable-length frames (1 to 8,192 bytes) and a simpler single-service-class quality of service model with three parameters: a guaranteed data transfer rate (Committed Information Rate (CIR)), an allowed burst size at the CIR and a maximum amount that a burst may exceed the CIR~\cite{FrameRelay1, FrameRelay2}.
In contrast, Arqon schedules exclusive end-to-end physical access for demand-dependent durations separated by demand-dependent minimum separation intervals when routes fully release access. Network-wide schedules are pre-computed sequentially for extended scheduling intervals that are long relative to typical access durations. Arqon's admission control evaluates user demands, with accepted demands receiving service agreements guaranteeing a minimum scheduling frequency per interval while respecting minimum separation requirements.



Many proposals of quantum network architectures and protocols, including those of~\cite{cicconetti_request_2021, pirker_quantum_2019, ESDI, van_meter_path_2013, dur_quantum_1999, van_meter_quantum_2022, LeveragingInternetPrinciples, QuantNetTwoLevel}, have considered the problem of coordinating the generation of end-to-end entangled links, requiring the use of network resources. Van meter~\textit{et~al.} in~\cite{van_meter_quantum_2022} further proposed a framework for routing of entanglement across many smaller networks, and identified the types of protocols required to do so. A common feature of these architectures and protocols is that they do not aim to provide end nodes with a guaranteed quality of service and they do not develop protocols for preventing arriving requests from disrupting service to existing requests. In contrast, Arqon is designed to provide end nodes with guaranteed quality of service agreements which derive from end node application execution requirements and it incorporates an admission control process that can reject arriving demands to prevent disrupting existing service agreements. 

In~\cite{skrzypczyk_architecture_2021} Skrzypczyk~\textit{et~al.} defined a quantum network architecture for meeting quality of service requirements in multi-user quantum networks. Quality of service was defined with respect to classically inspired metrics (throughput and jitter of entangled link generation) and quantum specific metrics (fidelity of generated entanglement). The architecture employs Time-Division Multiple-Access (TDMA)~\cite{TDMA} scheduling with fixed-duration time slots, granting contention-free access to all qubits along a demand's assigned route for integer numbers of consecutive slots. The authors identified the need for an admission control protocol to prevent arriving demands from disrupting service to existing demands. However, no admission control protocol was defined or implemented, leaving the architecture unable to satisfy reliability requirement~\ref{req: 4: independent}. In contrast, Arqon defines a concrete admission control protocol which rejects demands that cannot be satisfied without disrupting service to accepted demands.

\section{Design Considerations}\label{sec: design considerations}
Our task is to design a suite of control applications, which when integrated into the modular architecture defined in~\cite{NetArx}, create a quantum network satisfying all of the reliability requirements~\ref{req: 1: responses}-\ref{req: 4: independent}. From these requirements we derive the following core considerations which guide the design of the Arqon suite of control applications.
\subsection{Core Considerations}\label{ssec: core design considerations}
\begin{enumerate}[label=\textit{(C\arabic*)}]
    \item \label{consideration: fast} Algorithms that handle time-sensitive network scheduling must execute within bounded time, requiring operational complexity that is polynomial in the number of demands.
    
    \item \label{consideration: service agreements} 
    Service agreements must be created for accepted demands, defining the type and quality of service that Arqon commits to provide. In complement, rejection messages must be sent to rejected demands.
    
    \item \label{consideration: AC prevents overload} A demand may be accepted by Arqon if and only if acceptance will not disrupt any existing service agreements.

    \item \label{consideration: robustness} Arqon must support dynamic updates of the network topology by allowing new nodes to register and removing network components that fail or become unresponsive.
\end{enumerate}

Service agreements will be formally defined in Section~\ref{sec: Arqon Suite of Control Applications}, following necessary preliminaries. 

\subsection{Network Model} \label{sec:NetworkModel} 

\begin{figure*}
\center
\resizebox{0.96\textwidth}{!}{
\begin{tikzpicture}[
            node distance = 1cm,
            onOffLink/.style = {rectangle,
                                draw=orange!60,
                                fill=orange!5,
                                minimum width = 5.5 cm,
                                minimum height = 0.5cm,
                                very thick,
                                align = center},
            processingNode/.style= {circle,
                draw=ProcessBlue!60,
                fill=ProcessBlue!5,
                very thick,
                align = center,
                minimum width = 0.25cm
                },
            metroHub/.style = {rectangle,
                                draw=blue!50!red!80,
                                fill=blue!50!red!10,
                                thick,
                                align = center,
                                inner sep = 0.25cm}
                    ]

            \node[jct] (J1) at (0,0) {$J_1$};
            
            \node[MH] (M1) [above= 0.2cm of J1, xshift=-1.5cm] {$I_1$};
            \node[MH] (M2) [below = 0.2cm of J1, xshift=-1.5cm] {$I_2$};

            \node[PN] (PN11) [xshift=1.5cm, yshift=0.0cm] at (M1) {};
            \node[PN] (PN12) [xshift=-1.5cm, yshift=0cm] at (M1) {};
            \node[PN] (PN13) [xshift=-1.0cm, yshift=1.0cm] at (M1) {};
            \node[PN] (PN14) [xshift=+1.0cm, yshift=1.0cm] at (M1) {};
            \node[PN] (PN15) [xshift=0cm, yshift=1.5cm] at (M1) {};

            \node[PN] (PN21) [xshift=-1.5cm, yshift=1.0cm] at (M2) {};
            \node[PN] (PN22) [xshift=1.5cm, yshift=0.0cm] at (M2) {};
            \node[PN] (PN23) [xshift=-1.5cm, yshift=0.0cm] at (M2) {};
            \node[PN] (PN24) [xshift=1.0cm, yshift=-1.0cm] at (M2) {};
            \node[PN] (PN25) [xshift=-1.0cm, yshift=-1.0cm] at (M2) {};

            \node[onOffLink] (L1) [right = 0.8cm of J1] {$B_1$};
            
            \node[jct] (J2) [right = 0.8cm of L1] {$J_2$};

            \node[MH] (M4) [xshift=1.5cm, yshift=0cm] at (J2) {$I_3$};
            
            \node[PN] (PN41) [xshift=-1.0cm, yshift=1.25cm] at (M4) {};
            \node[PN] (PN42) [xshift=0.0cm, yshift=1.5cm] at (M4) {};
            \node[PN] (PN43) [xshift=1.0cm, yshift=1.25cm] at (M4) {};
            \node[PN] (PN44) [xshift=-1.0cm, yshift=-1.25cm] at (M4) {};
            \node[PN] (PN45) [xshift=0cm, yshift=-1.5cm] at (M4) {};

            \path let \p1 = ($(L1.west)!0.5!(L1.east)$) in coordinate (midpointBB) at (\p1);

            \node[RN] (2) [left = 0.2 cm of midpointBB, yshift=-1cm] {};
            \node[RN] (1) [left = 0.4 cm of 2] {};
            \node[RN] (3) [right = 0.4 cm of 2] {};
            \node[RN] (4) [right = 0.4 cm of 3] {};
            \node[RN] (C) [left=0.4cm of 1] {};
            \node[RN] (D) [right = 0.4cm of 4] {};

            \draw (PN11) -- (M1);
            \draw (PN12) -- (M1);
            \draw (PN13) -- (M1);
            \draw (PN14) -- (M1);
            \draw (PN15) -- (M1);

            \draw (PN21) -- (M2);
            \draw (PN22) -- (M2);
            \draw (PN23) -- (M2);
            \draw (PN24) -- (M2);
            \draw (PN25) -- (M2);

            \draw (PN41) -- (M4);
            \draw (PN42) -- (M4);
            \draw (PN43) -- (M4);
            \draw (PN44) -- (M4);
            \draw (PN45) -- (M4);

            \draw (M1.south east) -- (J1.north west);
            \draw (M2.north east) -- (J1.south west);
            \draw[decorate,decoration=zigzag] (J1.east) -- (L1.west);
            \draw[decorate,decoration=zigzag] (L1.east) -- (J2.west);
            \draw (J2.east) -- (M4.west);

            \draw (C.east) -- (1.west);
            \draw (1.east) -- (2.west);
            \draw (2.east) -- (3.west);
            \draw (3.east) -- (4.west);
            \draw (4.east) -- (D.west);

            \draw[decorate, decoration = {brace,mirror,raise = -22pt, amplitude = 6pt}] (L1.east) -- (L1.west);
            \draw[decorate, decoration = {brace,raise = 10pt, amplitude = 6pt}] (L1.east) -- (L1.west);

            \node[rectangle, fill = purple!20, draw=purple!60] (BBLogic) [above = 0.25cm of L1] {B.B. Logic};

            \node[rectangle, fill = purple!20, draw=purple!60, align=center] (CC) [above= 1.75cm of L1] {Central\\Controller};

            \node[rectangle, dashed, draw=orange!60, thick] (BB encompassing) [yshift = 0cm, minimum width=6cm, minimum height = 3cm] at (midpointBB){};

            \node (PNLegend) [right = 1.75cm of M4, yshift=2.5cm] {};
            \node[PN] [left = 0cm of PNLegend] (PNL) {};
            \node[] [right = 0cm of PNLegend] {End Node};

            \node (MetroLegend) [below = 0.5cm of PNLegend] {};
            \node[rectangle, fill = blue!50!red!10, draw=blue!50!red!80] [left = 0cm of MetroLegend, minimum width=0.3cm, minimum height=0.3cm]  {};
            \node[] [right = 0cm of MetroLegend] (ML) {Metropolitan Hub};

            \node (MHLegend) [below = 0.5cm of MetroLegend] {};
            \node[MH] [left = 0cm of MHLegend, minimum width=0.3cm, minimum height=0.3cm]  {};
            \node[] [right = 0cm of MHLegend] (MHL) {EGI};
            
            \node (JctLegend) [below=0.5cm of MHLegend] {};
            \node[jct] [left = -0.05cm of JctLegend, minimum width=0.2cm, minimum height=0.2cm] {};
            \node[] [right = 0cm of JctLegend] (JctL) {Junction (Border) Node};

            \node (LDBBLegend) [below=0.5cm of JctLegend] {};
            \node[onOffLink] [left = 0cm of LDBBLegend, minimum width=0.3cm, minimum height=0.3cm]  {};
            \node[] [right = 0cm of LDBBLegend] (LDBBL) {Long Distance Backbone};

            \node (RNLegend) [below = 0.5cm of LDBBLegend] {};
            \node[RN] [left = -0.1cm of RNLegend, minimum width=0.1cm, minimum height=0.1cm]  {};
            \node[] [right = 0cm of RNLegend] (RNL) {Repeater Node};

            \node (ControlLegend) [below = 0.5cm of RNLegend] {};
            \node[rectangle, fill = purple!20, draw=purple!60] [left = 0cm of ControlLegend, minimum width=0.3cm, minimum height=0.3cm]  {};
            \node[] [right = 0cm of ControlLegend] (CL) {Control Processor};

            \begin{pgfonlayer}{background}
                \node[metroHub, fit = (M1) (M2)] (MHBox1) {};
                \node at (MHBox1.center) {$H_1$};
                
                \node[metroHub, inner sep = 0.45cm, fit = (M4)] (MHBox2) {};
                \node at ($(MHBox2)+(0.45,-0.63)$) {$H_2$};
                
                \node[draw=Black!60, fill = Black!5, fit = (PNL) (JctL) (MHL) (LDBBL) (RNL) (CL) (ML)] (LegendBox) {};
                \node[above=0cm of LegendBox.north] {\textbf{\underline{Legend}}};
            \end{pgfonlayer}

\end{tikzpicture}
}

\caption{Example quantum network resource graph $\mathcal G = (\mathcal V, \mathcal E)$ for a quantum network where the resources to be shared are a long distance backbone channel $B_1$, two junction nodes $J_{1}, J_{2}$, and three entanglement generation interfaces $I_1, I_2, I_3$ grouped into two metropolitan hubs ($H_1$ and $H_2$).
Edges represent logical connections, that is if $(v_1, v_2)\in\mathcal{E}$, then it is possible to create an optical path between components $v_1$ and $v_2$, possibly via some optical switch. 
}
\label{fig: example network topology}
\end{figure*}
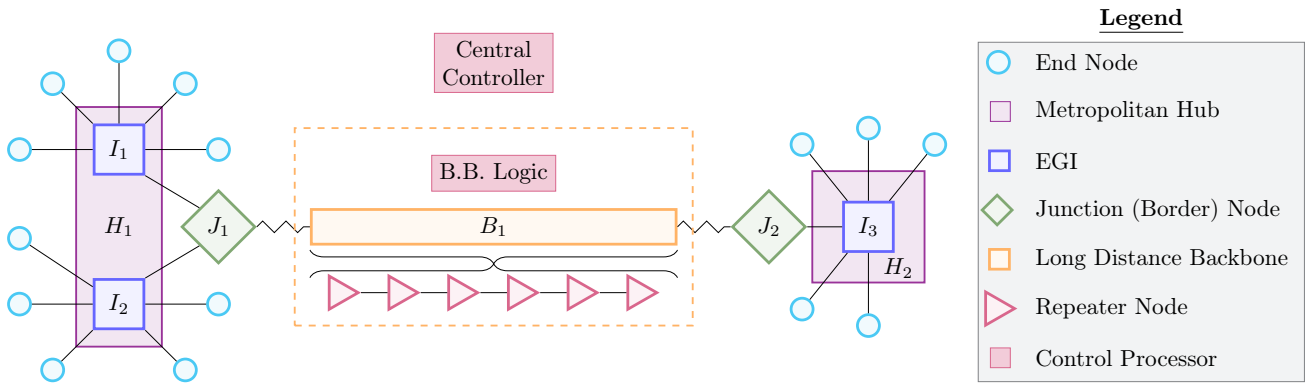

We consider a centrally controlled quantum network where user controlled end nodes submit demands encoding application execution requirements~\cite{NetArx}. Demands are not for single entangled pairs, but rather for repeated creation of packets containing multiple entangled links, with a minimum separation time between packets, until the demand expires. The network can be classified as a \textit{generate-when-requested} network, where end-to-end entangled links are created to serve these demands.  
This means we do not assume that end-to-end entangled links can be stored between any two scheduled periods of time.
This type of network is implementable with the technological maturity of current devices and those that will exist in the coming years. In the following sub-section, we introduce the network components that can be used to construct a general purpose quantum network. 

\subsubsection{Network Components}
\label{subsubsec: Network Components} 
At the physical layer, a general quantum network may consist of a heterogeneous mixture of quantum nodes and the interfaces between them, in an arbitrary topology. 
Figure~\ref{fig: example network topology} represents a possible network topology as a graph, where network components are illustrated as colored nodes and edges connecting components indicate logical connections. 
We consider the following network components:
\begin{itemize}
    \item \textbf{End Nodes} These devices execute programs realizing part of a quantum network application, operate under independent (local) control, and accept input from users.
    An end node may be a \textit{quantum processing node} with some quantum memory capabilities, as in~\cite{ruf_quantum_2021, ThreeNodeQN, HerEntTrappedIons2, HerEntTrappedIons1, NeutralAtoms1, NeutralAtoms2, HerEntOriginalAE, HerEntSecondAE}, or a device without quantum memory which is capable of preparing/measuring photons, as in~\cite{PhotonicClients1, MDIQKDDeployed}. Any end node can also perform classical operations, such as arithmetic operations and classical communication. Resources within the end nodes are managed locally, for example by the operating system QNodeOS~\cite{QNodeOS} or its upgrade Qoala~\cite{Qoala}, and are therefore not under the direct control of any network controller. 
    The local control includes a quantum network agent, which is a process responsible for communicating with a central network controller. 
    This agent submits demands to the central network controller, retrieves and installs network schedules and processes network status messages.

    \item \textbf{Central Controller}
    The purpose of the central controller is to produce network schedules in response to demands from end nodes.
    To do so, the central controller also maintains information about the overall network topology and the entanglement generation capabilities of each component in the network. 
    The central controller hosts a suite of control applications which implement its various functionalities.

    \item \textbf{Metropolitan Hubs} A metropolitan hub facilitates a scalable method of multi-user network access~\cite{RCPpaper, OnDemandPaper}. 
    A metropolitan hub contains a logical control element which enforces the network schedule by controlling the physical devices used to realize the schedule. These devices are: 
    \begin{itemize}
            \item \textbf{Entanglement Generation Interfaces (EGIs)} These devices mediate entangled link generation between neighboring nodes via remote entanglement generation protocols~\cite{RemoteEntProtocolsQubitswithPhotonicInterfaces, Jones_2016}. Different protocols require different EGIs: detection-in-midpoint protocols (e.g., heralded entanglement generation~\cite{theoryHeraldCCGZ, theoryHeraldDLCZ}) use Bell-State Analyzers (BSAs)~\cite{braunstein1995measurement, michler1996interferometric, walther2005experimental}; source-in-midpoint protocols~\cite{Jones_2016} use entangled photon sources; and sender-receiver protocols use only optical channels~\cite{RemoteEntProtocolsQubitswithPhotonicInterfaces, ConditionalPhaseReflection, HarvardConditionalAmplitudeReflection1}.
        Each EGI can only facilitate the generation of entangled links between a single pair of nodes at any one time.
    
        \item \textbf{Switches}
        A switch is a reconfigurable interface that routes signals between different paths.
        Optical switches are compatible with fiber-based or free-space implementations (where light travels through air or vacuum). 
        In both cases, an optical switch enables dynamic, high-speed reconfiguration of signal paths, supporting efficient routing.
        Switches allow the metropolitan hub to allocate pairs of nodes optical access to EGIs according to the network schedule.
    \end{itemize}

    \item \textbf{Long-Distance Backbones} A long-distance backbone is a repeater chain consisting of a linear sequence of \textit{repeater nodes}, such as those in~\cite{theoryHeraldCCGZ, theoryHeraldDLCZ, RepeatersPhotonic1, RepeatersAE1, RepeatersTrappedIon1}. The network treats each backbone as a single component that produces entangled links between its two border nodes~\cite{BorderNodes} at a fixed rate and average fidelity. A long-distance backbone controller configures the repeater chain according to a repeater policy (see e.g.~\cite{repeatersProtocols1, inesta_optimal_2023, repeatersTheory2}). The repeater nodes and border nodes are automated components that strictly implement a repeater policy without supporting additional control logic.

    \item \textbf{Junction Nodes} A junction node is a border node equipped with additional control logic to support network branching. 
    They provide an interface between multiple repeater chains and between repeater chains and a metropolitan hub. They have the physical capabilities to execute remote entanglement generation protocols, store the resulting entangled links, and perform entanglement swaps. The logical controller of a junction node manages local resources (such as quantum memories) to execute the network schedule.
\end{itemize}

\subsubsection{Independently Programmed End Nodes}
Multi-node quantum network applications are partitioned into separate single-node programs that are run
concurrently on different end nodes. The central controller is not aware of the programs on independent end nodes. Each program contains instructions for local classical and quantum computations, and programs on separate end nodes interact with each other only through classical messaging and entanglement generation~\cite{Qoala}. 

This hybrid classical-quantum programming model has an important implication for resource management: programs can be structured into sub-routines such that quantum states required by one sub-routine need not persist through following sub-routines. Individual quantum states therefore need not be stored for the entire program duration.

\subsection{Physical Considerations}

Generation of end-to-end entanglement requires implementing some intermediary protocols, in particular performing many physical attempts to generate entanglement between neighboring network components. 
Each of these physical attempts only succeeds probabilistically, resulting in some average rate of successfully generating end-to-end entangled links.
Successfully generated entangled links are stored in quantum memories, which are subject to time-dependent errors due to environmental interactions~\cite{DecoherenceZurek, DecoherenceUnruh}. These decoherence processes limit the storage time, characterized by the memory coherence time~\cite{QuantumMemory1, QuantumMemory2, QuantumMemory3}. To execute a sub-routine of an application program that requires a packet of entangled links, all links must remain sufficiently coherent throughout the duration of the sub-routine. Consequently, all links in a packet must be generated within a time window $w$ determined by both the quantum memory coherence times and the sub-routine execution duration.
The probability that a batch of attempts to generate end-to-end entangled links succeeds and the packet generation window $w$ are not simply modeling parameters, they are inherited from the physical capabilities of the network components.

\subsection{Timescales}

The time required for network elements to complete actions can only be estimated with finite precision. At the physical layer, actions have precisely characterized durations, allowing for accurate synchronization between multiple nodes, with timing precision ranging from tens of picoseconds (ps) to microseconds ($\mu$s) to milliseconds (ms), depending on the operations~\cite{ThreeNodeQN, QNodeOS}. Precise timing of a sequence of operations is crucial for processes like entanglement generation~\cite{DoubleClickDiamond, humphreys_deterministic_2018, HerEntTrappedIons1, HerEntTrappedIons2, HerEntOriginalAE, HerEntSecondAE, NeutralAtomsHeralded}.

In contrast, at higher layers of the network stack, actions have variable durations and latencies, limiting feasible timing precision to $\mu$s or ms~\cite{OSLatency, ModernOS, HardRealTimeComputing}. In particular, the time required to send a message over the internet varies, and a difference of several ms may occur in the communication times of repeated messages between a fixed sender and receiver. 

An additional type of timing consideration is that entanglement generation involves sequential non-overlapping attempts, each with a particular sequence of operations, setting a minimum period and maximum rate for the process. This also means operations cannot change mid-attempt without disrupting entanglement generation.
The large differences in timescales between the physical layer and higher layers of the network stack imply that scheduling and triggering of actions at the physical layer is not a function of any higher layer of the network stack, but is the responsibility of dedicated and sophisticated electronics that are key components of the physical layer.

\subsection{Architecture}\label{subsec:Architecture}

\begin{figure*}
    \center
\resizebox{0.90\textwidth}{!}{\begin{tikzpicture}
    [
    			node distance = 1cm,
                background rectangle/.style={fill=white},
                endNode/.style={rectangle,
                                draw=blue!60, 
                                fill=blue!20, 
                                very thick, 
                                align=center, 
                                minimum width=8.5cm, 
                                minimum height = 6.5cm,
                                rounded corners,
                                },
                Component/.style={rectangle,
                                draw=red!60, 
                                fill=red!20, 
                                very thick, 
                                align=center, 
                                minimum width=3cm, 
                                minimum height = 1cm,
                                rounded corners,
                                },
                PU/.style={rectangle,
                                draw=red!60, 
                                fill=red!20, 
                                very thick, 
                                align=center, 
                                minimum width=1cm, 
                                minimum height = 1cm,
                                rounded corners,
                                },
                QNtwkCmp/.style={rectangle,
                                draw=Fuchsia!60, 
                                fill=Fuchsia!20, 
                                very thick, 
                                align=center, 
                                minimum width=3.25cm, 
                                minimum height = 0.5cm,
                                },
                CNtwkCmp/.style={rectangle,
                                draw=BlueGreen!60, 
                                fill=BlueGreen!20, 
                                very thick, 
                                align=center, 
                                minimum width=2.9cm,  
                                minimum height = 0.5cm,
                                },
                CC/.style={rectangle,
                                draw=orange!60, 
                                fill=orange!20, 
                                very thick, 
                                align=center, 
                                minimum width=20cm, 
                                minimum height = 5cm,
                                rounded corners,
                                },
                SA/.style={thick, ->,},
                DA/.style={thick, <->,},
    			]
    
    \draw[thick] (-6,2.8) --  (-6, -0.5) node [below] {$k-3$} ;
    \draw[thick] (-3,2.8) --  (-3, -0.5) node [below] {$k-2$} ;
    \draw[thick] ( 0,2.8) --  ( 0, -0.5) node [below] {$k-1$} ;
    \draw[thick] ( 3,2.8) --  ( 3, -0.5) node [below] {$k$};
    \draw[thick] ( 6,2.8) --  ( 6, -0.5) node [below] {$k+1$};
    \draw[->, thick] (-6.5,-0.4) -- (6.8,-0.4); 
    
    \newcommand{\actionstack}[7]{
        \coordinate (root) at #3;
        \ifnum#1=1
    
            \node[CNtwkCmp] (execute schedule) [left=0.025cm of root, draw=Blue!60, preaction={fill, Blue!20}, pattern=#2, pattern color=Blue!60] {};
            \node[CNtwkCmp] (distribute schedule) [left=0.05cm of execute schedule, yshift=0.6cm, xshift = 0cm, minimum width = 2.9cm, draw=Red!60, preaction={fill, Red!20}, pattern=#2, pattern color=Red!60] {};
            \node[CNtwkCmp] (compute schedule) [left=1.39cm of distribute schedule, xshift=0cm, yshift=0.6cm, minimum width = 0.75cm, preaction={fill, Green!20}, draw=Green!60, pattern=#2, pattern color=ForestGreen!60] {};
            
            \node[CNtwkCmp] (admit demands) [left=-0.cm of compute schedule, yshift=0.6cm, minimum width = 0.78cm, preaction={fill, Green!20}, draw=Green!60, pattern=#2, pattern color=ForestGreen!60] {};
            
            \node[CNtwkCmp] (format demands) [left=9.025cm of root, xshift=0cm, yshift=2.4cm, minimum width = 2.9cm, preaction={fill, Fuchsia!20}, draw=Fuchsia!60, pattern=#2, pattern color=Fuchsia!60] {};
            
        \else
            \ifnum#4=1
                \node[CNtwkCmp] (execute schedule) [draw = black!20, preaction={fill, Black!5}, pattern=#2, pattern color=Black!20] [left = 0.025cm of root] {};
            \fi
            \ifnum#5=1
                \node[CNtwkCmp] (distribute schedule) [left=3.025cm of root, yshift=0.6cm, xshift = 0cm, minimum width = 2.9cm, draw=Black!20, preaction={fill, Black!5}, pattern=#2, pattern color=Black!20] {};
            \fi
            \ifnum#6=1
                \node[CNtwkCmp] (compute schedule) [left=7.35cm of root, xshift=0cm, yshift=1.2cm, minimum width = 0.74cm, preaction={fill, Black!5}, draw=Black!20, pattern=#2, pattern color=Black!60] {};
                \node[CNtwkCmp] (admit demands) [left=8.14cm of root, yshift=1.8cm, minimum width = 0.78cm, preaction={fill, Black!5}, draw=Black!20, pattern=#2, pattern color=Black!20] {};
            \fi
            \ifnum#7=1
                \node[CNtwkCmp] (format demands) [left=9.025cm of root, xshift=0cm, yshift=2.4cm, minimum width = 2.9cm, preaction={fill, Black!5}, draw=Black!20, pattern=#2, pattern color=Black!20] {};
            \fi
        \fi
    }
    
    \actionstack{0}{none}{(15,0)}0001
    \actionstack{0}{none}{(12,0)}0011
    \actionstack{0}{none}{(9,0)}0111
    \actionstack{1}{crosshatch}{(6,0)}1111
    \actionstack{0}{none}{(3,0)}1110
    \actionstack{0}{none}{(0,0)}1100
    \actionstack{0}{none}{(-3,0)}1000

    
    \coordinate (label root) at (-6,-0.4);
    
    \node[align=right] at ($(label root|-admit demands)$) [left] {Admit Tasks};
    \node[align=right] at ($(label root|-format demands)$) [left] {Demand Registration};
    \node[align=right] at ($(label root|-compute schedule)$) [left]  {Compute Schedule};
    \node[align=right] at ($(label root|-distribute schedule)$) [left]  {Distribute Schedule};
    \node[align=right] at ($(label root|-execute schedule)$) [left]  {Execute Schedule};
    
    \node[align=center] at (0.0,-1) [below] {Scheduling Interval};
        
    \end{tikzpicture}}

    \caption[Modular Timing Overview]{Process sequence resulting in  execution of a network schedule in scheduling interval $k$. The hatched, colorful processes are executed in the indicated scheduling intervals in order to produce and distribute the schedule which will be executed in scheduling interval $k$. The grey processes highlight the periodic nature of our architecture.}
    \label{fig:network scheduling timings}

\end{figure*}
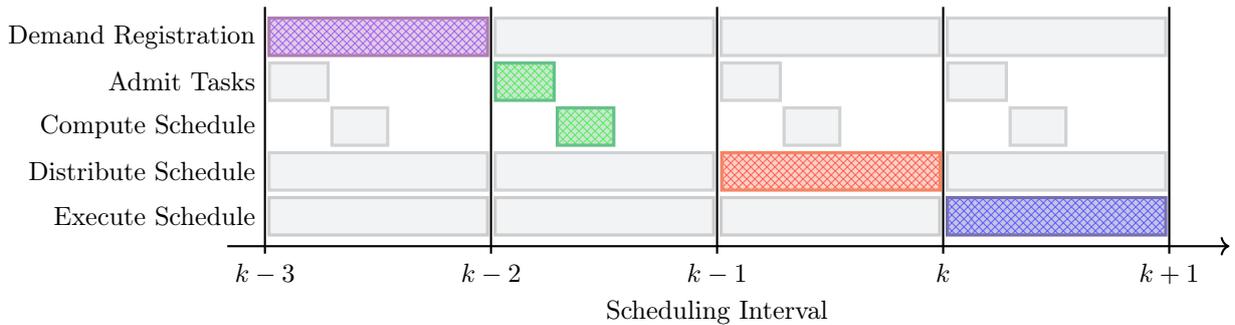

We state key features of our network architecture~\cite{NetArx}, as we make frequent reference to these features when defining the Arqon control applications. 

For simplicity, we assume quantum network applications are executed between two end nodes. However, our architecture is directly compatible with applications involving multiple end nodes.

\subsubsection{Application Sessions}\label{sssec: DC: Arx: Application Sessions}

The aim of users of a quantum network is to successfully run applications. To do so, they execute hybrid classical-quantum programs on independent end nodes. 
As many hybrid programs have probabilistic outcomes, it is often required to execute the same program many times to extract a useful and reliable output.
Realizing a single execution of a quantum network application requires each collaborating end node to successfully execute their local program.
We refer to each of these individual executions of an application as an \textit{application instance}.
In practice, applications are typically associated with a deadline, before which all application instances must be executed. 
To capture these requirements, we define an \textit{application session}.

\begin{definition}[Application Session]
Suppose end nodes $\mathcal{N}=(\emph{\texttt{node1}}, \emph{\texttt{node2}},...)$ wish to execute application \emph{\texttt{App}} at least $N^{\texttt{\emph{inst}}}$ times, before time $t^{\texttt{\emph{expiry}}}$.

Then we write the corresponding \emph{\textbf{application session}} as 
\begin{equation}
    \mathcal S = (\text{\normalfont\texttt{session\_id}}, \mathcal{N}, \text{\normalfont{\texttt{APP}}},N^{\texttt{\emph{inst}}}, t^{\texttt{\emph{expiry}}}) \label{eq: application session}, 
\end{equation} where \normalfont{\texttt{session\_id}} is a unique identifier for the session.
\end{definition}

To ensure feasibility, it is necessary that the specified expiry time, $\texpiry$, is longer than any time-scale of the network.
To establish a common understanding of how to locally configure programs for executing an application session~\eqref{eq: application session} and ensure appropriate resource allocation, the end nodes in~$\mathcal{N}
$ execute two processes: a \textit{network capability update} and \textit{capability negotiation}~\cite{NetArx}. In the network capability update, each node independently queries a network capability table obtained from the central controller to retrieve end-to-end entanglement generation capabilities for each suitable path, including the average entanglement generation rate and the corresponding minimum fidelity of entangled pairs. In capability negotiation, the end nodes in $\mathcal{N}$ communicate to coordinate their local resource requirements for executing application programs, using the acquired network capabilities as input.

\subsubsection{Application Demands}\label{sssec: DC: Arx: Demands}
Based on their requirements, application sessions on end nodes submit demands for packets of entanglement generation to the network. 
A distinct group of end nodes may simultaneously host multiple application sessions, hence the group may be associated with multiple demands at once.
\begin{definition}[Demand]
\label{def: demand}
Application sessions register \emph{\textbf{demands}}. These demands have the format \begin{equation}\label{eq: demand raw}
    d^{\text{\emph{raw}}} = (\mathfrak{p}; \tminsep; t^{\texttt{\emph{expiry}}};N^{\texttt{\emph{inst}}}) 
\end{equation} where $\mathfrak{p} = (w,s,F)$ defines a packet of $s$ pairs to be generated within time window $w$, with minimum fidelity $F$; $\tminsep$ is a minimum time separation between attempts to generate packets, $t^{\texttt{\emph{expiry}}}$ is the expiry time of the application session, $N^{\texttt{\emph{inst}}}$ is the number of successes required to satisfy the demand. 
A demand is associated with implementation specific metadata, denoted $M_d$. The full demand can then be represented as \begin{equation}\label{eq: demand full}
    d = (d^{\text{\emph{raw}}}, M_d).
\end{equation} 
\end{definition}

The metadata accompanying a demand is implementation specific and should include source and destination identifiers for the end nodes. 
Analogously to the metadata which accompanies a data packet in the classical internet, metadata can be used to inform control processes, including routing.

In our original formulation~\cite{NetArx}, the raw demand included a parameter $R_{\texttt{packet}}$ specifying the average requested rate of successful packet generation until the expiry time $\texpiry$. We have removed $R_{\texttt{packet}}$ from~\eqref{eq: demand raw}, as the number of instances $N^{\texttt{inst}}$ and expiry time $t^{\texttt{expiry}}$ together specify an average rate of packet generation. 

As a new addition to the model from~\cite{NetArx}, we introduce a service error parameter, $\epsilon^{\text{service}}_d$, associated with a demand $d$. This parameter allows precise specification of the reliability requirement~\ref{req: 2: satisfied} for each demand $d$,\begin{equation}\label{eq: prob satisfied epsilon service}
    \mathbb{P}[\text{demand d is satisfied}] \geq (1 - \epsilon^{\text{service}}_d).
\end{equation}
In a deployment of Arqon, the service error parameter may be chosen according to service contracts (i.e. contracts between the network operator and its clients) or priority classes associated with the demand. 
In some implementations, it may be set by the central controller. 

\subsubsection{Periodic Computation and Distribution of Schedules}\label{sssec: DC: Arx: Periodic Schedules}
The central controller must regularly accept demands, as well as compute and distribute network schedules.
Each schedule is associated with a version identifier and covers an identical execution time known as the \textit{scheduling interval} (denoted $T^{SI}$). Figure~\ref{fig:network scheduling timings} is an illustrated overview of the process sequence that enables execution of a network schedule for a scheduling interval $k$. 

In our network architecture there is a constant \textit{scheduling offset}, in terms of a number of scheduling intervals, between the interval at which a demand passes demand registration and the time at which the first schedule involving the demand will be executed if accepted.
In Figure \ref{fig:network scheduling timings}, a scheduling offset of two scheduling intervals is illustrated.
To account for this offset, when demands are registered their earliest possible start time $t^{\texttt{start}}$ is calculated and added to the internal representation of the demand.

\subsubsection{Packet Generation Attempts}\label{sssec: DC: Arx: PGA}
Network schedules allocate \textit{packet generation attempts}~(PGAs) to accepted demands. A PGA always has some duration $E$, success probability $\ppacket$, and reserves internal resources on an end-to-end route through the network. 
PGAs are associated with a success probability due to the probabilistic nature of generating entangled links.
It is impossible to guarantee that a packet will be generated when any finite execution time is allocated to a PGA.

A minimum separation time between successful packet generation, $\tminsep$, is specified as part of a demand. In the PGA framework, this is enforced as a minimum separation time between subsequently scheduled PGAs for the same demand. 
This minimum separation is included to ensure that there is sufficient time for the local runtime environment on end nodes to execute local operations before the next allocated period of time for generating entangled links begins. Examples of these local operations are additional blocks of quantum operations in the application program, or operations to reset the hardware between subsequent attempts to generate a packet.

\subsubsection{Packet Generation Tasks}\label{sssec: DC: Arx: PGT}
\textit{Packet Generation Tasks}~(PGTs) are the internal demand representation format in our network architecture~\cite{NetArx}.
They are persistent tasks requiring the central controller to schedule PGAs to serve an accepted demand until it expires or is terminated. 
To create a PGT $\gamma$ from a demand it is necessary to determine a suitable duration $E_\gamma$ and packet success probability $\ppacket_{\gamma}$ for the PGAs. 
These key parameters of a PGA must be set based on the end-to-end entanglement generation capabilities along a particular path.
For this reason, a PGT $\gamma$ is created for a particular end-to-end entanglement generation path $\pi_{\gamma}$.
We formally define PGTs in Section~\ref{ssec: Demand Registration}.

\section{Arqon Suite of Control Applications}\label{sec: Arqon Suite of Control Applications}
Arqon consists of four control applications:  Network Manager, a Demand Manager, a Network Scheduler, and a Schedule Manager. These control applications are hosted by a central controller, such as an implementation of an Software Defined Network (SDN) controller \cite{kreutz_software-defined_2015, halpern_forwarding_2010, mckeown_openflow_2008,  kozlowski_p4_2020}. The central controller manages all communication with other network components. Figure \ref{fig: systems diagram} is a systems interaction diagram that provides an overview of the control applications and their shared interfaces.


\begin{figure*}
    \centering
    \resizebox{0.88\textwidth}{!}{
\begin{tikzpicture}[x=1cm, y=1cm]
\pgfdeclarelayer{background}
\pgfdeclarelayer{background2}
\pgfdeclarelayer{background3}
\pgfsetlayers{background3, background2, background,main}

    \node[component, minimum height = 5.25cm, minimum width = 15cm] (Arqon) at (0,0) {};
    \node [below =0.1cm of Arqon.north, align=center] {
    \normalsize \itshape \textbf{Arqon}
    };

    \coordinate (ArqonLine) at ($(Arqon.north)!0.22!(Arqon.center)$);

    \draw ($(Arqon.west |- ArqonLine)$) -- (ArqonLine) -- ($(Arqon.east |- ArqonLine)$);
    
    \node[component, minimum height=4.0cm] (DM) [right = 0.3cm of Arqon.south west, anchor=south west, yshift=0.3cm] {\small\textsc{«Component»}\\ \normalsize Demand Manager};
    \node[component] [right = 2cm of DM.south east, anchor=south west] (NM) {\small\textsc{«Component»}\\ \normalsize Network Manager};
    \node[component] [left = 0.3cm of Arqon.south east, anchor=south east, yshift=0.3cm] (SM) {\small\textsc{«Component»}\\ \normalsize Schedule Manager};
    \node[component] [anchor = north] at ($(DM.north -| SM)$) (NS) {\small\textsc{«Component»}\\ \normalsize Network Scheduler};

    \node[component, minimum width = 15cm, below = 1.2cm of Arqon] (SDN) {Central Controller};


    \coordinate (NS1) at ($(NS.south west)!0.25!(NS.north west)$);
    \coordinate (NS2) at ($(NS.north west)!0.25!(NS.south west)$);

    \coordinate (NM1) at ($(DM.east |- NM.west)$);
    \draw[-{Circle[open, scale=2]}] (NM.west) -- ($(NM.west)!0.35!(NM1)$) node [below left, pos=0.9] {\tiny \scshape \textbf{NM:P}};
    \draw[-{Arc Barb[reversed, scale=2.2]}] (NM1) -- ($(NM.west)!0.3!(NM1)$);

    \coordinate (NM2) at ($(NM.west)!0.5!(NM1)$);
    
    \draw (NM2) |- ($(NM2 |- NS1)$) -- (NS1);

    \coordinate (DM1) at ($(DM.east |- NS2)$);
    \coordinate (DM2) at ($(DM.east |- NS.west)$);
    
    \draw[-{Circle[open, scale=2]}] (DM1) -- ($(DM1)!0.2!(NS2)$) node [above left, pos=0.9] {\tiny \scshape \textbf{DM:B}};
    \draw[-{Circle[open, scale=2]}] (DM2) -- ($(DM2)!0.3!(NS.west)$) node [below left, pos=0.9] {\tiny \scshape \textbf{DM:DS}};

    \draw[-{Arc Barb[reversed, scale=2.2]}] (NS2) -- ($(DM1)!0.185!(NS2)$);
    \draw[-{Arc Barb[reversed, scale=2.2]}] (NS.west) -- ($(DM2)!0.285!(NS.west)$);

    \draw[-{Arc Barb[reversed, scale=2.2]}] (NS.south) -- ($(NS.south)!0.5!(SM.north)$) node [below left, pos=1.0] {\tiny \scshape \textbf{SM:SD}};
    \draw[-{Circle[open, scale=2]}] (SM.north) -- ($(NS.south)!0.43!(SM.north)$);

    \coordinate (external) at ($(Arqon.south)!0.5!(SDN.north)$);

    \coordinate (NM3) at ($(NM.south)!0.5!(NM.south west)$);
    \coordinate (SDN1) at ($(SDN.north -| NM3)$);
    \draw[-{Circle[open, scale=2]}] (NM3) -- ($(NM3)!0.55!(SDN1)$) node [above left, pos=0.9] {\tiny \scshape \textbf{NM:N}};
    \draw[-{Arc Barb[reversed, scale=2.2]}] (SDN1) -- ($(NM3)!0.475!(SDN1)$);

    \coordinate (NM4) at ($(NM.south)!0.5!(NM.south east)$);
    \coordinate (SDN2) at ($(SDN.north -| NM4)$);
    \draw[-{Circle[open, scale=2]}] (NM4) -- ($(NM4)!0.55!(SDN2)$) node [above left, pos=0.9] {\tiny \scshape \textbf{NM:NC}};;
    \draw[-{Arc Barb[reversed, scale=2.2]}] (SDN2) -- ($(NM4)!0.475!(SDN2)$);

    \coordinate (SM1) at ($(SM.south)!0.5!(SM.south west)$);
    \coordinate (SDN3) at ($(SDN.north -| SM1)$);
    \draw[-{Circle[open, scale=2]}] (SM1) -- ($(SM1)!0.55!(SDN3)$)  node [above left, pos=0.9] {\tiny \scshape \textbf{SM:S}};;
    \draw[-{Arc Barb[reversed, scale=2.2]}] (SDN3) -- ($(SM1)!0.475!(SDN3)$);

    \coordinate (DM3) at ($(DM.south)!0.5!(DM.south east)$);
    \coordinate (SDN4) at ($(SDN.north -| DM3)$);
    \draw[-{Circle[open, scale=2]}] (DM3) -- ($(DM3)!0.55!(SDN4)$) node [above left, pos=0.9] {\tiny \scshape \textbf{DM:D}};
    \draw[-{Arc Barb[reversed, scale=2.2]}] (SDN4) -- ($(DM3)!0.475!(SDN4)$);

    \coordinate (DM4) at ($(DM.south)!0.5!(DM.south west)$);
    \coordinate (SDN5) at ($(SDN.north -| DM4)$);
    \draw[-{Arc Barb[reversed, scale=2.2]}] (DM4) -- ($(DM4)!0.7!(SDN5)$) node [below left, pos=0.95] {\tiny \scshape \textbf{X}};
    \draw[-{Circle[open, scale=2]}] (SDN5) -- ($(DM4)!0.625!(SDN5)$);

    \coordinate (SM2) at ($(SM.south)!0.5!(SM.south east)$);
    \coordinate (SDN6) at ($(SDN.north -| SM2)$);
    \draw[-{Arc Barb[reversed, scale=2.2]}] (SM2) -- ($(SM2)!0.7!(SDN6)$) node [below right, pos=0.95] {\tiny \scshape \textbf{Z}};
    \draw[-{Circle[open, scale=2]}] (SDN6) -- ($(SM2)!0.625!(SDN6)$);

    \coordinate (NM5) at ($(NM.east)!0.5!(SM.west)$);
    \coordinate (SDN7) at ($(SDN.north -| NM5)$);
    \draw[-{Arc Barb[reversed, scale=2.2]}] (NM.east) -- (NM5) -- ($(NM5)!0.8!(SDN7)$) node [below right, pos=0.95] {\tiny \scshape \textbf{Y}};
    \draw[-{Circle[open, scale=2]}] (SDN7) -- ($(NM5)!0.75!(SDN7)$);


    \node[rectangle, fill=blue!10, draw=black] at (NS1) {};
    \node[rectangle, fill=blue!10, draw=black] at (NS2) {};
    \node[rectangle, fill=blue!10, draw=black] at (NS.west) {};
    
    \node[rectangle, fill=blue!10, draw=black] at (DM1) {};
    \node[rectangle, fill=blue!10, draw=black] at (DM2) {};
    \node[rectangle, fill=blue!10, draw=black] at (DM3) {};
    \node[rectangle, fill=blue!10, draw=black] at (DM4) {};
    
    \node[rectangle, fill=blue!10, draw=black] at (NM1) {};
    \node[rectangle, fill=blue!10, draw=black] at (NM.west) {};
    \node[rectangle, fill=blue!10, draw=black] at (NM3) {};
    \node[rectangle, fill=blue!10, draw=black] at (NM4) {};
    \node[rectangle, fill=blue!10, draw=black] at (NM.east) {};

    \node[rectangle, fill=blue!10, draw=black] at (SM1) {};
    \node[rectangle, fill=blue!10, draw=black] at (SM2) {};
    \node[rectangle, fill=blue!10, draw=black] at (SM.north) {};

    \node[rectangle, fill=blue!10, draw=black] at (NS.south) {};

    \node[rectangle, fill=blue!10, draw=black] at (SDN1) {};
    \node[rectangle, fill=blue!10, draw=black] at (SDN2) {};
    \node[rectangle, fill=blue!10, draw=black] at (SDN3) {};
    \node[rectangle, fill=blue!10, draw=black] at (SDN4) {};
    \node[rectangle, fill=blue!10, draw=black] at (SDN5) {};
    \node[rectangle, fill=blue!10, draw=black] at (SDN6) {};
    \node[rectangle, fill=blue!10, draw=black] at (SDN7) {};

\begin{pgfonlayer}{background}

\end{pgfonlayer}
\end{tikzpicture}
}
    \caption{System Component diagram for Arqon. The interfaces CC:II are described in the main text. The interface \textsc{\textbf{X}} allows Arqon to update the status of demands, interface \textsc{\textbf{Y}} allows Arqon to retrieve the current capabilities of network nodes and \textsc{\textbf{Z}} allows Arqon to push network schedules to components.}
    \label{fig: systems diagram}
\end{figure*}
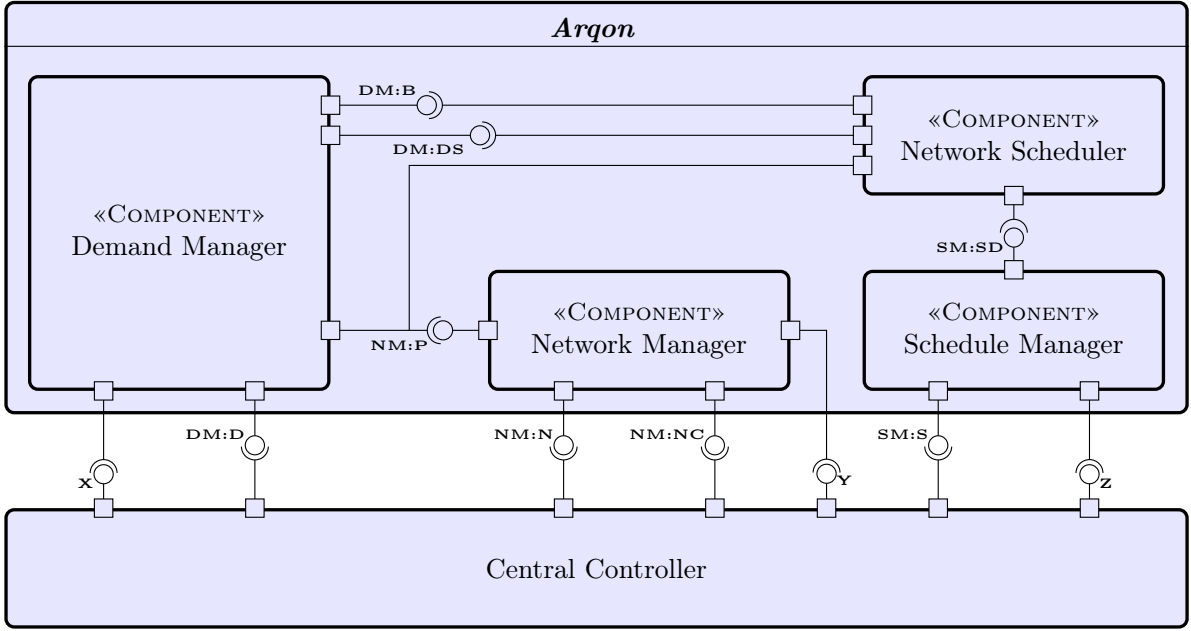

To allow the clearest possible description of the Arqon control applications, we first define the service agreements that Arqon creates when a demand is accepted for service, which is accompanied by delivery of an \texttt{ACCEPT} message to the end node(s) that submitted the demand. Service agreements may only be updated or canceled in the case of network component failure(s) disrupting the compatible end-to-end entanglement generation routes through the network. 

\begin{definition}[Service Agreement]\label{def: ServiceAgreement} 
  The \textit{service agreement} created for an accepted demand is a guarantee from Arqon that:
  \begin{center}
\textbf{    ``In every scheduling interval between $t^{\texttt{\emph{start}}}$ and the demand being terminated or expiring, at least a number of PGAs called a minimal allocation will be scheduled for the PGT accepted to realize the demand.''  }
  \end{center}
\end{definition}

The minimum allocation is a number $\minimumallocation_{\gamma_d}$ of PGAs that must be scheduled for the PGT $\gamma_d$ in each scheduling interval between $t^{\texttt{start}}$ and its expiry time, in order to ensure the demand will be satisfied with at least probability $(1-\epsilon^{\text{service}}_d)$.

\begin{definition}[Minimal allocation]\label{def: minAlloc}
  Let $\gamma_d$ be a PGT for demand $d$, let $t^\texttt{\emph{start}}$ be the start time of the first scheduling interval in which demand $d$ may be scheduled, and let ${n^{SI} = \lceil(t^{\texttt{\emph{expiry}}}_d - t^{\texttt{\emph{start}}}_d)/T^{SI}\rceil}$. 
  Let $\minimumallocation_{\gamma_d}$ be such that if \begin{equation}
      X\sim\mathrm{Binomial}(\minimumallocation_{\gamma_d} n^{SI}, p^{\texttt{\emph{packet}}}_{\gamma_d}),
  \end{equation} then \begin{equation}
      \mathbb{P}[X<N^{\texttt{\emph{inst}}}_d] < \epsilon^{\text{\emph{service}}}_d.
  \end{equation}
  Let $\#_d(S)$ be the number of PGAs for demand $d$ scheduled in schedule $S$. 
  Then we say a demand receives \emph{\textbf{minimal allocation}} in schedule $S$ if ${\#_{d}(S) >= \minimumallocation_{\gamma_d}}$. 
\end{definition}

In the rest of this section we describe each control application in turn, beginning with the Network Manager. 

\subsection{Network Manager}
The Network Manager is responsible for maintaining an accurate overview of all network components, their relevant capabilities, and current statuses.
Its most basic function is to compose and maintain a graph $\mathcal G = (\mathcal V,\mathcal E)$, called the \textit{network resource graph}. The vertices $\mathcal{V}$ are network components, and the logical connections between components are indicated by the set of edges $\mathcal{E}$. 
The Network Manager sets the duration of a scheduling interval, which is a global parameter that must be communicated to the other control applications.
The other core processes of the Network Manager are \textit{Node Registration}, \textit{Path Computation} and a \textit{Network Capabilities Manager} (NCM). 

The Network Manager provides two interfaces to the central controller, known as the \textit{node} (NM:N) interface and 
the \textit{network capabilities} (NM:NC) interface, and an internal interface to the other control applications, known as the \textit{path} (NM:P) interface.
The node interface is used by nodes (re-)joining the network to register themselves. 
The network capabilities interface allows end nodes to obtain information about the entanglement generation capabilities of the network.
The path interface allows other Arqon control applications to access information about the paths in the network either by requesting a data structure called the path partition, denoted $\Pi$, or by requesting sets of paths between specific pairs of nodes.

\subsubsection{Network Resource Graph}\label{sssec: Arqon: Network resource graph}

The network resource graph, $\mathcal{G} = (\mathcal{V}, \mathcal{E})$ is an undirected graph in which all the schedulable components in the network are vertices.  
The set of vertices can be decomposed as $\mathcal V = E\sqcup\mathfrak{R}$ where $E$ is the set of end nodes and $\mathfrak{R}$ is the set of schedulable internal resources: EGIs ($I$), junction nodes ($J$) and long-distance backbones ($B$). 
Each long-distance backbone is allocated as a single unit and is therefor represented by a single vertex in the network resource graph.
A network resource graph for a dumbbell network connecting two metropolitan areas across a long-distance backbone is illustrated in Figure \ref{fig: example network topology}.

The edges of the network resource graph $\mathcal{G}$ represent the logical connections between components.
For example, if a metropolitan hub contains two EGIs which can be accessed by end node $e$, $\mathcal{G}$ has an edge between $e$ and each of these EGIs.
In practice, such a setup could be realized with a single physical connection from the end node to an optical switch located at the metropolitan hub. The optical switch could then be configured to open a physical path from the end node to either of the specific EGIs. 

A path ${\pi= (v_0, v_1,...,v_k)}$ through the network resource graph $\mathcal{G}$ is a sequence of vertices ${v_i \in \mathcal{V}}$ where for each adjacent pair ${(v_i, v_{i+1})}$ there is an edge $e \in \mathcal{E}$ connecting $v_i$ and $v_{i+1}$, and no vertex is visited more than once.
It is only possible to serve demands for end-to-end entanglement generation if there is at least one path through the internal resources connecting two end nodes. If such a path exists between all pairs of end nodes, the network resource graph is internally connected. 

\begin{restatable}[Internally Connected]{definition}{defInternallyConnected}\label{def: internally connected} 
    We say that a network resource graph ${\mathcal{G} = (\mathcal V, \mathcal E)}$ with end nodes $E$ is \emph{\textbf{internally connected}} if for all pairs of nodes ${v,v'\in\mathcal V}$, there exists a path ${\pi = (v, \pi_1, ...,\pi_{k-1}, v')}$ in $\mathcal{G}$ such that ${\forall i, \pi_i \notin E}$.
\end{restatable}

End-to-end paths must not pass through end nodes, because these are under local control only and are not bound to execute network schedules.

\subsubsection{Node Registration}
When a node joins the network for the first time, or re-joins the network after a period of inactivity, it needs to register itself with the central controller. 
This is handled by the Node Registration process of the Network Manager via the node (NM:N) interface with the central controller. 
When nodes register with the Network Manager it assigns a unique node ID that Arqon uses in the creation and distribution of network schedules, and it triggers updates to the network resource graph and the network capabilities table maintained by the NCM. 
A registration response is returned to the node with the assigned node ID. 

\subsubsection{Path Computation}\label{sssec: arqon: NM: path computation}
Path Computation identifies which paths through the network resource graph are suitable for generating end-to-end entangled links between any two end nodes. It is not a time-critical process and is only triggered by updates to the network resource graph. 
As end nodes are under local control only, valid end-to-end entanglement generation paths do not traverse end nodes.
\begin{definition}[Valid entanglement generation paths]
    Let $\mathcal{G}$ be a network resource graph with end nodes $E$. Let ${\pi=(\pi_0, ..., \pi_k)}$ be a path in $\mathcal{G}$. Then we say that $\pi$ is a \emph{\textbf{valid entanglement generation path}} if: \begin{enumerate}
        \item $\pi_0\neq\pi_k$
        \item $\pi_0, \pi_k\in E$
        \item $\forall i = 1,...,k-1; \pi_{i}\notin E$.
    \end{enumerate} We denote the set of all valid entanglement generation paths by \normalfont $\mathcal{P}_{\text{valid}}$.
    \label{def:valid ent gen paths}
\end{definition}

The number of valid entanglement generation paths through a network with resource graph $\mathcal{G}$ is ${O(2^{|\mathcal{V}|})}$. For large networks, this can quickly grow to a size which is infeasible to efficiently compute. 
Therefore, implementations of Arqon may benefit from reducing the set of valid entanglement generation paths to a set of \textit{allowable paths} by imposing additional restrictions on the paths which may be used to realize a demand. 

The outcome of Path Computation is a data structure called the path partition, $\Pi$, which is a disjoint partition of all valid or allowable entanglement generation paths. In Section~\ref{sec: Implementation} we define a specific path partition, but Arqon is compatible with any disjoint partition of $\mathcal{P}_{\text{valid}}$. Looking forward, a disjoint path partition is useful in enabling efficient scheduling of network resources, as non-overlapping paths can be simultaneously scheduled to produce entanglement serving different demands without any ill effects.

\subsubsection{Network Capabilities Manager}\label{sssec: NCM}
The NCM maintains a \textit{network capabilities table} containing the end-to-end entanglement generation capabilities of each path ${\pi \in\Pi}$. The network capabilities table has a version identifier that is incremented with each update.

The demands issued by end nodes (see \eqref{eq: demand raw} and \eqref{eq: demand full}) include a minimum fidelity $F$, a number of instances $\ninst$, and an expiry time $\texpiry$, all of which are not arbitrary, but should be informed by the entanglement generation capabilities along paths ${\pi \in \Pi}$. 
To obtain this information, end nodes query the central controller for their entries of the network capabilities table~\cite{NetArx}. 
The metadata that is included with a demand includes the network capability version id, which is used by the Demand Manager to screen demands before registering them. 

Other Arqon control applications also need to read the entanglement generation capabilities along end-to-end routes, in particular the Demand Manager. For this reason, the path partition $\Pi$ is updated by the Network Capabilities Manager to specify these capabilities for each path ${\pi \in \Pi}$.

In practice, the NCM may establish the end-to-end entanglement generation capabilities in a variety of ways. A direct method would be for the central controller to periodically schedule blocks of time for characterization of end-to-end routes. An implementation may include a method of determining the entanglement generation capabilities of an end-to-end route from the capabilities of every link in the route individually.

\subsection{Demand Manager}\label{sec: Demand Manager}
The Demand Manager, is responsible for receiving, screening and registering incoming demands, converting demands into PGTs, maintaining an accurate ledger of demand statuses, and creating the status messages that the central controller sends to end nodes.  
These processes are implemented by an \textit{Intake and Status Manager}, a \textit{Demand Registration} process, and a \textit{Buffer Manager}. 
The Demand Manager provides three interfaces: the \verb|DemandStatus| (DM:DS) interface, the \verb|Buffer| (DM:B) interface and the \verb|Demand| (DM:D) interface. 
The DM:DS interface is used by the Network Scheduler to update demand statuses with \texttt{ACCEPT}/\texttt{REJECT} outcomes after admission control. 
The DM:B interface allows the Network Scheduler to retrieve a buffer of registered demands and a buffer of demands to terminate, held by the Buffer Manager. 
Finally, the DM:D interface is an interface to the central controller through which the Demand Manager retrieves arriving demands and termination requests for existing demands.
The Demand Manager with all internal components and interfaces is illustrated in Figure~\ref{fig: demand manager}.



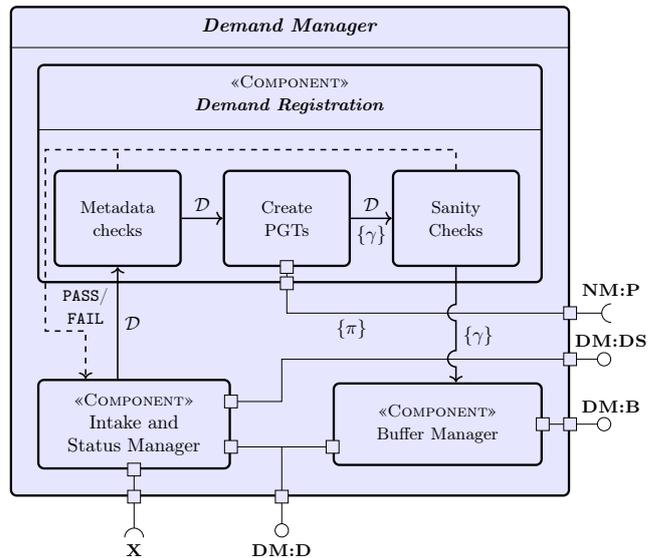
\begin{figure}
    \centering
    \resizebox{0.98\linewidth}{!}{
        \begin{tikzpicture}[x=1cm, y=1cm]

\node[component, minimum height=9cm, minimum width = 10.2cm] (DM) at (0,0) {};
\node [below = 0.1cm of DM.north, align=center] {\textbf{\textit{Demand Manager}}};
\coordinate (DMLine) at ($(DM.north)!0.17!(DM.center)$);
\draw ($(DMLine -| DM.east)$) -- (DMLine) -- ($(DMLine -| DM.west)$);

\coordinate (DMTopLeft) at ($(DMLine -| DM.west)$);
\coordinate (DMTopRight) at ($(DMLine -| DM.east)$);

\coordinate (DM1) at ($(DMTopRight)!0.1!(DM.south east)$);
\coordinate (DM2) at ($(DMTopRight)!0.2!(DM.south east)$);
\coordinate (DMExternalRight) at ($(DMTopRight) + (0.75cm, 0)$);
\coordinate (DMExternalBelow) at ($(DM.south) + (0, -0.75cm)$);

\node[component, inner sep = 0.25cm] (IM) [right = 0.5cm of DM.south west, anchor=south west, yshift =0.5cm] {\small\textsc{«Component»}\\ \normalsize Intake and Status Manager};

\node[component, minimum height=4cm, minimum width =9.2cm, right=0.5cm of DMTopLeft, anchor=north west,] (DR)  [yshift=-0.3cm] {};

\node[below=0.1cm of DR.north, align=center] (DRText) {\small «\textsc{Component}» \\ \small \textbf{\textit{Demand Registration}}};
\coordinate (DRLine) at ($(DR.north)!0.6!(DR.center)$) {};
\draw ($(DRLine -| DR.east)$) -- (DRLine) -- ($(DRLine -| DR.west)$);

\node[component, text width = 1.5cm, minimum height=1.75cm] (Meta) [above=0.3cm of DR.south west, anchor = south west, xshift=0.3cm] {
\small Metadata checks
};
\node[component, text width = 1.5cm, minimum height=1.75cm] (Format) [right=0.75cm of Meta.east, anchor = west] {
\small Create \\ PGTs
};
\node[component, text width = 1.5cm, minimum height=1.75cm] (Sanity) [right=0.75cm of Format.east, anchor = west] {
\small Sanity Checks
};

\node[component] (BM) [xshift=-0.5cm, anchor=east] at ($(DM.east |- IM.east)$) {\small «\textsc{Component}» \\ Buffer Manager};


\draw[-{Circle[open, scale=2]}] (BM.east) -- ($(BM.east -| DMExternalRight)$) node [above, pos=1, yshift=0.1cm] {\small \textbf{\textsc{DM:B}}};

\node[rectangle, fill=blue!10, draw=black] at (BM.east) {};
\node[rectangle, fill=blue!10, draw=black] at ($(BM.east -| DM.east)$) {};

\coordinate (DMDS1) at ($(BM.north)!0.2!(Sanity.south)$);
\coordinate (DMDS2) at ($(IM.east)!0.5!(BM.west)$);
\coordinate (DMDS3) at ($(IM.east)!0.5!(IM.north east)$);
\coordinate (DMDS4) at ($(DMDS3 -| DMDS2)$);

\draw[-{Circle[open, scale=2]}] (DMDS3) -- (DMDS4) |- (DMDS1) -- ($(DMDS1 -| DMExternalRight)$) node [above, yshift=0.1cm] {\small \textbf{\textsc{DM:DS}}};
\node[rectangle, fill=blue!10, draw=black] at (DMDS3) {};
\node[rectangle, fill=blue!10, draw=black] at ($(DMDS1 -| DM.east)$) {};

\coordinate (NMP1) at ($(BM.north)!0.6!(Sanity.south)$);
\draw[-{Arc Barb[reversed, scale=2.2]}] (Format.south) |- (NMP1) node [pos=0.7, below] {$\{\pi\}$} -- ($(NMP1 -| DMExternalRight)$) node [above, yshift=0.2cm] {\small \textbf{\textsc{NM:P}}};
\node[rectangle, fill=blue!10, draw=black] at (Format.south) {};
\node[rectangle, fill=blue!10, draw=black] at ($(Format.south |- DR.south)$) {};
\node[rectangle, fill=blue!10, draw=black] at ($(NMP1 -| DM.east)$) {};

\coordinate (DMD1) at ($(IM.south east)!0.5!(IM.east)$);
\coordinate (DMD2) at ($(IM.east)!0.5!(BM.west)$);
\coordinate (DMD3) at ($(DMD1 -| DMD2)$);
\coordinate (DMD4) at ($(DMD3 -| BM.west)$);

\draw (DMD1) -- (DMD3) -- (DMD4);
\draw[-{Circle[open, scale=2]}] (DMD3) -- ($(DMD3 |- DMExternalBelow)$) node [below] {\small \textbf{\textsc{DM:D}}};
\node[rectangle, fill=blue!10, draw=black] at ($(DMD3 |- DM.south)$) {};
\node[rectangle, fill=blue!10, draw=black] at (DMD1) {};
\node[rectangle, fill=blue!10, draw=black] at (DMD4) {};

\draw[-{Arc Barb[reversed, scale=2.2]}] (IM.south) -- ($(IM.south |- DMExternalBelow)$) node [below] {\small \textbf{\textsc{X}}};;
\node[rectangle, fill=blue!10, draw=black] at (IM.south) {};
\node[rectangle, fill=blue!10, draw=black] at ($(IM.south |- DM.south)$) {};

\draw[<-, thick] (Meta.south) -- ($(Meta.south |- IM.north)$) node [pos =0.5, right] {$\mathcal{D}$};
\draw[->, thick] (Meta.east) -- (Format.west) node [pos=0.5, above] {$\mathcal{D}$};
\draw[->, thick] (Format.east) -- (Sanity.west) node [pos=0.5, above] {$\mathcal{D}$} node [pos=0.5, below] {$\{\gamma\}$};

\coordinate (PGTS1) at ($(Sanity.south |- BM.north)$);
\coordinate (PGTS2) at ($(PGTS1)!0.675!(Sanity.south)$);
\coordinate (PGTS3) at ($(PGTS1)!0.275!(Sanity.south)$);

\draw[->, thick] (Sanity.south) -- (PGTS2) arc (90:270:0.15cm) -- (PGTS3) node [pos=0.5, right] {$\{\gamma\}$} arc (90:270:0.15cm) -- (PGTS1);

\coordinate (AR0) at ($(Sanity.north |- DRLine)$);
\coordinate (AR1) at ($(Sanity.north)!0.5!(AR0)$);
\coordinate (AR2) at ($(Meta.west -| DR.west)$);
\coordinate (AR3) at ($(AR2)!0.5!(Meta.west)$);

\coordinate (AR4) at ($(IM.north west)!0.5!(DR.south west)$);
\coordinate (AR5) at ($(IM.north)!0.5!(IM.north west)$);
\coordinate (AR6) at ($(AR5 |- AR4)$);

\coordinate (AR7) at ($(Meta.north |- AR1)$);

\draw[->, dashed, thick] (Sanity.north) |- (AR1) -| (AR3) |- (AR6) node [above, align=center] {\small \texttt{PASS}/\\\small \texttt{FAIL}} -- (AR5);
\draw[dashed, thick] (Meta.north) -- (AR7);

\end{tikzpicture}

    }
    \caption{Detailed diagram of the internals of the Demand Manager. The arrows show the process of registering a demand, showing how the demand submission $\mathcal{D}$ is passed from metadata checks to Create PGTs. The relevant paths $\{\pi\}$ are obtained via the NM:P interface, before the created PGTs $\{\gamma\}$ are also passed to the sanity checks. Finally, any PGTs which pass the sanity checks are passed to the buffer manager and an \texttt{PASS/FAIL} message is sent back to the intake and status manager as appropriate. The interface X provided by the central controller is used to update the statuses of demands at the relevant end nodes.}
    \label{fig: demand manager}
\end{figure}


\subsubsection{Intake and Status Manager}\label{ssec: Demand Intake Manager}

The Intake and Status Manager is the first process which a demand interacts with. For every incoming demand it first assigns a unique demand ID and creates an entry in a ledger of demand statuses, with an initial status flag indicating that the demand is not yet registered. The demand status flag is updated as it progresses through each further processing step. The demand is not removed from the status ledger until it is either rejected, expires, or is terminated by the end nodes. 
Each incoming demand is added to a First In First Out (FIFO) queue, where it remains until a worker of the Demand Registration process removes it. 
An implementation may also specify an alternative queuing system for Demand Registration, including a priority based queuing system. 

The Intake and Status Manager interprets \texttt{PASS}/\texttt{FAIL} outcomes received from Demand Registration and \texttt{ACCEPT}/\texttt{REJECT} outcomes received from the Network Scheduler through the DM:DS interface. Upon receipt of a \texttt{FAIL} or \texttt{REJECT} outcome it issues a demand rejection message to end nodes, with an error code indicating whether the demand failed in Demand Registration (\texttt{FAIL}) or was rejected by the Network Scheduler (\texttt{REJECT}). Upon receipt of a \texttt{PASS} or \texttt{ACCEPT} outcome, it updates the status of the demand. For an \texttt{ACCEPT} outcome the Intake and Status Manager further creates a demand acceptance message for end nodes, which includes a service agreement. 

\subsubsection{Demand Registration}\label{ssec: Demand Registration}
Demand Registration comprises three phases: a metadata check, the creation of a set of Packet Generation Tasks (PGTs), and implementation specific sanity checks. 
There are a finite number of parallel instances of the Demand Registration process, known as \textit{workers}, each of which handles the registration of one demand at a time. 
Workers return a \texttt{PASS}/\texttt{FAIL} outcome for each demand to the Intake Manager and take the next demand from the demand queue, as long as it is not empty. 
An implementation may specify any number of workers, including just a single worker.

\begin{enumerate}
    \item Metadata check:

    The following assertions must be true for a demand to be valid:
\begin{enumerate}
    \item The network capability version identifier is of the expected data type and corresponds to the latest version;
    \item Both end node id's correspond to registered nodes and are of the expected data type;
    \item There is an allowed/valid path $\pi$ connecting the source and destination nodes;
    \item The message is a well-formatted demand adhering to the format in \eqref{eq: demand raw}. 
    \item Each of $\ninst, \ F, \ w, s \ > 0$ and all fields are of the expected data type.
    \item The expiry time $\texpiry$ is later than $t^{\texttt{start}}$, ${\texpiry > t^{\texttt{start}}}$.
\end{enumerate}
These assertions are necessary to ensure that a demand can be handled in a consistent manner by the network. Additional metadata checks may be specified by an implementation. 

\item Create Packet Generation Tasks:
A PGT is associated with a particular path from the source to destination end node. 
To create a PGT $\gamma_{d}$ for a demand $d$ with a path ${\pi_{\gamma_d} \in \Pi}$ it is necessary to set the duration of the PGAs ($E_{\gamma_d}$), to calculate the packet success probability ($\ppacket_{\gamma_d}$), and to calculate the minimum allocation ($\minimumallocation_{\gamma_d}$). In Section~\ref{sec: Implementation}, we specify a particular method of setting $E_{\gamma_d}$ and calculating $\ppacket_{\gamma_d}$.

The minimal allocation (Definition~\ref{def: minAlloc}) is the key parameter of the service agreement that Arqon creates for each accepted demand. It depends on the service error parameter, $\epsilon^{\text{service}}_d$, which is either specified by the end nodes as part of the demand, or set by the central controller. 


\begin{definition}[Packet Generation Task]
   A \emph{\textbf{packet generation task}} is a tuple \begin{equation}\label{eq:PGT}
     \gamma = (E_\gamma, \minimumallocation_\gamma, \pi_\gamma, \tminsep_\gamma, t^{\texttt{\emph{expiry}}}_{\gamma}, t^{\texttt{\emph{start}}}_{\gamma})
   \end{equation} where $E_\gamma$ is the execution time of each packet generation attempt (PGA), $\minimumallocation_\gamma$ is the minimum allocation for $\gamma$, ${\pi_{\gamma}\in \Pi}$ is a valid entanglement generation path and $\tminsep_{\gamma}, t^{\texttt{\emph{expiry}}}_\gamma$ are as in the raw demand \eqref{eq: demand raw}, and $t^{\texttt{\emph{start}}}_{\gamma}$ is the first scheduling interval in which a PGA for $\gamma$ could be scheduled, based on the scheduling offset.  
\label{def: PGT}
\end{definition}

In general, there may be multiple valid entanglement generation paths between two end nodes.
Each of these paths is only suitable to realize the demand if its average rate of end-to-end entanglement generation is high enough so that demanded $s_d$ pairs can be generated within a window of time $w_d$ with a non-zero probability ${\ppacket_{\gamma_d} > 0}$.
In the Create PGTs phase of Demand Registration each demand $d$ is mapped to a set $\bm{\Gamma}_d$ of suitable PGTs, from which at most one particular PGT $\gamma$ will be accepted by the Network Scheduler.
This also means that a routing decision for a demand $d$ is not made until the Admit Tasks process of the Network Scheduler, where a particular ${\gamma \in \bm{\Gamma}_d}$ may be accepted for scheduling.

\item Implementation specific sanity checks: 

Once $\bm{\Gamma}_d$ is constructed, then each possible PGT  ${\gamma\in \bm{\Gamma}_d}$ may undergo further sanity checks. 
If a PGT $\gamma$ fails a sanity check, then it is removed from $\bm{\Gamma}_d$. 
We define one necessary sanity check, to verify that for each proposed PGT $\gamma \in \bm{\Gamma}_d$ \begin{equation}
    E_{\gamma} \leq T^{SI},\label{eq: execution time sanity check}
\end{equation}
meaning that the execution time of a single PGA of the PGT $\gamma$ must not exceed the scheduling interval.

If all PGTs fail the sanity check, then the demand is marked as invalid and a \texttt{FAIL} outcome is returned to the Intake Manager.
Otherwise, a \texttt{PASS} outcome is returned to the Intake Manager and $\bm{\Gamma}_d$ is passed to the Buffer Manager to be added to the demand intake buffer, denoted $\bm{\Gamma}$. 

\end{enumerate}

\subsubsection{Buffer Manager}\label{sssec: Buffer Manager}
The sets of PGTs $\bm{\Gamma}_d$ which can realize each demand $d$ that passes Demand Registration are needed as an input to the Network Scheduler control application. 
Between Demand Registration and the Network Scheduler however there is a type of timing boundary, in that the Network Scheduler executes a set of periodic processes once per scheduling interval, whereas Demand Registration occurs based on the queue arrival dynamics of demands and the availability of workers. 
Due to this mismatch in process timings, the set of PGTs $\gamma_d$ for each registered demand $d$ are placed into a task intake buffer $\bm{\Gamma}$, which is maintained by the Buffer Manager. 

Within the Buffer Manager, $\bm{\Gamma}$ is write only. In contrast, the Demand Manager exposes $\bm{\Gamma}$ to the Network Scheduler control application via the DM:B interface, but the buffer is exposed with read only access. 
The structure used to represent what is read from the task intake  buffer by the Network Scheduler is referred to as the \textit{task intake object} and is also denoted $\bm{\Gamma}$ for simplicity.  
After the buffer has been read by the Network Scheduler, the Buffer Manager must implement a flush of the buffer. In an implementation, the Buffer Manager may flush the entire buffer after each read. Alternatively, it may be that the Network Scheduler only reads a maximum number $M$ of sets of PGTs ${\bm{\Gamma}_i, \ i \in \{ 0, 1, \cdots, M-1\}}$ from the buffer in each read. In that case the Buffer Manager may implement a flush of only the top $M$ sets of PGTs after a read.

Similarly, it needs to maintain a buffer of requests from end nodes to terminate demands, called the \textit{termination buffer} $\tau$.
As with $\bm{\Gamma}$, this is write-only from the within the Buffer Manager itself, and via the DM:B interface it is exposed in read-only form to the Network Scheduler. 
Unlike $\bm{\Gamma}$, the Network Scheduler should always read the entire termination buffer.
Once this has occurred, the entire buffer is flushed so that it is ready to receive more termination requests in the next scheduling interval.

\subsection{Network Scheduler}\label{ssec: Network Scheduler control application}


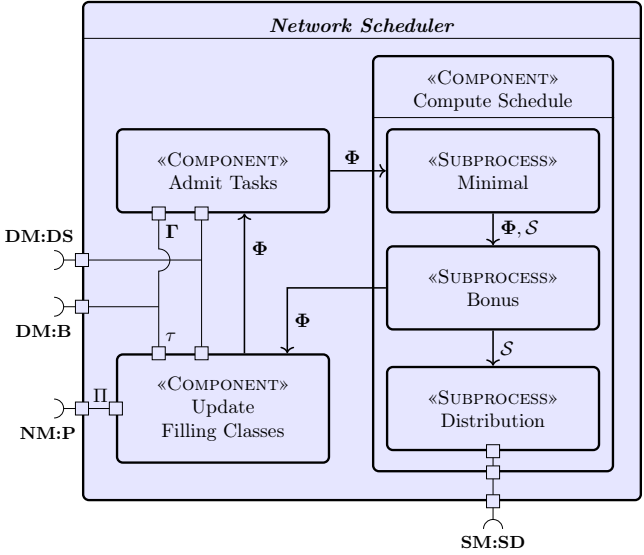
\begin{figure}
    \centering
    \resizebox{0.98\linewidth}{!}{
        \begin{tikzpicture}[x=1cm, y=1cm]

\node[component, minimum width = 10cm, minimum height=9cm] (NS) at (0,0) {};
\node[below =0.2cm of NS.north] {\textbf{\textit{Network Scheduler}}};

\coordinate (NSLine) at ($(NS.north)!0.15!(NS.center)$);
\coordinate (NSTopLeft) at ($(NS.west |- NSLine)$);
\coordinate (NSTopRight) at ($(NS.east |- NSLine)$);

\draw (NSTopLeft) -- (NSLine) -- (NSTopRight);

\node[component] (CS) [
    left = 0.5cm of NSTopRight,
    minimum height = 7.5cm,
    text width = 3.5cm,
    yshift=-0.3cm,
    anchor=north east
] {};

\node[below=0.15cm of CS.north, anchor=north, align=center] {«\textsc{Component}»\\Compute Schedule};
\coordinate (CSLine) at ($(CS.north)!0.3!(CS.center)$);
\coordinate (CSTopLeft) at ($(CS.west |- CSLine)$);
\coordinate (CSTopRight) at ($(CS.east |- CSLine)$);
\draw (CSTopLeft) -- (CSTopRight);

\coordinate (CSMinimal) at ($(CSTopLeft)!0.15!(CS.south west)$);
\coordinate (CSBonus) at ($(CSTopLeft)!0.48!(CS.south west)$);
\coordinate (CSDist) at ($(CSTopLeft)!0.82!(CS.south west)$);

\node[component] (Minimal) at ($(CSMinimal -| CS.north)$) {«\textsc{Subprocess}» \\ Minimal};
\node[component] (Bonus) at ($(CSBonus -| CS.north)$) {«\textsc{Subprocess}» \\ Bonus};
\node[component] (Distribute) at ($(CSDist -| CS.north)$) {«\textsc{Subprocess}» \\ Distribution};

\node[component] (AC) [left = 1cm of Minimal,anchor=east] {«\textsc{Component}»\\Admit Tasks};

\node[component] (UFC) at ($(AC |- Distribute)$) {«\textsc{Component}»\\Update\\Filling Classes};

\coordinate (NSExtLeft) at ($(NS.west) + (-0.5cm, 0)$);

\draw[-{Arc Barb[reversed, scale=2.2]}] (UFC.west) -- ($(UFC.west -| NSExtLeft)$) node [below, yshift=-0.2cm, pos=1.1] {\small \textbf{NM:P}} node [pos=0.25, above] {$\Pi$};
\node[rectangle, fill=blue!10, draw=black] at (UFC.west) {};
\node[rectangle, fill=blue!10, draw=black] at ($(UFC.west -| NS.west)$) {};

\coordinate (DMDS1) at ($(AC.south west)!0.333!(UFC.north west)$);
\coordinate (DMDS2) at ($(AC.south west)!0.4!(AC.south east)$);
\coordinate (DMDS3) at ($(UFC.north west)!0.4!(UFC.north east)$);

\draw[-{Arc Barb[reversed, scale=2.2]}] ($(DMDS2 |- DMDS1)$) -- ($(DMDS1 -| NSExtLeft)$) node [above, yshift=0.2cm, pos=1.1] {\small \textbf{DM:DS}};
\draw (DMDS2) -- (DMDS3);

\node[rectangle, fill=blue!10, draw=black] at (DMDS2) {};
\node[rectangle, fill=blue!10, draw=black] at (DMDS3) {};
\node[rectangle, fill=blue!10, draw=black] at ($(DMDS1 -| NS.west)$) {};

\coordinate (DMB1) at ($(AC.south west)!0.667!(UFC.north west)$);
\coordinate (DMB2) at ($(AC.south west)!0.2!(AC.south east)$);
\coordinate (DMB3) at ($(UFC.north west)!0.2!(UFC.north east)$);

\coordinate (DMB4) at ($(DMB3 |- DMDS1) + (0,-0.2cm)$);

\draw[-{Arc Barb[reversed, scale=2.2]}] ($(DMB2 |- DMB1)$) -- ($(DMB1 -| NSExtLeft)$) node [below, yshift=-0.2cm, pos=1.1] {\small \textbf{DM:B}};
\draw (DMB3) -- (DMB4) node [pos=0.2, right] {$\tau$} arc (-90:90:0.2cm) -- (DMB2) node [pos=0.5, right] {$\mathbf{\Gamma}$};

\node[rectangle, fill=blue!10, draw=black] at (DMB3) {};
\node[rectangle, fill=blue!10, draw=black] at (DMB2) {};
\node[rectangle, fill=blue!10, draw=black] at ($(DMB1 -| NS.west)$) {};

\coordinate (UFC1) at ($(UFC.north west)!0.6!(UFC.north east)$);
\draw[->, thick]  (UFC1) -- ($(UFC1 |- AC.south)$) node [pos=0.75, right] {$\mathbf{\Phi}$};
\draw[->, thick] (AC.east) -- (Minimal.west) node [pos=0.4, above] {$\mathbf{\Phi}$};
\draw[->, thick] (Minimal.south) -- (Bonus.north) node [pos =0.5, right] {$\mathbf{\Phi}, \mathcal{S}$};
\draw[->, thick] (Bonus.south) -- (Distribute.north) node [pos=0.5, right] {$\mathcal{S}$};

\coordinate (UFC2) at ($(UFC.north west)!0.8!(UFC.north east)$);
\draw[->, thick] (Bonus.west) -| (UFC2) node [pos=0.75, right] {$\mathbf{\Phi}$};

\coordinate (NSExtBelow) at ($(NS.south) + (0,-0.5cm)$);
\draw[-{Arc Barb[reversed, scale=2.2]}] (Distribute.south) -- ($(Distribute.south |- NSExtBelow)$) node [below] {\small \textbf{SM:SD}};
\node[rectangle, fill=blue!10, draw=black] at (Distribute.south) {};
\node[rectangle, fill=blue!10, draw=black] at ($(Distribute.south |- CS.south)$) {};
\node[rectangle, fill=blue!10, draw=black] at ($(Distribute.south |- NS.south)$) {};

\end{tikzpicture}
}
\caption{Detailed diagram of the Network Scheduler. Arrows show the movement of the set of filling classes $\mathbf{\Phi}$ and the network schedule $\mathcal{S}$ through the scheduling process. }
\label{fig: network scheduler diagram}

\end{figure}
\begin{algorithm}[t]
    \SetKwInOut{Input}{Input}
    \SetKwInOut{Output}{Output}
    \SetKw{Return}{return}
    \SetKw{Continue}{continue}
    \SetKw{Break}{break}
    \SetKw{Sleep}{sleep}
    \SetKwProg{Fn}{Function}{:}{end}
    \Fn{\textsc{NetworkScheduler}}{
    \Input{Set of filling classes $\Phi$}
    \Output{List of accepted PGTs $\bm{A}$, network schedule $S$}
    \While{True}{
        Get $\Pi$ from the Network Manager \;
        Get $\tau, \ \bm{\Gamma}$ from the Demand Manager\;
        Set $\Phi, \xi \leftarrow$ \textsc{UpdateFillingClasses}($\Phi, \tau, \Pi$) \;
        Set $\Phi, \mathbf{A}\leftarrow\textsc{AdmitTasks}(\bm{\Gamma}, \Phi, \xi)$\;

        Send ($\mathbf{A}, \bm{\Gamma}\setminus\mathbf{A}$) to the Demand Manager\;
        Set $S\leftarrow\textsc{ComputeSchedule}(\Phi)$ \;

        Send $S$ to the Schedule Manager\;
        \Sleep until the start of the next scheduling interval\;
    }
}

\caption{Algorithm for the main process of the Network Scheduler control application.
\textsc{UpdateFillingClasses}, \textsc{AdmitTasks} and \textsc{ComputeNewSchedule} are given in Algorithms~\ref{alg: Update Filling Classes}, \ref{alg: Admit Tasks} and \ref{alg: Compute Schedule} respectively. 
}
\label{alg: Network Scheduler main}
\end{algorithm}

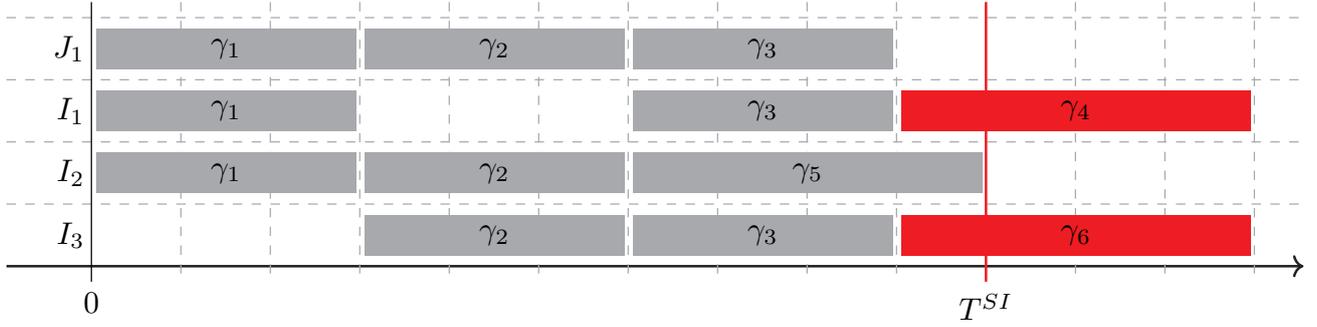
\begin{figure*}
\centering
\resizebox{0.95\textwidth}{!}{\begin{tikzpicture}

\pgfmathsetmacro{\rowHeight}{0.7}  

\coordinate (J1) at (0, -0 * \rowHeight);
\coordinate (I1) at (0, -1 * \rowHeight);
\coordinate (I2) at (0, -2 * \rowHeight);
\coordinate (I3) at (0, -3 * \rowHeight);

\foreach \a in {0.5, -0.5, -1.5, -2.5} {
    \draw[draw=Black!40, dashed] (-1, \a * \rowHeight) -- (13.5, \a * \rowHeight);
}

  \draw[draw=Black, thick, ->] (-1, -3.5*\rowHeight) -- (13.5, -3.5*\rowHeight);

  \foreach \a in {1,...,13}{
    \draw[draw=Black!40, dashed] (\a-0.05, 0.75*\rowHeight) -- (\a-0.05, -3.75*\rowHeight);}

  \draw (-0.05, 0.75 * \rowHeight) -- (-0.05, -3.75 * \rowHeight) node [below] {0};

\node[rectangle, fill=Black!40, minimum width=2.9cm, right=0cm of J1] {$\gamma_1$};
\node[rectangle, fill=Black!40, minimum width = 2.9cm, right = 0cm of I1] {$\gamma_1$};
\node[rectangle, fill=Black!40, minimum width = 2.9cm, right = 0cm of I2] {$\gamma_1$};

\node[rectangle, fill=Black!40, minimum width=2.9cm, right=3cm of J1] {$\gamma_2$};
\node[rectangle, fill=Black!40, minimum width = 2.9cm, right = 3cm of I2] {$\gamma_2$};
\node[rectangle, fill=Black!40, minimum width = 2.9cm, right = 3cm of I3] {$\gamma_2$};

\node[rectangle, fill=Black!40, minimum width = 2.9cm, right = 6cm of J1] {$\gamma_3$};
\node[rectangle, fill=Black!40, minimum width = 2.9cm, right = 6cm of I3] {$\gamma_3$};
\node[rectangle, fill=Black!40, minimum width = 2.9cm, right = 6cm of I1] {$ \gamma_3$};

\node[rectangle, fill=Red, minimum width = 3.9cm, right = 9cm of I1] {$\gamma_4$};

\node[rectangle, fill=Black!40, minimum width = 3.9cm, right = 6cm of I2] {$\gamma_5$};

\node[rectangle, fill=Red, minimum width = 3.9cm, right = 9cm of I3] {$ \gamma_6$};

\node [left=0cm of J1] {$J_1$};
\node [left=0cm of I1] {$I_1$};
\node [left=0cm of I2] {$I_2$};
\node [left=0cm of I3] {$I_3$};

  \draw[draw=Red, thick] (9.95, 0.75*\rowHeight) -- (9.95, -3.75*\rowHeight) node [below] {$T^{SI}$};

\end{tikzpicture}}
  \caption{Excerpt from a possible network schedule for four components $J_1$, $I_{1-3}$ and a set of PGTs $\gamma_{1-6}$. PGAs for a PGT $\gamma_{i}$ are illustrated by gray boxes scheduled simultaneously on every component on their path $\pi_{\gamma_i}$. PGAs for $\gamma_{4}$ and $\gamma_{6}$, illustrated by red boxes, can not be added to the schedule because $I_1$ (for $\gamma_4$) and $I_{3}$ (for $\gamma_6$) are not available for sufficiently long blocks of uninterrupted time before the end of the scheduling interval, illustrated by the red line labeled $T^{SI}$.}
\label{fig:badAccounting}
\end{figure*}

The Network Scheduler is responsible for determining which PGTs to accept and for producing network schedules that uphold service agreements with end nodes.
These network schedules must include a minimum allocation of PGAs for each accepted PGT. It includes an \textit{Admit Tasks} process that examines PGTs ${\gamma \in \bm{\Gamma}}$ in the task intake buffer and determines if they can be accepted without disrupting already accepted PGTs.  
It also includes a \textit{Compute Schedule} process that occurs after Admit Tasks and produces these network schedules. 
This control application is time critical, as its processes must start and complete within a single scheduling interval to ensure distribution and execution of a network schedule will not be delayed.
Algorithm~\ref{alg: Network Scheduler main} for \textsc{NetworkScheduler} is an overview of the main loop of the Network Scheduler control application, and Figure~\ref{fig: network scheduler diagram} illustrates the Network Scheduler with all internal components.

In our network model, we assume that each network component can only be scheduled to execute PGAs for a single PGT at any given time, and each PGA is non-preemptive. Two different PGAs $\gamma_1$ and $\gamma_2$ can be scheduled simultaneously as long as ${\pi_{\gamma_1} \cap \pi_{\gamma_2} = \emptyset}$. 
To account for these considerations, Arqon structures network schedules as a set of aligned resource schedules, one per network component.
To schedule a PGA for a PGT $\gamma$, all resources on the path $\pi_{\gamma}$ of the PGT must be simultaneously available for an uninterrupted period of time that is at least as long as the PGA duration $E_{\gamma}$. 
As illustrated in Figure \ref{fig:badAccounting}, it is possible for a network schedule to contain gaps of time in which certain resources have no PGAs scheduled, but the subsequently scheduled PGAs can not be started earlier due to the requirement that all resources on the path $\pi_{\gamma}$ be simultaneously available. Figure \ref{fig:badAccounting} also illustrates that there are situations in which the total amount of available time on a network resource is greater than or equal to a PGA duration, however the available time is non-consecutive and therefore the PGA cannot actually be scheduled.

To determine which PGTs can be simultaneously scheduled, Arqon leverages a partition of active PGTs into cells called \textit{filling classes} that are based on the path partition $\Pi$. 
Active PGTs are those that have been accepted in some scheduling interval and have neither expired nor been terminated. The filling classes are used by both Admit Tasks and Compute Schedule.
\begin{restatable}[Filling Classes]{definition}{defFillingClasses}\label{def: filling classes}
Let $\Pi = \{\Pi_\phi\}$ be a disjoint partition of a set of paths $\mathcal{P}$ and let $Z$ be a set of PGTs such that $\forall \gamma \in Z, \ \exists \phi \text{ s.t. } \pi_{\gamma} \in \Pi_{\phi}$. The set of \emph{\textbf{filling classes}} $\Phi\subset 2^Z\times\Pi$ has elements \begin{equation}\label{eq: filling classes}
        \phi = (Z_\phi, \Pi_\phi),
    \end{equation} where $Z_\phi = \{\gamma\in Z : \pi_\gamma\in\Pi_\phi\}$.
\end{restatable}

The set of paths $\mathcal{P}$ in Definition \ref{def: filling classes} may be either the set of valid entanglement generation paths $\mathcal{P}_{\text{valid}}$, or a further restriction to a set of allowed paths $\mathcal{P}_{\text{allowed}}$.

\begin{remark}
    The set of filling classes $\Phi$ are essentially containers ${\phi \in \Phi}$ for active PGTs, based on the cells of the path partition ${\Pi_{\phi} \in \Pi}$. The set of filling classes does not need to be populated by PGTs, which means that $\Phi$ can be created based on the path partition $\Pi$ and an empty set of PGTs $Z=\emptyset$. Once created, the set of filling classes may be updated by populating the containers ${\phi \in \Phi}$ with a set of PGTs $Z\neq{\emptyset}$.
\end{remark}


To identify which resources are required by at least one path in a filling class we introduce the mapping of associated resources between a set of paths and the set of resources on at least one path in the set.
 \begin{restatable}[Associated Resources]{definition}{defFCAssociatedRes}\label{def: FCAssociatedResource} 
    Let $\mathcal{P}$ be a set of paths and let $\mathcal{R}$ be the set of resources which are on some path in $\mathcal{P}$. Let $\{\pi_x\}_x \subset \mathcal P$. The \emph{\textbf{resources associated}} with $\{\pi_x\}_x$ are indicated by $\xi\big{(}\{\pi_x\}_x \big{)}$, as defined by the following mapping: \begin{align}
        \xi :& \ 2^{\mathcal{P}} \to 2^\mathcal{R}\label{eq: domain and range resource dicationary mapping}\\ 
        & \{\pi_x\}_x \mapsto \{r\in\mathcal{R} : \exists\pi\in \{\pi_x\}_x~s.t.~r\in\pi\}\label{eq: resources associated with a filling class}.
    \end{align}
\end{restatable}

\subsubsection{Update Filling Classes}\label{sssec: update filling classes}

In every scheduling interval, the first process that the Network Scheduler does is Update Filling Classes. This process happens first, as both Admit Tasks and Compute Schedule require the updated set of filling classes as input.

The Network Scheduler carries forward the filling classes from a scheduling interval to the subsequent one. However, before proceeding it is necessary to retrieve the path partition from the Network Manager via the NM:P interface and check for a change to the version id. In case of an update to the path partition, the filling classes are rebuilt. 
There may be one or more PGTs $\gamma$ for which the path $\pi_{\gamma}$ is no longer in the path partition. In that case, Update Filling Classes must apply a \textit{missing path} rule. 
The default missing path rule is to attempt to re-admit the demand as an alternative PGT $\gamma'$ realizing the demand, for which $\pi_{\gamma'}$ is in the updated path partition. This may result in an update the service agreement for the demand to reflect the minimal allocation of the alternate PGT.  
If no alternate PGT can be admitted then Update Filling Classes removes the PGT and updates the demand status to removed so that the Demand Manager can notify the end nodes and cancel the service agreement.

Update Filling Classes is also responsible for removing terminated or expired PGTs. 
To check for PGT terminations, the Network Scheduler reads the termination buffer $\tau$ from the Demand Manager via the DM:B interface. 
Every PGT is checked and those for which there is a termination message or the expiry time has been reached are removed from their filling class.  

The full process Update Filling Classes is described by  
Algorithm~\ref{alg: Update Filling Classes}, \textsc{UpdateFillingClasses} in Appendix~\ref{app: Network Scheduler Algorithms}.

\subsubsection{Admit Tasks}\label{sssec: Admission Control}
\begin{algorithm*}[t]
    \SetKwInOut{Input}{Input}
    \SetKwInOut{Output}{Output}
    \SetKw{Return}{return}
    \SetKw{Continue}{continue}
    \SetKw{Break}{break}
    \SetKwProg{Fn}{Function}{:}{end}
    \Fn{\textsc{CalculateRequiredTime}}{
    \Input{Filling class $\phi = (Z_{\phi}, \Pi_{\phi})$}
    \Output{Required time for filling class, $R(Z_{\phi})$}

            Order PGTs $\gamma \in Z_{\phi}$ by a set of indices $X=\{1, \cdots, |Z_{\phi}|\}$ such that $\minimumallocation_x\leq \minimumallocation_{x+1}$ \;
            Set $M \leftarrow |Z_{\phi}|$ \;
            Set $n^{\phi}_0 \leftarrow\minimumallocation_0  -1 $\;
            Set $c^{\phi}_{0} \leftarrow \max\left( \underset{y \in \{0, \cdots M-1\}}{\max} \big{(} E_y + \tminsep_y \big{)} , \sum_{y=0}^{M-1} E_y \right)$ \; 
            \For{$x \in \{1, \cdots, M - 1\}$ }{
            Set $n^{\phi}_x \leftarrow \minimumallocation_x - \minimumallocation_{x-1}$\;
            Set $c^{\phi}_{x} \leftarrow$  $\max\left( \underset{y \in \{x, \cdots M-1\}}{\max} \big{(} E_y + \tminsep_y \big{)} , \sum_{y=x}^{M-1} E_y \right) $ \; 
            }
        
        Set $R(Z_{\phi}) \leftarrow \overset{M-1}{\underset{x=0}{\sum}} (n^{\phi}_{x} \cdot c^{\phi}_{x} + E_x)$ \;
        \Return $R(Z_{\phi})$
}
\caption[Admission Control]{Calculation of an upper bound on the minimum required time to schedule a minimum allocation of PGAs $\minimumallocation_{\gamma}$ for every PGT $\gamma \in Z_{\phi}$ in a filling class $\phi$. \label{alg: Calculate Required Time}}
\end{algorithm*}

The process Admit Tasks is described in full by Algorithm~\ref{alg: Admit Tasks},~ \textsc{AdmitTasks} in Appendix~\ref{app: Network Scheduler Algorithms}. 
It tests PGTs $\gamma'$ in the task intake buffer $\bm{\Gamma}$ to determine whether there is sufficient time available to schedule a minimum allocation $\minimumallocation_{\gamma}$ of PGAs in subsequent scheduling intervals without disrupting service to any active PGTs. 
Before testing any new PGTs ${\gamma' \in \bm{\Gamma}}$, it calculates the time required to schedule a minimum allocation of each active PGT and reduces the available time on all resources associated with the PGTs filling class by this duration. 
The function \textsc{CalculateRequiredTime}, defined in Algorithm \ref{alg: Calculate Required Time} implements this calculation. 
Recall that PGAs for a PGT $\gamma$ must be scheduled for uninterrupted periods of duration $E_{\gamma}$ on all components in $\pi_{\gamma}$. Gaps of time between PGAs of PGTs in the same filling class are not treated as available time by Admit Tasks, because they may not be long enough to allow a PGA to be scheduled. By reducing the available time for all resources associated with a PGTs filling class, rather than just for the resources on the PGTs path, Admit Tasks avoids considering gaps between PGAs as available time. This simplification prioritizes the satisfaction of service agreements, but it may induce an overestimation of the required time and limit the number of accepted PGTs.

Incoming PGTs in the task intake buffer are organized into sets ${\bm{\Gamma}_d \in \bm{\Gamma}}$ by the demand they realize, with each specific PGT ${\gamma'_d \in \bm{\Gamma}_d}$ corresponding to a different path between the end nodes. For each of these PGTs, Admit Tasks repeats the calculation of \textsc{CalculateRequiredTime} with $\gamma'_d$ included in the filling class it would be in (the filling class $\phi$ such that ${\pi_{\gamma'_{d}} \in \Pi_{\phi}}$). This required time is compared to the duration of the scheduling interval. If it is less than or equal to the scheduling interval the PGT is accepted, the other PGTs in $\bm{\Gamma}_d$ are removed and the demand status is set to accepted. Otherwise, the PGT is rejected and Admit Tasks checks the next PGT ${\gamma''_d \in \bm{\Gamma}_d}$, until either one is accepted or no alternatives remain -- in which case the demand status is set to rejected. If Admit Tasks accepts a particular PGT to realize a demand, it fixes the path along which the demand will be served. Hence Admit Tasks performs the function of routing.


\subsubsection{Compute Schedule}\label{ssec: Compute Schedule}
\begin{figure*}
    \centering
    \includegraphics[width=0.96\linewidth, trim={1.0cm 0.5cm 1.5cm 0cm}]{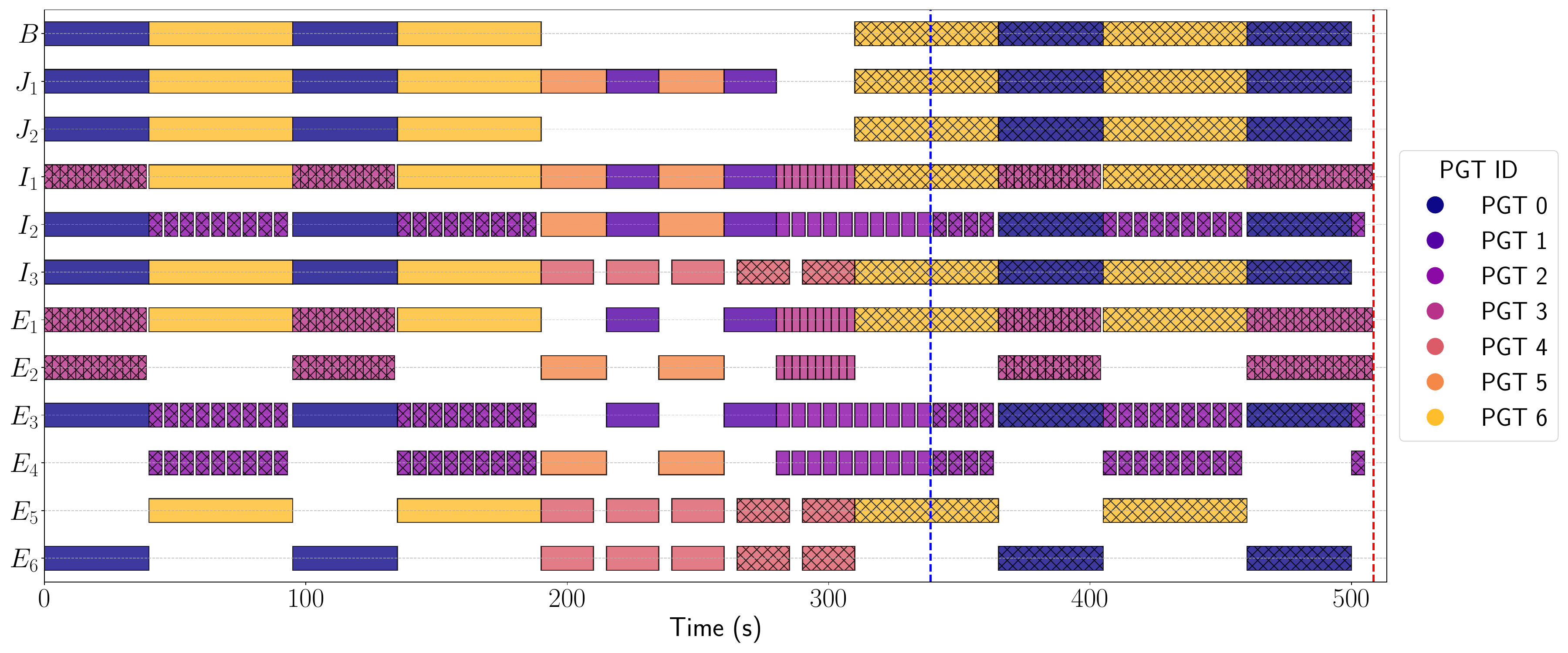}
    \caption{An example network schedule produced by \textsc{ComputeSchedule} for a simplified version of the dumbbell topology network illustrated in Figure~\ref{fig: example network topology}, with a six end nodes and seven active PGTs. End nodes $E_{1-2}$ connect to $I_1$, $E_{3-4}$ connect to $I_2$, and $E_{5-6}$ connect to $I_3$. The PGTs are configured as examples, with different values for their minimal allocations, minimum separations between PGAs, and PGA durations. They are organized into the specific filling classes defined in Section~\ref{sec: Implementation}, so that $
    {Z_{\phi^B} = \{\gamma_0, \gamma_6\}},
    {Z_{\phi^J_1} = \{\gamma_5, \gamma_1\}}, 
    {Z_{\phi^I_1} = \{\gamma_3\}}, {Z_{\phi^I_2} = \{ \gamma_2\}}, {Z_{\phi^I_3} = \{\gamma_4 \}}.$
    Solid (hatched) color boxes represent PGAs for PGTs of the matching color scheduled by the minimal (bonus) allocation phases of Compute Schedule. The vertical blue dashed line indicates the end of the required time to schedule a minimal allocation of PGAs for each PGT. The red vertical dashed line indicates the end of a scheduling interval. } 
    \label{fig: example network schedule}
\end{figure*}

\begin{algorithm}[t]
    \SetKwInOut{Input}{Input}
    \SetKwInOut{Output}{Output}
    \SetKw{Return}{return}
    \SetKw{Continue}{continue}
    \SetKw{Break}{break}
    \SetKwProg{Fn}{Function}{:}{end}

     \Fn{\textsc{MinimalAllocation}}{
       \Input{Set of Filling Classes $\Phi$, network schedule $S$}
       \Output{Updated network schedule $S$.}

    \For{$\phi\in\Phi$}{
        $S = \text{\textsc{DirectAllocation}}(\phi, S, t_0 + \sum_{\phi' < \phi} R(Z_{\phi'}))$\;
    }

     }

     \Fn{\textsc{BonusAllocation}}{
       \Input{Set of Filling Classes $\Phi$, network schedule $S$}
       \Output{Updated network schedule $S$.}
    
    \For{$\phi\in\Phi$}{
        $S = \text{\textsc{RoundRobinBonus}}(\phi,S)$\;
    }

  }

  \Fn{\textsc{ComputeSchedule}}{
    \Input{Set of Filling classes $\Phi$}
    \Output{Compiled Schedule $\tilde{S}$}
    Initialise new empty schedule $S$\;

    Set $S\leftarrow $ \textsc{MinimalAllocation}$(\Phi, S)$\;
    Set $S\leftarrow $ \textsc{BonusAllocation}$(\Phi, S)$\;
    }

    \caption[Overall scheduling algorithm for creating network schedules.]{
      Overall scheduling algorithm for creating network schedules.
    \textsc{DirectAllocation} and \textsc{RoundRobinBonus} are respectively as in Algorithms~\ref{alg: direct allocation subroutine} and \ref{alg: bonus round robin resource schedule} in Appendix~\ref{app: Network Scheduler Algorithms}.
    }
    \label{alg: Compute Schedule}

\end{algorithm}

The Network Scheduler computes the network schedule for all components of the network, including end nodes. 
The schedules of network components are aligned such that a PGA for the PGT $\gamma$ with path $\pi_{\gamma}$ is added to the schedule of every component on the path $\pi_{\gamma}$, so that the PGA covers the same span of time in each component schedule. Due to this alignment, it suffices to refer to a single network schedule S, although the schedule is actually a set of individual schedules for each component of the network. 

The Compute Schedule process is described by the function \textsc{ComputeSchedule} in Algorithm~\ref{alg: Compute Schedule}.
It is a time-critical process that must compute a schedule including a minimal allocation of PGAs (Definition \ref{def: minAlloc}) for every accepted PGT once per scheduling interval. 
This is essential for satisfying the service agreements established with accepted demands (Definition~\ref{def: ServiceAgreement}).
Hence, the computation time must not exceed the scheduling interval. 
A secondary aim of the Compute Schedule process is to minimize the amount of idle time in which no PGAs are scheduled on internal resources.
It is divided into two phases: a \textit{minimal allocation} phase which is time-critical and must be completed by the end of a scheduling interval, and a subsequent \textit{bonus allocation} phase which can be cut off at any time without consequences for service agreements. 

The minimal allocation phase relies on \textsc{DirectAllocation}, Algorithm~\ref{alg: direct allocation subroutine} in Appendix~\ref{app: Network Scheduler Algorithms}, which produces a schedule for one filling class at a time. The \textsc{DirectAllocation} algorithm is fundamentally based on a modification of the classical Round Robin scheduling algorithm. 
A schedule $S$ is constructed by calculating a \textit{direct allocation} schedule for each filling class ${\phi \in \Phi}$ which includes a minimal allocation of PGAs for every PGT ${\gamma \in Z_{\phi}}$, and sequentially inserting the direct allocation schedules into $S$ according to a partial ordering of ${\phi \in \Phi}$. 
A partial ordering of the filling classes ${\phi \in \Phi}$ is specified in an implementation and derives from the path partition $\Pi$ which determines the structure of $\Phi$. 
The same logic underpins both the \textsc{DirectAllocation} scheduling algorithm and the \textsc{CalculateRequiredTime} function (Algorithm~\ref{alg: Calculate Required Time}) from the Admit Tasks process. \textsc{CalculateRequiredTime} essentially calculates the duration of a schedule that would be produced by the \textsc{DirectAllocation} algorithm, which actually adds PGAs to a network schedule. 
In this way, Admit Tasks leverages insight into the Compute Schedule process when determining which PGTs to admit.

The bonus allocation phase relies Algorithm~\ref{alg: bonus round robin resource schedule},~\textsc{RoundRobinBonus} of Appendix~\ref{app: Network Scheduler Algorithms}, which is also based on a modification of the classical Round Robin scheduling algorithm.
In this phase Compute Schedule passes again through the schedule constructed during the minimal allocation phase, and wherever possible schedules additional PGAs. 
This decreases the idle time of each resource in the network, as well as decreasing the expected time to generate all $\ninst_{d}$ packets of each accepted demand $d$. 

An example network schedule produced by \textsc{ComputeSchedule} is illustrated in Figure~\ref{fig: example network schedule}. This schedule is for a simplified version of the dumbbell topology network illustrated in Figure~\ref{fig: example network topology}, with six end nodes instead of fifteen. There are seven active PGTs in the example, configured to have different values for their minimal allocations, minimum separations between PGAs, and PGA durations. The solid (hatched) colored boxes represent PGAs scheduled during the minimal (bonus) allocation phase. 
The \textsc{DirectAllocation} schedule for each filling class has a block like structure and may include gaps of time in which no PGAs are scheduled. These gaps of time carry forward into the minimal allocation part of the schedule $S$ built from concatenated direct allocation schedules. In the bonus allocation phase, additional PGAs are scheduled in some of these gaps, improving the resource utilization and decreasing idle time.

The time complexity of the Compute Schedule process motivates the division of the process into the minimal and bonus allocation phases. In Appendix~\ref{app: complexity proofs} we prove that an implementation of the minimal allocation phase has operational complexity ${O(NR) + O(N^2)}$, where $N = |Z_{\Phi}|$ is the number of active PGTs and $R=|\mathfrak{R}|$ is the number of resources in the network. In contrast, the bonus allocation phase has operational complexity $O(N^3 R)$.
In the two phase system implemented by \textsc{ComputeSchedule}, the bonus allocation phase can be cut off by the end of the scheduling interval without impacting demand service agreements. When the scheduling offset is set to two scheduling intervals there is an entire scheduling interval to distribute schedules (as in Figure~\ref{fig:network scheduling timings}) and it is possible to cut off schedule computation arbitrarily close to the end of a scheduling interval.

\subsection{Schedule Manager}
The Schedule Manager control application stores the current and subsequent network schedules and formats the component schedules for distribution. 
It offers two interfaces, the \textit{schedule distribution} (SM:SD) interface and the schedule (SM:S) interface. 
The SM:SD interface is an internal interface which allows the Network Scheduler to push new schedules to the schedule manager as they are computed. 
The SM:S interface is an external interface which allows network components to request the current or next schedule. The Schedule Manager consumes the Z interface with the central controller, enabling it to push network schedules for distribution to components.

\section{Performance Analysis}\label{sec: Performance Analysis}

Upholding the service agreement that Arqon establishes for each accepted demand (Definition \ref{def: ServiceAgreement}) is essential to ensure the delivery of reliable service. 
Here we develop the performance analysis of \arqon, by stating and proving results which qualify the operating conditions under which Arqon is guaranteed to satisfy these service agreements.
As a supplement to these results, we motivate the connection between satisfaction of service agreements and successful execution of quantum network applications by end nodes. We prove that satisfaction of a service agreement for a demand $d$ with service error parameter $\epsilon^{\text{service}}_d$ translates into end nodes successfully generating all packets of pairs required to execute their targeted application with probability at least $(1- \epsilon^{\text{service}}_{d})$.

We define \textit{relevant} scheduling intervals as those that are relevant to the satisfaction of the service agreement. 

\begin{definition}[Relevant scheduling intervals]
    The \emph{\textbf{relevant}} scheduling intervals for an accepted demand $d$ are those between $t^{\texttt{\emph{start}}}_{d}$ and the scheduling interval in which the demand is either terminated or expires.
\end{definition}

\begin{theorem}[Deterministic Satisfaction of Service Agreements]
\label{thm: QoSAcceptedTasks}
    Let $d$ be a demand accepted by Arqon as the PGT $\gamma_d$. Suppose the path partition $\Pi$ remains static in all relevant scheduling intervals for $d$ and that every component on the path $\pi_{\gamma_d}$ retrieves each network schedule before its start time. Then, Arqon deterministically satisfies the service agreement for demand $d$.
\end{theorem}

To prove Theorem~\ref{thm: QoSAcceptedTasks}, which is our main performance result, we first state the following theorem, a proof of which is given in Appendix~\ref{app: performance analysis proofs}. The proof of the following theorem is non-trivial and relies upon many intermediate results, each of which is stated and proved in Appendix~\ref{app: performance analysis proofs}. The statement requires the set of filling classes to be \textit{well-behaved}. This property, defined formally in Appendix~\ref{app: performance analysis proofs}, means that a method of comparison between filling classes can be defined consistently based on the resources associated with the filling classes. 

\begin{restatable}[Minimal Allocation is guaranteed]{theorem}{MinAllocGuaranteed}\label{thm: min alloc guaranteed}
Let $d$ be a demand accepted by Arqon as the PGT $\gamma_d$. Suppose ${\pi_{\gamma_d} \in \Pi}$ in every relevant scheduling interval. Suppose $\Phi$ is a well-behaved set of filling classes. Then, in every relevant scheduling interval the schedule ${S= \textsc{ComputeSchedule}(\Phi)}$ contains at least $\minimumallocation_{\gamma_d}$ PGAs of the PGT $\gamma_d$. Moreover, the duration of $S$ does not exceed $T^{SI}$.
\end{restatable}

\begin{proof}[Proof: Theorem \ref{thm: QoSAcceptedTasks}]
 To satisfy the service agreement it is necessary that \begin{enumerate}
    \item The schedule is computed before the end of each relevant scheduling interval;
    \item The schedule is retrieved successfully by every component on $\pi_{\gamma_d}$ in every relevant scheduling interval;
    \item The schedule contains at least $\minimumallocation_{\gamma_d}$ PGAs of the PGT $\gamma_d$. 
    \item The duration of the schedule $S$ does not exceed the duration of one scheduling interval $T^{SI}$.
\end{enumerate}
Property 4. ensures that no schedules overruns the scheduling interval, which would prevent complete execution of the schedule for the subsequent scheduling interval. 
In the hypothesis, it is assumed that 1. and 2. hold. Properties 3. and 4. are guaranteed by Theorem~\ref{thm: min alloc guaranteed}, a full proof of which is given in Appendix~\ref{app: performance analysis proofs}. 
\end{proof}


An application session succeeds if $\ninst$ instances of the quantum application are successfully executed by the end nodes which register it. For this to be possible, at least $\ninst$ packets of end-to-end links must be successfully generated. In this way, application success hinges upon the receipt of minimal service.

\begin{definition}[Minimal service]\label{def: min service} 
  Recall from Definition~\ref{def: demand} that a demand $d$ from an application session on end nodes specifies at least $N^{\texttt{\emph{inst}}}_{d}$ packets to be generated.
  We say that an application session achieves \emph{\textbf{minimal service}} if at least $N^{\texttt{\emph{inst}}}_d$ packets of end-to-end links are successfully generated before the demand's expiry time. 
\end{definition}


\begin{lemma}\label{lemma:ServiceAgreementMeansMinService}
    Let $d$ be a demand accepted as a PGT $\gamma_d$ for which Arqon satisfies the service agreement. Suppose all nodes on the path $\pi_{\gamma_d}$ attempt all scheduled PGAs. Then, the application session which submitted $d$ achieves minimal service with probability at least ${(1-\epsilon^{\text{\emph{service}}}_{d})}$. 
\end{lemma}

\begin{proof}(Lemma \ref{lemma:ServiceAgreementMeansMinService})\label{proof:minService}
    By the hypothesis, the service agreement for demand $d$ is satisfied and all nodes on the path $\pi_{\gamma_d}$ of the PGT $\gamma$ realizing the demand attempt all scheduled PGAs.
    Let $x$ be the total number of packets generated for $d$ as a result of the PGAs attempted. There are two possible scenarios in which a service agreement is satisfied: the demand is either terminated or expires. We are only concerned with demands that expire. If $d$ was terminated, then we may assume that either $\ninst_d$ packets were successfully generated and the application session achieved minimum service, or the application session instance was canceled for some reason unknown to the central controller -- in which case the service agreement is also terminated.
    It follows from the Definition~\ref{def: minAlloc} of the minimal allocation that if $\minimumallocation_{\gamma_d}$ PGAs are attempted in each of ${n^{SI} = \lceil (t^{\texttt{expiry}}_d - t^{\texttt{start}}_d)/T^{SI} \rceil}$ scheduling intervals, then the number $x$ of successfully generated packets satisfies ${x \geq \ninst_d}$ with probability at least ${(1- \epsilon^{\text{service}}_{d})}$.
\end{proof}

\section{Implementation}\label{sec: Implementation}
We have developed a particular implementation of Arqon written in Python~\cite{arqon-sim} that is designed to enable configurable simulations of realistic networks.


We define a family of network topologies that adheres to a set of physically motivated principles and limit the scope of our implementation to this family of networks. 
Crucially, our implementation includes the specification of a set of allowed paths $\mathcal{P}_{\text{allowed}}$, a specification of the path partition $\Pi$, and the specific choice of filling classes $\Phi$ induced by $\Pi$. The specific choice of filling classes thus defined has the property of being well-behaved, ensuring compatibility with the performance results in Section~\ref{sec: Performance Analysis}.

To enable numeric evaluations of heterogeneous networks from within this family, with a diverse range of entanglement generation capabilities, we implement a method of generating a network capabilities table which relies on input parameters for mean and standard deviations of the entanglement generation rate and fidelity on fixed length paths, and an scaling calculation based on the length of end-to-end paths. 

The full details of our implementation are provided by our open source simulator and its documentation~\cite{arqon-sim}. Here we highlight key features, particularly those that are relevant to the numeric evaluations in Section~\ref{sec: Evaluation}.


\subsection{End Nodes}
Our implementation is not intended to provide a detailed simulation of end nodes. However, we implement a simple physically motivated model for assigning applications to end nodes and configuring application demands, and we implement demand submission and termination behaviors of end nodes because they are critical for simulating Arqon.
In a deployed network, user input will determine which applications an end node runs.
We define a set of realistic quantum network applications and implement a method of assigning these applications to end nodes~\cite{arqon-sim}. 

\subsection{Applications}
We define a dictionary of possible applications~\cite{arqon-sim} based on minimum fidelity and number of entangled pair requirements of Quantum Key Distribution (QKD)~\cite{e91, BB84}, Blind Quantum Computing (BQC)~\cite{bqc1, bqc2}, quantum teleportation~\cite{TeleportationOriginal, HorodeckiFidelityTeleportation}, or any of these applications combined with entanglement purification~\cite{DEJMPS,BBPSSW}. 
We derive the minimum entanglement fidelity and number of pairs requirements from resources~\cite{FeasableParamsQKD, pairsPaperBethany, HorodeckiFidelityTeleportation, HorodeckiFidelityTeleportation, DEJMPS}. 
We treat teleportation type applications as \textit{test-applications}, which may be used to validate near term quantum network deployments. For this reason, we simply require an entanglement fidelity of greater than $0.5$ for teleportation, which is the classical bound.

\subsubsection{Assigning Applications}

Each application we define in~\cite{arqon-sim} specifies critical information required to produce a valid demand.
In particular this includes the (minimum) window duration $w$. 
For real end nodes, the maximum memory lifetimes of qubits will affect which application programs can be supported, as these set hard upper bounds for window durations. 
Application programs don't have an inherent window. 
In a deployed setting, the local runtime of an end node must set a window based not only on qubit memory lifetimes but also on the duration of time in which a generated packet must remain in memory when executing an instance of the application program. 
We define a minimum window for an application based on an estimate of how long exemplary quantum hardware may require to execute an application instance.  
We model the quantum hardware of end nodes based on either the NV center in diamond~\cite{NVSeminal, NV_as_network_node_and_processor, NV_multiqubit_Seminal, ThreeNodeQN} or trapped ion platforms~\cite{IonsSeminal1, IonsPhotonEnt, IonsHighRateEnt, CoherenceTrappedIonQIA}.
We use data from~\cite{ThreeNodeQN} for the NV platform and from~\cite{CoherenceTrappedIonQIA} for the trapped ion platform to set the memory lifetimes of individual end nodes to currently experimentally achievable lifetimes in a networked setting (i.e. memory lifetime while simultaneously attempting entanglement generation). 
Then, when randomly assigning applications to end nodes, we ensure that the memory lifetime of the end node can support the requirements of the application. 

In a deployment, some end nodes may be operated as server nodes, for example for serving BQC applications, where less powerful end nodes may be the requesting clients. 
To support such use cases, we implement a registration of end nodes as either \textit{discoverable} or not. 
Practically, a server should be discoverable. 
When assigning applications, we use use this discoverable property as a proxy for a node being a server or a client.
We only assign applications between two clients, or between a client and a server. 

By mocking these constraints when assigning applications to end nodes, we are effectively replicating some of the behaviors of \textit{capability negotiation}, without fully implementing such a protocol. 

\subsubsection{Demand Submission}
For each application assigned to an end node, we draw the initial demand submission time from a ${\mathrm{Uniform}[0,\texpiry_{\text{rel}})}$ distribution, where ${\texpiry_{\text{rel}}=t^{\texttt{start}}+\texpiry}$ is the \textit{relative} expiry time. 
This ensures that the initial demand submissions are not all clustered together around the start of a simulation and mocks the demand submission process by independent end nodes. 
Thereafter, whenever a demand obtains minimal service (and is then terminated), expires or is removed from the network, a new demand is submitted time ${\mathrm{Exponential}(t^\texttt{resubmit})}$ later, where the value of ${t^\texttt{resubmit}}$, the average resubmission time, is set separately for each application source assigned to an end node. 

\subsection{Creating Packet Generation Tasks}

To create a PGT $\gamma_d$ for a demand $d$ there needs to be a method of calculating $E_{\gamma_d}$, the execution time of a PGA for $\gamma_d$, $\ppacket_{\gamma_d}$, the packet generation probability for a PGA of duration $E_{\gamma}$, and $\minimumallocation_{\gamma_d}$, the minimum allocation of PGAs for $\gamma_d$. These calculations depend on the demand parameters, the service error parameter $\epsilon^{\text{service}}_d$, and the rate of successful entanglement generation along the path $\pi_{\gamma_d}$ of $\gamma_d$.
In our implementation $\epsilon^{\text{service}}_d$ is a configurable parameter set by the central controller and the same value is used for every demand $d$.

To calculate possible values of $E_{\gamma_d}$ for a PGT $\gamma_d$ we use approximations by Naus from \cite{naus_approximations_1982} for the probability of $k$ successes in a window $w$ given a rate of average successes $\lambda$. For a specified rate of entanglement generation success, these approximations are used to calculate possible pairs $E_{\gamma_d}$ and $\ppacket_{\gamma_d}$. 
To calculate $\minimumallocation_{\gamma_d}$ for a PGT $\gamma_d$ as in Definition \ref{def: minAlloc} we use Hoeffding's inequality \cite{HoeffdingInequality}. This calculation depends on the value of $\ppacket_{\gamma_d}$.
See Appendix~\ref{app: Implementation} for more information about the calculation of $E_{\gamma}$ based on $\ppacket$, and~\cite{arqon-sim} for the Python implementation of these calculations. 

We set the final values for $E_{\gamma}, \ppacket_{\gamma},$ and $\minimumallocation_{\gamma}$ based on the solution of \begin{equation}\label{eq: ppacket selection}
    \ppacket_{\gamma_d} = \argmin_{p}\minimumallocation_{\gamma_d}(p)\bigg(E_{\gamma_d}(p)+\tminsep_{\gamma_d}\bigg).
\end{equation}
The quantity minimized in \eqref{eq: ppacket selection} captures a notion of how much load will be placed on the network by accepting the PGT. The calculation therefore optimizes the selection of $\ppacket$ for the value that minimizes the load. 

\subsection{Network Topology}\label{ssec: implementation: network topology}

We specify a set of physically motivated constraints on the network resource graph which define a family of networks.
These constraints help to ensure that the network topology of any network we simulate represents a physically realizable quantum network. 
In a deployment of Arqon these constraints will be replaced by simple input of the actual network topology and the logical connections present in the network. 
If a deployed network is also an element of our broad family of networks, then it may be equipped with $\mathcal{P}_{\text{allowed}}$ and $\Pi$ as in our implementation.

\subsubsection{Network Resource Graph}\label{sssec: implementation: NM: NRG}
The first constraint we impose is to consider only internally connected network resource graphs (Definition~\ref{def: internally connected}). 
Every network component ${v \in \mathcal{V}}$ is uniquely identified as either an end node ($v\in E$), an EGI ($v\in I$), a junction node ($v \in J$), or a long-distance backbone ($v \in B$). Therefore, the set of vertices $\mathcal{V}$ in the network resource graph is a disjoint union of component types, \begin{equation}\label{eq: implementation NRG vertices are disjoint over resource types}
    \mathcal{V} = E \sqcup I \sqcup J \sqcup B.
\end{equation}

To further define a family of internally connected networks we introduce a constraint on the $\binom{\mathcal{V}}{2}$ possible edges to a set of allowed edges, denoted $\mathcal{E^{*}}$.
The definition of allowed edges is based on three principles: 
\begin{enumerate}
    \item An EGI is always needed for entanglement generation between two nodes with entanglement generation capabilities \cite{RemoteEntProtocolsQubitswithPhotonicInterfaces, REGPs2}.
    \item Entanglement swapping must always be mediated by a node that has at least one qubit with entanglement generation capabilities and an additional qubit with some memory capabilities \cite{theoryHeraldDLCZ,repeatersProtocols1,ThreeNodeQN, vardoyan2020exact}.
    \item Long-distance backbones always own a private set of EGIs, both internally to the long-distance backbone, as well as at the end points where they are exposed to other network components \cite{theoryHeraldDLCZ, RepeatersTrappedIon1, RepeatersAE1, RepeatersPhotonic1, CoherenceTrappedIonQIA}.
\end{enumerate}  


The set of edges $\mathcal E$ in a network resource graph must be a subset of the set of allowed edges, $\mathcal{E} \subseteq \mathcal{E}^{*}$.
Informally, the following types of edges are allowed: $(e,i)$, $(i, j)$, $(j, b)$, where $e \in E, \ i \in I, \ j \in J, \ b \in B$ and allowed edges are reflexive, so that $(e, i)$ is an allowed edge means that also $ (i, e)$ is an allowed edge. A formal definition of the set $\mathcal{E}^{*}$ of allowed edges is stated in Appendix \ref{app: Implementation}.
For details on generating random network topologies from within this family of networks, see~\cite{arqon-sim}.

\begin{restatable}[Local Area]{definition}{LocalArea}\label{def: local area}
    Let $\mathcal{G}=(\mathcal{V}, \mathcal{E})$ be a network resource graph, and $B$ the set of long-distance backbones. Let ${\mathcal E_{B}= \{(e_1, e_2) \in \mathcal E \ | \ e_1, e_2 \in B \}}$ and ${\overline{\mathcal{E}_{B}} = \mathcal E \setminus \mathcal{E}_B}$.
    Let \begin{equation}
        \overline{\mathcal{G}_B} := (\mathcal V \setminus B, \overline{\mathcal{E}_B})
    \end{equation}
     be the subgraph of $\mathcal{G}$ formed by removing every long-distance backbone vertex and corresponding edges.
    Then the \emph{\textbf{local areas}} of $\mathcal{G}$ are the (maximal) internally connected sub-components of $\overline{\mathcal{G}_B}$. 
    The set of local areas is denoted $\mathcal{L}$.
\end{restatable}

The concept of a local area is introduced to aid in describing the scale of a quantum network. Local areas induce a structure on the network resource graph such that it can be decomposed into the graph $\overline{\mathcal{G}_{B}}$ and a graph $\mathcal{G}_C$ which contains only local areas and long distance backbones. The graph $\mathcal{G}_{C}$ captures the local area connectivity of $\mathcal{G}$.

\subsection{Path Computation}\label{sssec: implementation: path computation}

The set of allowed paths ${\mathcal{P}_{\text{allowed}}\subset \mathcal{P}_{\text{valid}}}$ is a restriction on the set of valid entanglement generation paths. 
In Appendix~\ref{app: Implementation} we formally define $\mathcal{P}_{\text{allowed}}$ following a construction which minimizes the number of long-distance backbones traversed, rather than strictly minimizing the absolute number of hops in a path. 
This construction is motivated by the expectation that a relatively large loss of entanglement generation success rate will result from traversing of a long-distance backbone, as compared to traversing multiple hops within a local area~\cite{DeliveryTimeCoopmans, CoherenceTrappedIonQIA, RepeatersTrappedIon1}. 

\begin{restatable}[Specific choice of path partition]{definition}{SpecificPathPartition}\label{def: specific path partition}
Let ${\mathcal G = (\mathcal{V}, \mathcal{E})}$ be an internally connected resource graph and let ${\mathcal{E} \subseteq \mathcal{E}^{*}}$.
Let ${\mathcal{L} = \{L_i = (\mathcal V_i, \mathcal E_i)\}}$ be the set of local areas of $\mathcal{G}$.
Let ${J_i = J \cap \mathcal V_i}$ be the junction nodes in sub-graph $L_i$. Denote by ${\pi = (\pi_1, \pi_2, \cdots)}$ a path $\pi$ with vertices ${\pi_i \in \mathcal{V}}$.

We define the partition $\Pi$ of $\mathcal P_{\text{\emph{allowed}}}$ in the following manner:
\small
    \begin{align*}
        \Pi^B &= \{\pi \in \mathcal P_{\text{\emph{allowed}}} : \exists k~s.t.~\pi_k\in B\} \\
        \Pi^J_j &= \{\pi \in \mathcal{P}_{\text{\emph{allowed}}}\setminus\Pi_b : \exists k~s.t.~\pi_k \in J_j\}, \phantom{xxxxxxx}  \ \forall  j \in J\\
        \Pi^I_i &= \{\pi \in \mathcal{P}_{\text{\emph{allowed}}}\setminus(\Pi_b\cup \bigcup_{j \in J}\!\Pi^J_j) : \exists k~s.t.~\pi_k = i\}, \phantom{x} \forall i \in I,
    \end{align*} \normalsize
    where $\mathcal{P}_{\text{\emph{allowed}}}$ is as defined Appendix~\ref{app: Implementation}, Definition~\ref{def: allowed entanglement generation paths}. 
\end{restatable}

\begin{restatable}{lemma}{SpecificPathPartitionDisjoint}\label{lemma: specific path partition disjoint}
\begin{equation}
    \Pi = \{\Pi^B\}\cup\{\Pi_j\}_{j\in J}\cup\{\Pi_i\}_{i\in I},
\end{equation}
where ${\Pi^B, \ \Pi^J_j \ \forall j, \text{ and } \Pi^I_i \ \forall i}$ are as in Definition~\ref{def: specific path partition}, is a disjoint partition of ${\mathcal{P}_{\text{\emph{allowed}}}}$.
\end{restatable}

\begin{proof}
    See Appendix \ref{app: Implementation}.
\end{proof}

\subsection{Network Capabilities}
The network capabilities manager populates the network capabilities table using a simple heuristic model when the network is initialized. Our implementation does not include dynamic update handling for the capabilities of network components or paths through the network. See~\cite{arqon-sim} for full details of the model.

\subsection{Filling Classes}

The properties of the set of filling classes are directly inherited from the specific choice of path partition, which defines their structure (Definition~\ref{def: specific path partition}). The specific choice of filling classes can be written as
\begin{equation}\label{eq: specific FC partition}
    \Phi = \{\Phi^{B}\} \cup \{\Phi^J_j\}_{j \in J} \cup \{\Phi^I_{i}\}_{i \in I}.
\end{equation}
In Appendix~\ref{app: Implementation} we prove that this specific choice of filling classes has the property of being well-behaved.

\subsection{Execution of Network Schedules}
\label{subsec: simulating minimal service}
Since we do not implement a detailed model of end nodes, in our implementation we mock the execution of network schedules by sampling random variables to determine how many PGAs successfully generate packets.
We then assume that whenever a packet is generated as the result of a PGA, then the corresponding application instance is successfully executed. 
Our implementation does include session termination behavior in which end nodes terminate application sessions after they have obtained minimal service.

For any network schedule in which a PGT $\gamma$ cannot reach the termination criteria of obtaining minimal service (i.e. less PGAs are scheduled than the remaining instances to execute), we determine the number of packets generated by drawing from a Binomial distribution with probability of success $\ppacket_{\gamma}$.
Otherwise, we sample from a Bernoulli distribution with probability of success $\ppacket_{\gamma}$ for each PGA in turn.
This two-stage approach allows us to minimize the time required to simulate execution of network schedules, while still allowing accurate reporting of when minimal service is obtained. 


\section{Complexity Analysis}\label{sec: Complexity}
Arqon must produce a new network schedule every scheduling interval, with a total computation time not exceeding the scheduling interval.  
To validate that this is feasible, we analyze the operational complexity of our implementation of the Network Scheduler control application. 
Our operational complexity analysis is based on determining which fundamental operations, such as value assignment, list insertion, list appending, list sorting, etc. occur in the implementations of our algorithms, as in Listings~\ref{lst: directallocation}-\ref{lst: resource schedule} in Appendix~\ref{app: code listings}. For each fundamental operation, such as value assignment, we determine the time complexity of the operation, for example: value assignment is O(1). We therefore calculate the time complexity of our algorithms based on the number of operations of each type, and the time complexity of each type of operation. 
The details of our analysis are given in Appendix~\ref{app: complexity proofs}. 

\begin{restatable}{theorem}{thmOverallComplexity}\label{thm: overall complexity} 
The complete program \textsc{NetworkScheduler} has operational complexity \begin{equation}\label{eq: overall complexity}
   O \bigg{(}k 2^{R} \big{(} (N + k)^2 +R \big{)} + (N+R)^2  + N^3 R \bigg{)},
\end{equation}
where ${R = |\mathfrak{R}|}$ is the number of internal resources in the network, ${N=|Z_{\Phi}|}$ is the number of active PGTs, and ${k=|\bm{\Gamma}|}$ is the number of PGTs in the task intake object.
\end{restatable}
\begin{proof} 
    The program \textsc{NetworkScheduler} consists of the program sequence \textsc{UpdateFillingClasses}, \textsc{AdmitTasks}, and \textsc{ComputeSchedule}.
    By Lemma~\ref{lemma: UFC complexity} \textsc{UpdateFillingClasses} has operational complexity \begin{equation}\label{eq: UFC complexity main text}
        O((N+R)^{2}).
    \end{equation} 
    By Theorem~\ref{thm: admit new tasks complexity} \textsc{AdmitTasks} has operational complexity \begin{equation}\label{eq: admit tasks complexity main text}
        O \bigg{(}k 2^{R} \big{(} (N + k)^2 +R \big{)} \bigg{)}.
    \end{equation}
    The program \textsc{ComputeSchedule} consists of a minimal allocation phase, a bonus allocation phase, and a schedule compilation phase. The minimal allocation phase contributes (Lemma~\ref{lemma: single FC direct allocation})  \begin{equation}\label{eq: minimal allocation complexity main text}
        O(NR) + O(N^2)
    \end{equation} to the operational complexity.
     The bonus allocation phase contributes (Lemma~\ref{lemma: single FC Bonus RR}) \begin{equation}\label{eq: bonus round robin complexity main text}
         O(N^3 R).
     \end{equation}
     Compiling the resulting schedule contributes only $O(NR)$ (Lemma~\ref{lemma: compile schedule complexity}). 
    The total operational complexity of \textsc{ComputeSchedule} is therefore dominated by the bonus allocation phase and is given by Theorem \ref{thm: compute NS complexity} as \begin{equation}\label{eq: compute schedule complexity main text}
        O(N^3 R).
    \end{equation}
    Overall, the program \textsc{NetworkScheduler} has operational complexity, \begin{align*}
         O \bigg{(}k 2^{R} \big{(} (N + k)^2 +R \big{)} + (N+R)^2 \bigg{)} + O(N^3 R). 
    \end{align*}
\end{proof}

Our implementation of the Network Scheduler has an operational time complexity with a polynomial dependency on the number of active PGTs $N$ and the size of the task intake buffer $k$. 
However, there is an exponential dependence on the number of resources in the network, $R$. 
In a deployment, the number of resources is expected to be relatively static, and will typically act as a constant factor in the operational complexity of the Network Scheduler. 
Nevertheless, the exponential dependence on $R$ does indicate that Arqon is best suited for networks in which $R < 1000 $. This result conforms with the heuristic expectation that distributed control architectures are more suitable than centrally controlled network architectures for very large networks~\cite{DistributedNetworks}.

\section{Evaluation} \label{sec: Evaluation}
Arqon is designed to deliver reliable service, according to reliability requirements \ref{req: 1: responses}-\ref{req: 4: independent}. 
The first requirement, \ref{req: 1: responses} is satisfied by design, as all demands receive accept/reject messages from the demand manager within two scheduling intervals following their registration~\cite{arqon-sim}. 
By Theorem~\ref{thm: overall complexity}, which captures the overall operational complexity of the Network Scheduler control application, it is expected that network schedules can always be computed in time to be delivered.
Then, by Theorem~\ref{thm: QoSAcceptedTasks}, it is expected that Arqon  satisfies requirements~\ref{req: 2: satisfied}-\ref{req: 4: independent} as long as the path partition $\Pi$ remains static throughout an evaluation.
However, from the perspective of a network operator, there remain several relevant questions regarding the manner in which reliable service is delivered, including the following. \begin{enumerate}[label=\textit{(Q\arabic*)}]
    \item \label{eval Q: e service}What effect does the service error parameter $\epsilon^{\text{service}}$ have on the service delivered? 
    \item \label{eval Q: relative service delivery time}How quickly does Arqon satisfy service agreements, relative to demand expiry times? 
    \item \label{eval Q: bonuns round robin impact}What impact does the bonus allocation phase of the Compute Schedule process of the Network Scheduler have on the service delivered?
    \item \label{eval Q: prop accepted}What proportion of submitted demands are accepted? 
    \item \label{eval Q: choice of FCs}What impact does the specific choice of filling classes \eqref{eq: specific FC partition} have on which of the registered PGTs are accepted by Admit Tasks?
\end{enumerate}

In this section we first validate that Arqon delivers reliable service on static network topologies and then address these questions, which capture additional performance metrics. 
Finally, we validate our operational complexity analysis of the Network Scheduler with thorough numeric evaluations that supply concrete time values for computation times in a variety of parameter regimes.

In all simulations, the demands for service submitted by pairs of end nodes are based on the real minimum fidelity requirements of BQC (2, 6, or 10 qubit) \cite{pairsPaperBethany, HorodeckiFidelityTeleportation}, e91 based QKD \cite{FeasableParamsQKD}, teleportation \cite{HorodeckiFidelityTeleportation},
or any of these applications pre-pended by purification following the DEJMPS protocol \cite{DEJMPS} of the bipartite entangled pairs. The other parameters of each demand, such as the number of pairs and the window duration are set based on the application and the network capabilities database for that simulation~\cite{arqon-sim}.
Lists of the applications simulated in each simulation are included as JSON files in~\cite{arqon-sim}.

\subsection{Service Reliability}
\label{subsec: eval Demand Satisfaction}
\begin{figure}
    \centering
    \includegraphics[width=0.98\linewidth, trim={0.5cm 0.25cm 0.0cm 0.0cm}]{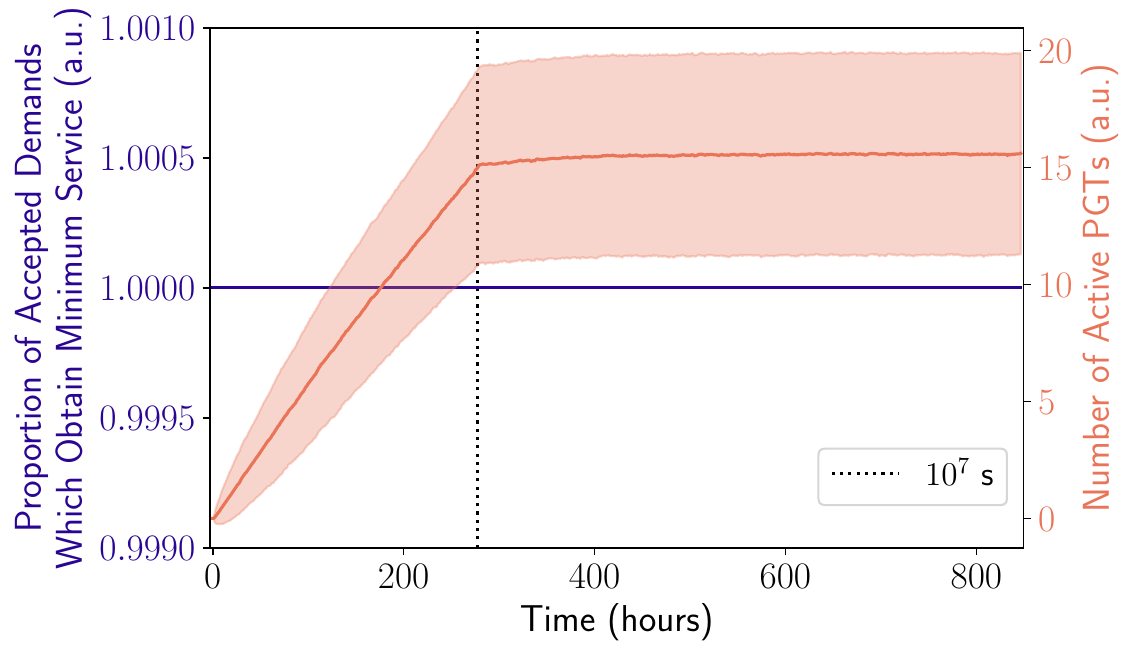}
    \caption{Proportion of accepted demands which obtain minimal service (left, blue) and the number of active PGTs (right, orange), compared to the time at which the scheduling interval occurs. The scheduling interval duration 30 minutes and the data corresponds to a simulated 10 weeks of continuous network operation. Shaded regions show 1 standard deviation of the mean.}
    \label{fig: satisfaction service agreements}
\end{figure}

To meet the reliability requirements \ref{req: 2: satisfied}-\ref{req: 4: independent}, Arqon must satisfy demands before their deadlines and the arrival of new demands must not disrupt service to already accepted demands. A demand is satisfied when its associated application session achieves minimal service (Definition \ref{def: min service}). 
By Lemma ~\ref{lemma:ServiceAgreementMeansMinService}, if the service agreement for a demand $d$ is satisfied, then the corresponding application session achieves minimal service with at least probability ${(1-\epsilon^{\text{service}}_{d})}$.

We validate that Arqon delivers reliable service first with simulations of a network with the dumbbell topology illustrated in Figure~\ref{fig: example network topology}, and with $\epsilon^{\text{service}}=1e-5$ for all demands. Then, we validate that the reliability of service delivery is not conditional on a specific network topology or number of demand sources by simulating random network topologies from the family of networks defined in Section~\ref{sec: Implementation}.

\subsubsection{Dumbbell Topology Network}
We simulate one thousand randomly seeded runs, each corresponding to 846 hours (1694 half hour scheduling intervals) of continuous operation of the dumbbell topology network in Figure~\ref{fig: example network topology}. In every run, applications are assigned to ${80\%}$ of all possible end node pairings connected by a route in ${\Pi^I_i \ \forall i \in I}$, ${10\%}$ of all end node pairings connected by a route in ${\Pi^J_j \ \forall j \in J}$, and ${10 \%}$ of all end node pairings connected by a route in $\Pi^B$. As long as at least one of the nodes assigned the application is `not discoverable' (see Section~\ref{sec: Implementation}), then the application session submits demands for service. These simulations included an average of 18 demand sources, with the (static) number per simulation ranging from 3 to 26.

Figure~\ref{fig: satisfaction service agreements} shows how two metrics vary over time for the dumbbell topology simulations: (1) the proportion of accepted demands that obtain minimal service, and (2) the number of active PGTs at any given time. 
For the minimal service proportion metric, time corresponds to when each demand was accepted as a particular PGT. Across all one thousand simulation runs, a total of 180,662 accepted demands reached their expiry times.
With the stringent value of $\epsilon^{\text{service}}=1e-5$ used in these simulations, Lemma~\ref{lemma:ServiceAgreementMeansMinService} predicts that at most 2 accepted demands should fail to obtain minimal service. 
The simulation results confirm this expected behavior: all accepted demands obtained minimal service (proportion = 1.0, standard deviation = 0).These outcomes validate that the service delivered satisfies reliability requirements \ref{req: 2: satisfied} and \ref{req: 3: on time}. 

The variation in the number of active PGTs over time enables validating that
arriving demands do not disrupt service to already accepted demands~\ref{req: 4: independent}.
Starting from network initialization at time zero, pairs of end nodes begin submitting demands and some of these are accepted as specific PGTs by Admit Tasks.
In these simulations, the time at which different end node pairs submit their first demand is uniformly distributed over the first $10^7$ seconds of network operation. 
This process results in the number of PGTs to schedule increasing over time, before reaching a steady state average of 15.5 PGTs after about 278 hours ($10^7$ seconds) of continuous operation. 
The ramp-up in the number of PGTs to schedule has no impact on the proportion of demands obtaining minimal service, confirming that~\ref{req: 4: independent} is satisfied.


\subsubsection{Random Network Topologies}
\renewcommand{\arraystretch}{1.1}\small
\begin{table}[t]
    \centering
    \begin{tabular}{|c|c|c|}
    \hline
         \thead{\textbf{Long-distance} \\ \textbf{Backbones}} & \thead{\textbf{Local Areas}} & \thead{\textbf{End Nodes}} \\
         \hline
         1 & 2 & 15 \\
         2 & 2 & 15 \\
         2 & 2 & 50 \\
         2 & 3 & 30 \\
         2 & 3 & 50 \\
         3 & 3 & 30 \\
         5 & 4 & 40 \\
         6 & 3 & 35 \\
         7 & 5 & 50 \\
         12 & 4 & 40 \\
         \hline
    \end{tabular}
    \caption{Parametrization of the randomly generated topologies simulated.}
    \label{tab: random topology parameters}
\end{table}
\normalsize

The random topologies we simulate are parameterized in Table~\ref{tab: random topology parameters} by the number of long-distance backbones, the number of local areas, and the number of end nodes in the network. 
For each parameterization, we randomly generate ten network topologies and then, we simulate ten randomly seeded runs for every topology. Every run corresponds to 757 hours (1514 half hour scheduling intervals) of continuous operation. In every run we assign application sessions to ${20\%}$ of all possible end node pairings connected by a route in ${\Pi^I_i \ \forall i \in I}$, ${15\%}$ of end node pairings connected by a route in ${\Pi^J_j \ \forall j \in J}$, and ${5 \%}$ of end node pairings connected by a route in $\Pi^B$. 

To address~\ref{eval Q: e service} and study the impact of different values for the $\epsilon^{\text{service}}$ parameter, we perform each simulation with five different settings: \begin{equation*}
    \epsilon^{\text{service}} \in \{1{e-5}, \ 0.001, \ 0.01, \ 0.1, \ 0.5 \}.
\end{equation*}
In Table~\ref{tab: random topos reliability} we record the proportion of accepted demands which obtain minimal service, which is always greater than the lower bound of $(1-\epsilon^{\text{service}})$ predicted by Lemma~\ref{lemma:ServiceAgreementMeansMinService}.
These results validate that Arqon delivers reliable service for a wide range of network topologies and number of demand sources. 

\subsection{Additional Performance Metrics}
The average values of performance metrics addressing~\ref{eval Q: e service}-\ref{eval Q: bonuns round robin impact} are reported in Table~\ref{tab: random topos reliability}. To characterize the impact of the bonus allocation phase of Compute Schedule, we disabled this feature and repeated all simulations of randomly generated network topologies with the same random seeds as in the simulations with default Arqon.
The total proportion of accepted demands which obtain minimal service in these simulations was always ${>99\%}$, both with and without the bonus allocation phase. This outcome persisted even when $\epsilon^{\text{service}}$ was relaxed to the very high value of $0.5$, which Lemma~\ref{lemma:ServiceAgreementMeansMinService} predicts to allow up to ${50\%}$ of demands to fail to obtain minimal service. 

These results indicate that Arqon over-serves accepted demands, as compared to the service error parameter $\epsilon^{\text{service}}$. As this effect persists when the bonus allocation phase is disabled, we trace the cause of this excess service to how Arqon calculates the minimal allocation for each demand (Definition~\ref{def: minAlloc}), which has proven to be a large overestimate.
Application sessions on end nodes may benefit from over-service of accepted demands, but wider network performance metrics such as the proportion of accepted demands may suffer from over-service to demands.
In an alternate implementation, this may be addressed by replacing the method which solves the the system of equations determining the minimal allocation with a weaker solution that determines smaller values for the minimal allocation. 

Arqon schedules a minimal allocation of PGAs for all active PGTs in every scheduling interval, regardless of whether or not the expected number of successfully generated packets,\begin{equation}\label{eq: expected successful PGAs}
    \mathbb{E}[\# \text{ successful PGAs}] = N^{\text{PGAs scheduled}} \times \ppacket
\end{equation} already exceeds the number of required successes ($\ninst$). To prevent over-serving demands, an alternate implementation may decrease the minimal allocation for PGTs following the scheduling interval in which \eqref{eq: expected successful PGAs} exceeds $\ninst$. 

The data in Table~\ref{tab: random topos reliability} also addresses~\ref{eval Q: relative service delivery time} and~\ref{eval Q: bonuns round robin impact}. \begin{enumerate}[label=\textit{(Q\arabic*):}, start=2]
    \item Arqon satisfies service agreements for accepted demands significantly in advance of demand expiry times. This is true for all $\epsilon^{\text{service}}$ values simulated and remains true if the bonus allocation phase of Compute Schedule is disabled. 
    
    \item The bonus allocation phase significantly impacts how quickly an accepted demand obtains minimum service, relative to its expiry time. The mean service to expiry time more than doubles when the bonus allocation phase is disabled.
    
    Approximately half of all PGAs scheduled by Arqon are scheduled by the bonus allocation phase of Compute Schedule, regardless of the value of $\epsilon^{\text{service}}$. This proportion directly accounts for the speedup of the service to expiry time introduced by this phase.  
\end{enumerate}

\renewcommand{\arraystretch}{1.2}\small
\begin{table*}[t]
    \centering
    \begin{tabular}{|c|c|c|c|}
    \hline
    $\mathbf{\epsilon_{service}}$ &\thead{\textbf{Mean Proportion of Accepted Demands} \\ \textbf{Which Obtain Minimal Service}} & \thead{\textbf{Mean Service} \\ \textbf{to Expiry Time}} & \thead{\textbf{Mean Proportion of PGAs}\\\textbf{Scheduled in Bonus Allocation}} \\\hline
    \multicolumn{4}{|c|}{\footnotesize \textbf{\textit{Default Arqon: Bonus Allocation Enabled in Compute Schedule}} \small}\\
    \hline
    1e-05 & 1.0000 & 0.3706 & 0.4950 \\
    0.001 & 1.0000 & 0.3777 & 0.5279 \\
    0.01  & 1.0000 & 0.3809 & 0.5382 \\
    0.1   & 1.0000 & 0.3847 & 0.5386 \\
    0.5   & 0.9998 & 0.3837 & 0.5423 \\
    \hline
    \multicolumn{4}{|c|}{\footnotesize \textbf{\textit{with Bonus Allocation Disabled in Compute Schedule}} \small}\\
    \hline
    1e-05 &1.0000 & 0.7481 & 0 \\
    0.001 &1.0000 & 0.8194 & 0 \\
    0.01 &1.0000 & 0.8483 & 0 \\
    0.1 &0.9998 & 0.8639 & 0 \\
    0.5 &0.9956 & 0.8646 & 0 \\
    \hline
    \end{tabular}
    \caption{Additional performance metrics addressing~\ref{eval Q: e service}-\ref{eval Q: bonuns round robin impact} from simulations of the randomly generated network topologies parameterized in Table~\ref{tab: random topology parameters}.}
    \label{tab: random topos reliability}
\end{table*}
\normalsize

\subsubsection{Proportion of Demands Accepted}

\begin{figure}[t]
    \centering
    \includegraphics[width=0.98\linewidth, trim={0.2cm 0.5cm 0.1cm 0}]{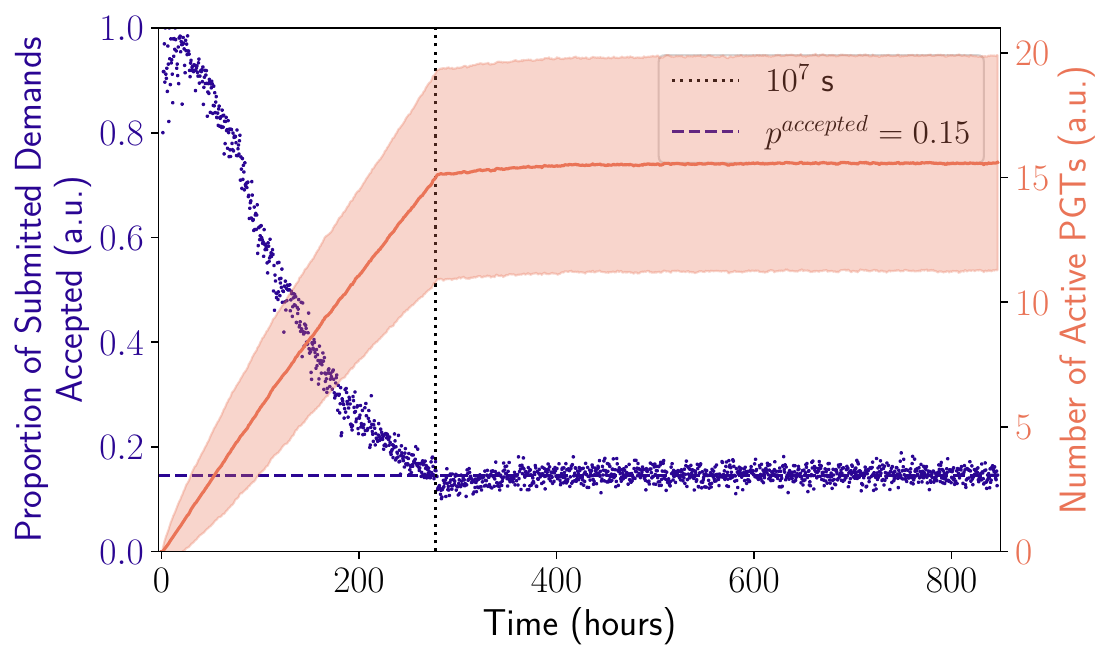}
    \caption{Proportion of submitted demands which are accepted at a specific time (left, blue) and the number of active PGTs (right, orange) over time. Application sessions are first created at times that are uniformly distributed over the first $10^7$ s of network operation, emphasized by the vertical black dotted line. The horizontal dashed blue line indicates the average of the points after the black dotted line. }
    \label{fig: proportion accepted demands}
\end{figure}

The proportion of submitted demands that are accepted, the subject of~\ref{eval Q: prop accepted}, is a performance metric that captures information about the load on a network relative to the tolerance of Arqon's admission control algorithm (Admit Tasks). It is a metric that directly impacts application sessions on end nodes, as a low proportion of accepted demands may translate into long wait times to receive a service agreement, a result of demands needing to be re-submitted multiple times. 

For simulations of the Dumbbell topology network, Figure~\ref{fig: proportion accepted demands} plots the average proportion of submitted demands that are accepted by Admit Tasks per scheduling interval and the number of active PGTs in each scheduling interval. 
The proportion of demands accepted as some particular PGT decreases as a function of the number of active PGTs. It reaches a steady state value of $0.15$ after $10^{7}$ s, after which time all application sessions have submitted at least one demand. When demands are rejected, terminated, or expire, application sessions on end nodes submit new demands following a Poisson process with a session specific resubmission rate. 

\renewcommand{\arraystretch}{1.2}
\begin{table*}[t]
    \centering
    
    \begin{tabular}{|c|c|c|c|}
    \hline
    $\mathbf{\epsilon_{service}}$ & \thead{\textbf{ Total proportion of} \\ \textbf{Demands Accepted}} & \thead{\textbf{Total Proportion of} \\\textbf{Registered Demands Accepted}} & \thead{\textbf{Total Proportion of Demands} \\ \textbf{which Pass Demand Registration}}\\\hline
    \multicolumn{4}{|c|}{\footnotesize \textbf{\textit{Default Arqon: Bonus Allocation Enabled in Compute Schedule}} \small}\\
    \hline
    1e-05 &  0.127 & 0.131 & 0.992\\
    0.001 & 0.123 & 0.128 & 0.991\\
    0.01 & 0.102 & 0.105 & 0.993\\
    0.1 & 0.120 & 0.123 & 0.993\\
    0.5 & 0.211 & 0.222 & 0.984\\
    \hline
    \multicolumn{4}{|c|}{\footnotesize \textbf{\textit{with Bonus Allocation Disabled in Compute Schedule}} \small}\\
    \hline
    1e-05 & 0.073 & 0.077 & 0.992\\
    0.001 & 0.087 & 0.096 & 0.986\\
    0.01 & 0.042 & 0.047 & 0.991\\
    0.1 & 0.055 & 0.069 & 0.982\\
    0.5 & 0.138 & 0.155 & 0.977\\
    \hline
    \end{tabular}
    \caption{The total proportion of demands registered and accepted as particular PGTs, addressing~\ref{eval Q: prop accepted}, in simulations of the randomly generated network topologies parametrized in Table~\ref{tab: random topology parameters}.}
    \label{tab:random topos prop accepted}
\end{table*}
For the simulations of randomly generated network topologies, the average proportion of demands accepted and the proportion of demands which pass demand registration, both with default Arqon and with the bonus allocation phase of Compute Schedule disabled are recorded in Table~\ref{tab:random topos prop accepted}. With default Arqon, for all values of $\epsilon^{\text{service}}$ simulated the average proportion of demands accepted was between $0.12$ and $0.22$. In contrast, the proportion of demands which passed Demand Registration always exceeded $0.98$. Therefore, almost all rejected demands are rejected by Admit Tasks rather than be Demand Registration.
With bonus allocation disabled, the average proportion of demands accepted decreased by at least $0.036$ for every simulated value of $\epsilon^{\text{service}}$, and by up to $0.073$ for ${\epsilon^{\text{service}}=0.5}$.

In all simulations (dumbbell topology and random network topologies), the average proportion of demands accepted is low. This indicates that the average load on the network in our simulations is similar to the maximum supportable load calculated by Admit Tasks. 
Combined with the results from Table~\ref{tab: random topos reliability} which indicate Arqon over-serves accepted demands, we expect that the acceptance criterion of Admit Tasks is more stringent than is necessary to satisfy reliability requirement \ref{req: 4: independent}. An alternate implementation of Arqon may modify the function \textsc{CalculateRequiredTime}, with which AdmitTasks calculates the time required to schedule a minimal allocation of all active PGTs. A possible modification would be to omit from the calculation any accepted PGTs~$\gamma$ for which the expected number of successfully generated packets \eqref{eq: expected successful PGAs}, already exceeds the required number,~$\ninst_{\gamma}$.

\subsubsection{Profiling Accepted PGTs}
\label{subsec: eval Profiling PGT acceptance}

To address~\ref{eval Q: choice of FCs} Figure~\ref{fig: profiling PGTs by FC} plots the number of submitted PGTs (Figure~\ref{subfig: FC interplay submitted}) and the number of accepted PGTs (Figure~\ref{subfig: FC interplay accepted}) in each filling class. There is one data point for each scheduling interval, for each network topology parameterization in Table~\ref{tab: random topology parameters}. Data points are averages over the ten random topologies generated from each parameterization and the randomly seeded simulations for each topology. 

The number of active PGTs in one filling class can impact the proportion of PGTs admitted to another filling class because Admit Tasks reserves time for active PGTs on the resources associated with their filling classes.
In our implementation of Arqon, we employ the specific choice of filling classes \eqref{eq: specific FC partition} derived from the path partition in Definition~\ref{def: specific path partition}. 
In this path partition, all paths which contain at least one long-distance backbone are in the same cell of the partition. For network topologies with a large number of local areas and long-distance backbones, the result is that PGTs with paths between only two local areas are in the same filling class as PGTs with paths between multiple local areas. 

In Figure~\ref{fig: profiling PGTs by FC} there is some clustering of the registered PGTs (Figure~\ref{subfig: FC interplay submitted}) based on the number of local areas in the network parameterization. In contrast, after Admit Tasks (Figure~\ref{subfig: FC interplay accepted}) there is clear clustering of the active PGTs based on the unique network topology parameterization, with ten data groupings visible in each subfigure, each corresponding to one of the ten distinct network parameterizations. 

For all network parameterizations simulated, the number of registered PGTs in a junction (${\Phi^J_{j} \ \forall j \in J}$) or interface (${\Phi^{I}_{i} \ \forall i \in I}$) filling class increases nearly monotonically with the number of registered PGTs in the backbone filling class (Figure~\ref{subfig: FC interplay submitted}). Similarly, the number of registered PGTs in an interface filling class increases nearly monotonically with the number of registered PGTs in a junction filling class.
The relative number of PGTs registered in each filling class derives immediately from the simulation configuration, which sets the number of application sessions between end nodes connected by routes in interface, junction, and backbone filling classes. 
After Admit Tasks (Figure~\ref{subfig: FC interplay accepted}), there remains a nearly monotonic increase in the number of active PGTs in an interface filling class with the number of accepted PGTs in a junction filling class. In contrast, the numbers of active PGTs in a junction or interface filling class relative to the number in the backbone filling class indicate a different relationship. The number of PGTs accepted to the backbone filling class inhibit acceptance of PGTs to both junction and interface filling classes.
To combat this effect an alternate implementation of Arqon may define a path partition which sub-divides the backbone cell of the partition $\Pi^B$ based on the number of local areas traversed by a path. 

\begin{figure*}
\centering
        \begin{subfigure}{\textwidth}
        \centering
        \includegraphics[width=0.9\textwidth, trim={0cm 0 0cm 0}]{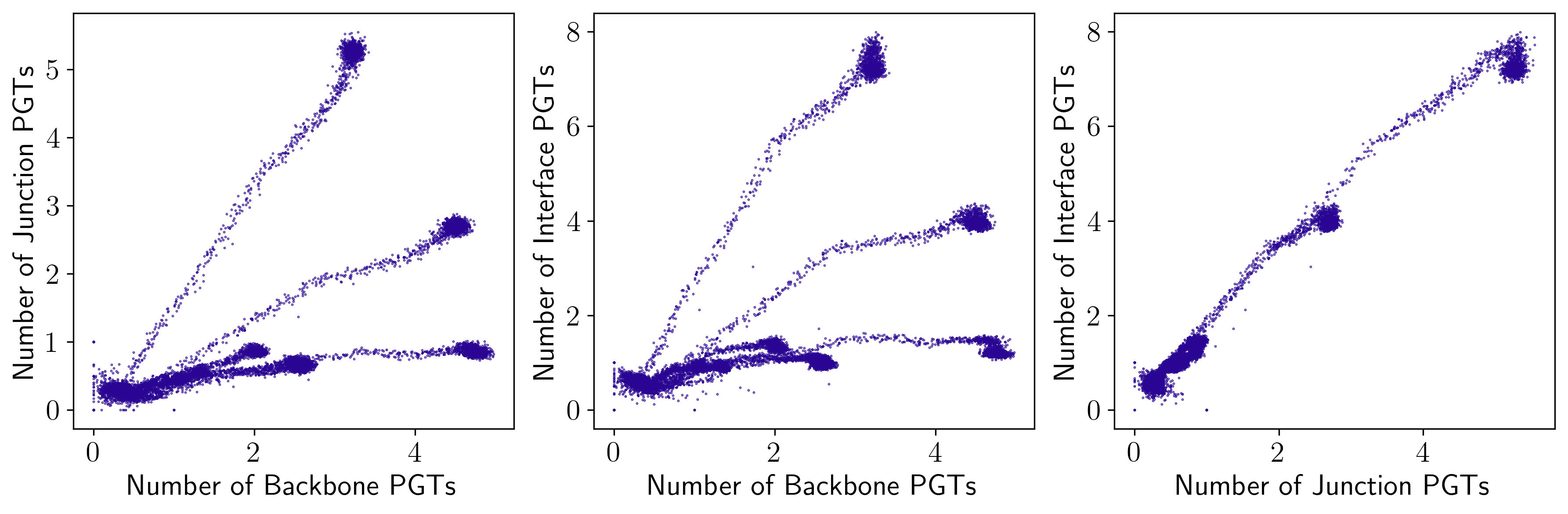}
        \caption{Number of Registered PGTs}
        \label{subfig: FC interplay submitted}
    \end{subfigure}
    \begin{subfigure}{\textwidth}
    \centering
    \includegraphics[width=0.9\textwidth, trim={0cm 0 0cm 0}]{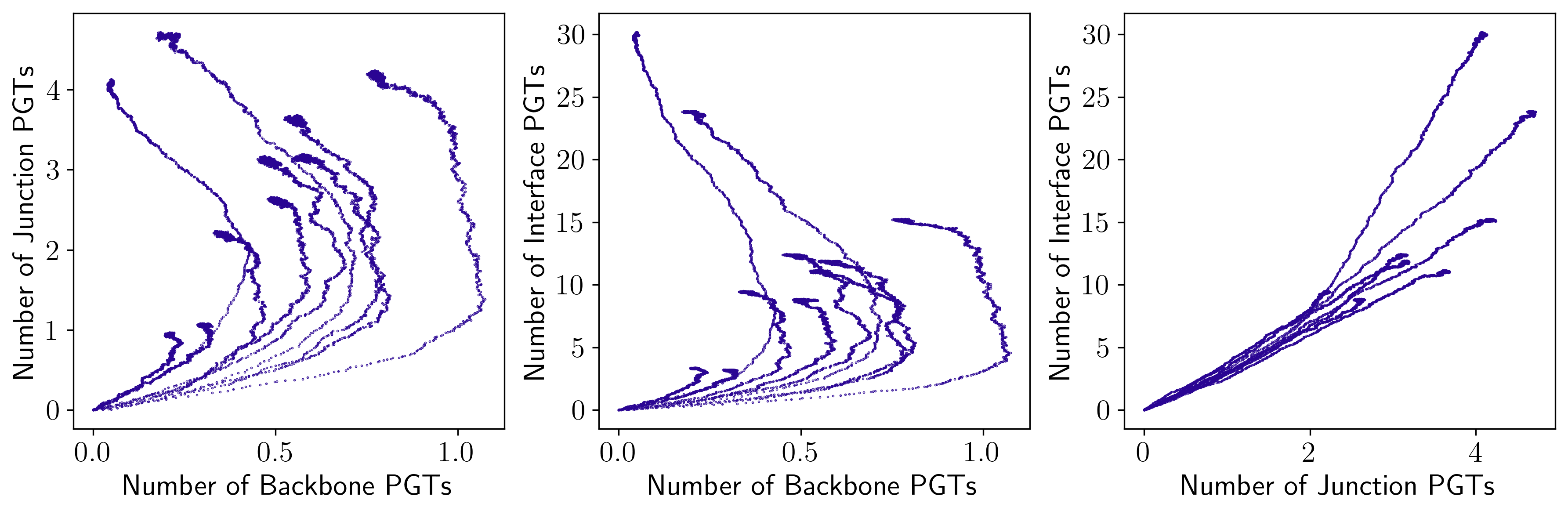}
    \caption{Number of Active PGTs}
    \label{subfig: FC interplay accepted}
    \end{subfigure}
    \caption{Profiling of PGTs registered by Demand Registration and active PGTs accepted by Admit Tasks based on their filling classes [backbone ($\Phi^B$), junction (${\Phi^J_j \ \forall j \in J}$), and interface (${\Phi^I_i \ \forall i \in I}$) filling classes]. There is one data point for each scheduling interval, for each network topology parameterization in Table~\ref{tab: random topology parameters}. Data is averaged over the ten random topologies generated from each parameterization, and ten randomly seeded simulations for each topology.} 
    \label{fig: profiling PGTs by FC}
\end{figure*}

\subsection{Numerical Complexity Analysis}
\label{subsec: Numerical Complexity Analysis}

In every simulation conducted, \arqon was able to compute the network schedule much faster than the cutoff time for schedule distribution ($t_{\text{compute}} \ll T^{SI}$). The average time to produce a network schedule for each of the random topologies simulated is recorded in Table~\ref{tab: random topos computation times} in Appendix~\ref{app: complexity proofs}. The maximum of the average times required to produce a network schedule was less than 0.25 seconds~($t_{\text{compute}}$), whereas the duration of the scheduling interval was 1800 seconds~($T^{SI}$). 

In Section~\ref{sec: Complexity} we proved Theorem~\ref{thm: overall complexity}, which states that the full program \textsc{NetworkScheduler} has operational complexity that is polynomial in the number of PGTs in the task intake buffer ($k=|\bm{\Gamma}|$) and the number of active PGTs ($N=|Z_{\Phi}|$), but exponential in the total number of internal resources in the network ($R=|\mathfrak{R}|$). 
We validate our complexity analysis with numerical results, demonstrating how operational complexity translates to execution time for various combinations of $k$, $N$, and $R$. 
First, we assess the computation time of \textsc{AdmitTasks} using randomly generated sets of active PGTs~($Z_{\Phi}$ and PGTs in the task intake buffer~($\bm{\Gamma}$). The methods for generating these random sets are included in our simulator~\cite{arqon-sim}. Second, we evaluate \textsc{ComputeSchedule} with a carefully designed set of PGTs to schedule~($Z^{*}_{\Phi}$) for which specific algorithm characteristics produce clear features in the computation time. 
This second evaluation complements our earlier simulations on random network topologies, where schedule computation times were consistently aligned with the complexity analysis.

\subsubsection{Admit Tasks}

\begin{figure*}
    \centering
    \includegraphics[width=0.9\textwidth]{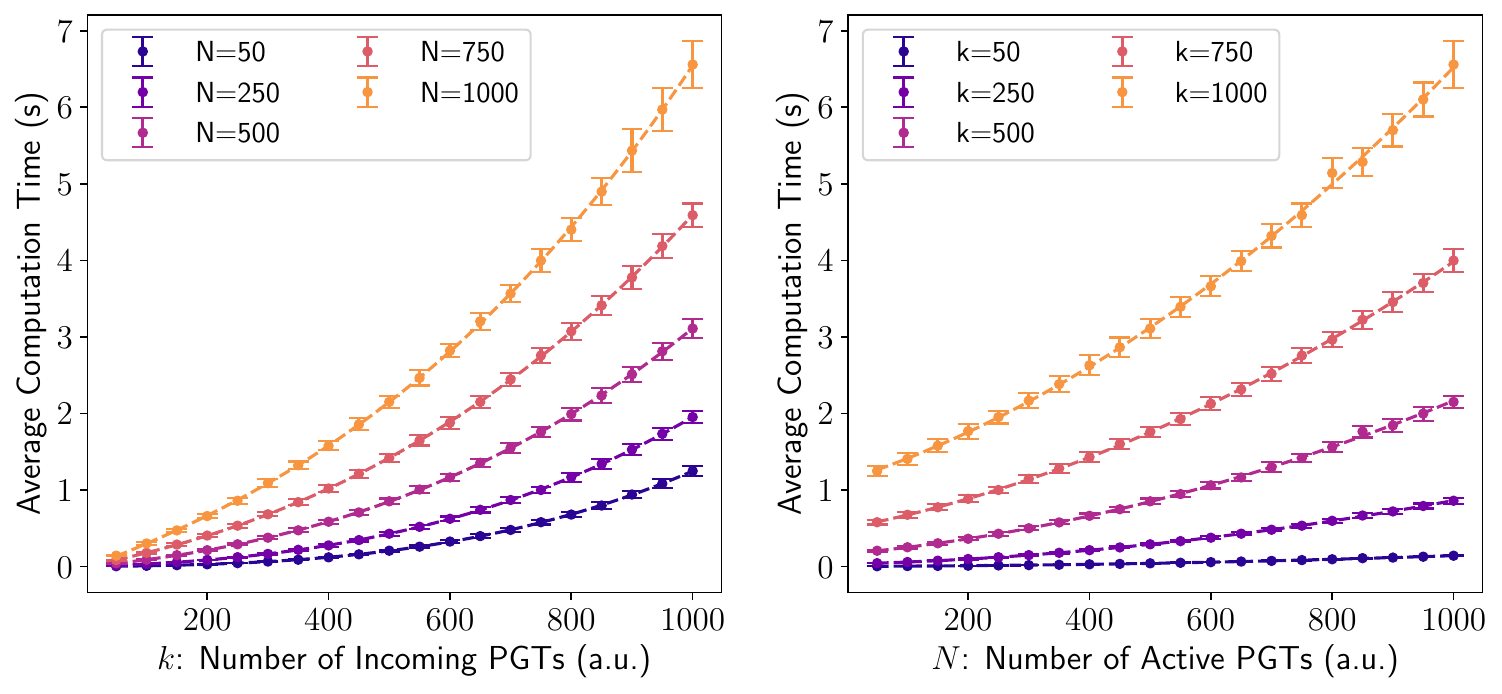}
    \caption{Average computation time of Admit Tasks for ${N=|Z_{\Phi}|}$ active PGTs and ${k=|\bm{\Gamma}|}$ incoming PGTs requiring an accept/reject decision from Admit Tasks.
    Each data point is the average of 100 randomly generated $Z_{\Phi}, \bm{\Gamma}$ pairs and error bars show one standard deviation. Data points are fitted with polynomials of $O(k^3)$ (left) and $O(N^2)$ (right).
    }
    \label{fig: admission control complexity}
\end{figure*}

The operational complexity of \textsc{AdmitTasks} is given by \eqref{eq: admit tasks complexity main text}. 
Figure~\ref{fig: admission control complexity} plots the computation time against ${k}$ for constant $N$ (left) and against $N$, for constant ${k}$ (right). 
Each data point is the average over one hundred randomly generated $Z_{\Phi} \text{ and } \bm{\Gamma}$ pairs, with error bars indicating one standard deviation. Polynomial fits of ${O(k^{3})}$ (left) and ${O(N^{2})}$ (right) confirm scaling consistent with~\eqref{eq: admit tasks complexity main text}, validating Theorem~\ref{thm: admit new tasks complexity}. 
Importantly, these results demonstrate that \textsc{AdmitTasks} is fast enough to process a large number of incoming PGTs ($k$) within a single scheduling interval, even when there are many active PGTs ($N$). For example, one thousand incoming PGTs ($k = 1000$) are processed in less than seven seconds when there are already one thousand active PGTs ($N=1000$).

\subsubsection{Compute Schedule}

\begin{figure*}
    \centering
    \includegraphics[width=0.98\textwidth]{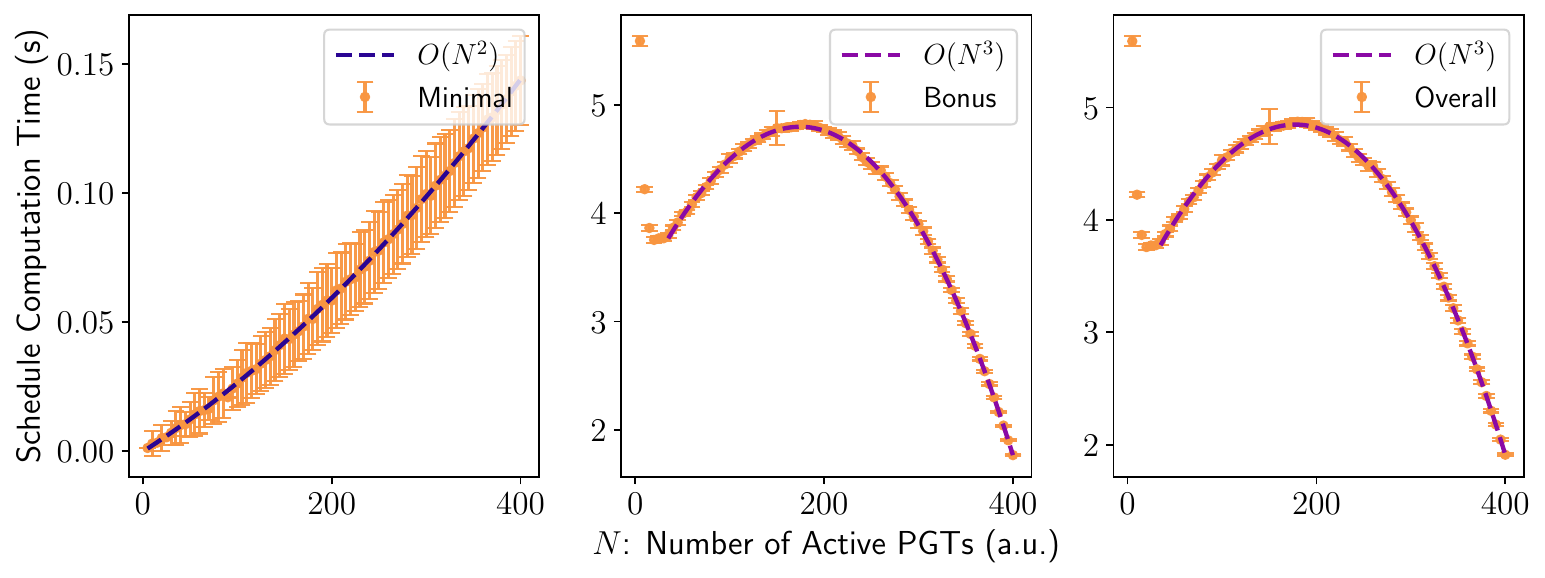}
    \caption{Average computation time of the minimal (left) and bonus (middle) allocation phases of Compute Schedule, and overall computation time (right), plotted against the number of active PGTs $N=|Z^{*}_{\Phi}|$. Each data point is the average of 50 repetitions of Compute Schedule for a set of active PGTs $Z^{*}_{\Phi}$. Error bars show one standard deviation. Polynomial fits to the data are based on least squares fitting.}
    \label{fig: scheduling complexity}
\end{figure*}

The operational complexity of \textsc{ComputeSchedule} is given by~\eqref{eq: compute schedule complexity main text}. 
The minimal allocation phase, implemented by Algorithm~\ref{alg: direct allocation subroutine}, contributes $O(NR) + O(N^2)$ and contains no branching. The total number of operations depends on $N$, but not on particular PGT characteristics. The resulting network schedule may contain gaps during which no PGAs are scheduled, as illustrated in Figure~\ref{fig: example network schedule}.
The bonus allocation phase, implemented by Algorithm~\ref{alg: bonus round robin resource schedule} attempts to fill these gaps. 
It contributes $O(N^3 R)$, which dominates the overall complexity. Its branching behavior depends non-trivially on the number and duration of gaps, the execution times ${E_{\gamma} \ (\forall \gamma \in Z_{\Phi}})$ of PGAs that may be added, and the minimum separation ${\tminsep_{\gamma} \ (\forall \gamma \in Z_{\Phi}})$ between PGAs  of each PGT.

To illustrate this branching behavior, we consider a specific set of PGTs ($Z^{*}_{\Phi}$) with the following structure: \begin{itemize}
    \item the first PGT $\gamma_0$ has ${\tminsep_{\gamma_0} = 200}$ seconds, and all other PGTs ${\gamma_i, \ i>0}$ have ${\tminsep_{\gamma_i}=0}$;
    \item ${E_{\gamma} = 1 \text{ second}, \ \forall \gamma \in Z^{*}_{\Phi}}$;
    \item $\minimumallocation_{\gamma} = 20, \ \forall \gamma \in Z^{*}_{\Phi}$.
\end{itemize}
To study this scenario, the scheduling interval was set to eight thousand seconds ($T^{SI} = 8000$) and we repeated schedule computation for scenarios with ${|Z^{*}_{\Phi}|}$ ranging from five to four hundred, in steps of five.  

In Figure~\ref{fig: scheduling complexity} the schedule computation time, decomposed into the minimal allocation phase (left), the bonus allocation phase (middle) and the overall computation time (right), is plotted against $N$.
Data points are the average over fifty independent repetitions and error bars indicate one standard deviation.
The time to complete the minimal allocation phase (left) is increasing in $N$ and a polynomial fit validates that it has $O(N^2)$ complexity (left).
The time to complete the bonus allocation phase (middle) exhibits non-trivial variation with $N$, yet all features are explained by our algorithmic analysis. 

Whenever there is a gap in the schedule, the \textsc{BonusRoundRobin} algorithm checks the next PGT in round robin order to determine if all the required resources are available ($O(R)$ operation). Checking $\gamma_0$ additionally requires both checking for both past and future minsep violations, since ${\tminsep_{\gamma_0} \neq 0}$. For small $N$, numerous gaps exist and the round robin index frequently revisits $\gamma_0$, resulting in the initial spike in computation time. 
The most time-intensive operation occurs when filling gaps between two already scheduled PGAs, which requires list insertion ($O(N)$ in our implementation). When no gaps exist but time remains in the scheduling interval, PGAs can simply be appended to the schedule ($O(1)$).
For ${N<200}$, gaps exists in the schedule output by the minimal allocation phase, and these are filled by list insertion. 
For ${{N} \geq 200}$, no gaps remain and new PGAs are only appended to the schedule by \textsc{RoundRobinBonus}. 
This transition produces the local maximum at ${N=200}$. 
As ${N}$ increases further, the number of PGAs that can be appended to the schedule decreases, until ${N=400}$ where the schedule output by the minimal allocation phase has no gaps and occupies the entire scheduling interval. A polynomial fit validates $O({N}^3)$ complexity for the bonus allocation phase. The overall complexity (right) confirms the result of Theorem~\ref{thm: compute NS complexity}, showing that the schedule computation time is dominated by the bonus allocation phase and has $O(N^3)$ complexity.
See Appendix~\ref{app: complexity proofs} for full details of our complexity analysis.

\section{Conclusion and Outlook}\label{sec: Conclusion}
We have designed Arqon, a novel suite of control applications for a centrally controlled quantum network that establishes and satisfies well-defined service agreements with end nodes. 
These service agreements address realistic application execution requirements on programmable end nodes, so that satisfying a service agreement implies that user applications can succeed. 
In contrast to existing quantum network control architectures, Arqon includes an admission control mechanism which prioritizes service to already accepted demands over incoming demands, preventing service disruptions.

We extended the concept of reliable service from classical computer networks and demonstrated through analysis and numeric evaluations with a proof of principle implementation that Arqon provides reliable service to all accepted demands in static network topologies.
We designed Arqon to respond to changes in network topology or capabilities. 
In our proof of principle implementation we did not implement the mechanisms for handling these updates.
A future research direction is to create concrete implementations of these mechanisms and evaluate the robustness of Arqon against changes to topology and capabilities in a variety of parameter regimes.  

We developed fast algorithms for producing network schedules and executing a demand admission control process, and demonstrated numerically that each algorithm can process up to one thousand demands in less than 7 seconds. The core of these algorithms relies on directly calculating how a minimal allocation of PGAs for all active PGTs can be scheduled in each scheduling interval. In contrast, many quantum network architectures implement scheduling algorithms which rely on optimization, which is generally slower than direct calculation of a schedule.   

Arqon's admission control process, which  ensures arriving demands do not disrupt service to accepted demands, comes at the cost of rejecting some demands. 
In simulations based on our proof of principle implementation we observed that there is scope to increase the proportion of demands that Arqon accepts.
We proposed modifying the calculation which Arqon's admission control process uses to estimate load on the network. 
To further increase the resource utilization and proportion of demands that may be accepted, we proposed defining a path partition that differentiates between paths through the network that cross one or multiple long-distance backbones.

The performance of any quantum network architecture needs to be tested in a real world deployment. We have developed an implementation of Arqon in Python as a simulator. The next step is to create a real-time implementation of Arqon, integrate it with an SDN-controller, and use the integrated system to produce network schedules for the components of a test-bed quantum network. 
Such a demonstration would bring scalable quantum networking one step closer to reality.

\section{Acknowledgements}
SG, TRB and SW acknowledge funding from the Quantum Internet Alliance (QIA). 
QIA has received funding from the European Union’s Horizon Europe research and innovation programme under grant agreement No. 101102140.
SW also acknowledges funding from NWO VICI.

\small
\printbibliography

\normalsize
\appendix

\begin{center}
{\Large\scshape \textbf{Appendices}}
\end{center}

    \section{Notation}\label{app: notation}
    \normalsize
    The notation used throughout this paper is summarized in Table~\ref{tab: notation summary}.
    \begin{table*}
        \centering
    \begin{tabular}{c|l}
        \textit{\textbf{Symbol}} & \textit{\textbf{Definition}} \\\hline
        \multicolumn{2}{c}{\textbf{Applications}}\\\hline
         $\mathcal{S}$ & An application session\\
         $\mathcal{N}$ & A set of end nodes which hosts an application session \\
         $N^{\text{inst}}$ & A number of application instances required to complete an application session\\
         $\texpiry$ & The expiry time of an application session and demand\\
         \hline
         \multicolumn{2}{c}{\textbf{Demands}}\\\hline
         $\mathfrak{p} = (w, s, F)$ & A packet of $s$ end-to-end links to be generated within time window $w$, with minimum fidelity $F$ \\
         $\tminsep$ & A minimum required time separation between attempts to generate a packet \\
         $d = (d^{\text{raw}}, M_d)$ & The full demand consists of the raw demand $d^{\text{raw}}$ and demand metadata $M_d$ \\
         $\epsilon^{\text{service}}_{d}$ & The service error parameter for a demand $d$ \\ \hline
         \multicolumn{2}{c}{\textbf{Packet Generation Tasks}}\\\hline 
         $\gamma$ & A PGT \\
         $E_{\gamma}$ & The duration of PGAs for the PGT $\gamma$ \\
         $p^{\texttt{packet}}_{\gamma}$ & The probability a PGA of duration $E_{\gamma}$ successfully produces the required packet of links \\
         $t^{\texttt{start}}_{\gamma}$ & The earliest time that a PGA for a PGT $\gamma$ could be scheduled if it is accepted\\
         $\pga(\gamma, t; r)$ & A PGA for a PGT~$\gamma$ scheduled at time~$t$ on resource~$r$ \\ \hline   
         \multicolumn{2}{c}{\textbf{Network Topology and Components}}\\ \hline
         $\mathcal{G} = (\mathcal{V}, \mathcal{E})$ & A network resource graph with vertices $\mathcal{V}$ and edges $\mathcal{E}$ \\ 
         $E \subset{\mathcal{V}}$ & The set of end nodes \\
         $I\subset{\mathcal{V}}$ & The set of entanglement generation interfaces \\
         $J\subset{\mathcal{V}}$ & The set of junction nodes \\
         $B\subset{\mathcal{V}}$ & The set of long-distance backbones \\
         $\mathfrak{R} = I \sqcup J \sqcup B $ & Set of all internal resources in a network \\
         $R=|\mathfrak{R}|$ & Number of internal resources in a network \\
         $\mathcal{L}$ & The set of local areas of $\mathcal{G}$ \\
         $\mathcal{E}^{*}\subset \mathcal{E}$ & A set of allowed edges of a network resource graph $\mathcal{G}$ \\ \hline
         \multicolumn{2}{c}{\textbf{Paths}}\\ \hline
         $\pi = (\pi_0, \cdots, \pi_{k})$ & a path through the network resource graph from source node $\pi_0$ to destination node $\pi_{k}$ \\
         $\mathcal{P}_{\text{valid}}$ & The set of valid entanglement generation paths \\
         $\mathcal{P}_{\text{allowed}}\subset{\mathcal{P}_{\text{valid}}}$ & A set of allowed paths, further restricted from the set of valid paths \\
         $\Pi = \{ \Pi_{\phi}\}$ & A disjoint partition  of $\mathcal{P}_{\text{valid}}$ or $\mathcal{P}_{\text{allowed}}$ \\ \hline   
         \multicolumn{2}{c}{\textbf{Demand Manager}}\\ \hline
         $\tau$ & The demand termination buffer \\
         $\bm{\Gamma}$ & The PGT intake buffer \\
         $\bm{\Gamma}_d \subseteq \bm{\Gamma}$ & A set of PGTs in the task intake buffer which can realize a demand $d$ \\
         $k = |\bm{\Gamma}|$ & Number of PGTs in the task intake buffer \\ \hline  
         \multicolumn{2}{c}{\textbf{Network Scheduler}}\\ \hline
         $S$ & A network schedule \\
         $T^{SI}$ & The duration of a scheduling interval\\
        $t_{\text{compute}}$ & The computation time of the Network Schedule application \\
        $\Phi$ & The set of filing classes \\
        $\phi = (Z_{\phi}, \Pi_{\phi})$ & A filling class based on the cell of the path partition $\Pi_{\phi}$ so that $\pi_{\gamma} \in \Pi_{\phi} \ \forall \gamma \in Z_{\phi}$  \\
        $Z_{\Phi}$ & The set of all active PGTs \\
        $N = |Z_{\Phi}|$ & Number of active PGTs\\
        $\xi$ & Mapping from a set of paths to the associated internal resources \\
    \hline
    \end{tabular}
        \caption{Summary of the notation used throughout this paper. The notation $\bm{f}[i]$ indicates the $i^{th}$ entry of a vector, tuple, matrix or dictionary-like object $\bm{f}$.}
        \label{tab: notation summary}
    \end{table*}

\section{Network Scheduler Algorithms}\label{app: Network Scheduler Algorithms}
In this appendix we provide pseudocode algorithms for all processes of the Network Scheduler for which an algorithm is specified in the main text. These are as follows: \begin{itemize}
    \item[--] \textsc{UpdateFillingClasses}, Algorithm~\ref{alg: Update Filling Classes} for Update Filling Classes. 
    \item[--] \textsc{AdmitTasks}, Algorithm~\ref{alg: Admit Tasks} for Admit Tasks.
    \item[--] \textsc{DirectAllocation}, Algorithm~\ref{alg: direct allocation subroutine}, which implements the inner loop of the minimal allocation phase of Compute Schedule.
    \item[--] \textsc{BonusRoundRobin}, Algorithm~\ref{alg: bonus round robin resource schedule}, which implements the inner loop of the bonus allocation phase of Compute Schedule.
    \item[--] \textsc{nextStartTime}, \textsc{allResourcesAvailable} and \textsc{getLeft[Right]MinsepViolations} are all in Algorithm~\ref{alg: auxiliary algoroithms used in bonus round robin}. These functions are used by the \textsc{BonusRoundRobin} algorithm to determine whether additional PGAs can be added to a network schedule. 
\end{itemize}

\begin{algorithm}
    \SetKwInOut{Input}{Input}
    \SetKwInOut{Output}{Output}
    \SetKw{Return}{return}
    \SetKw{Continue}{continue}
    \SetKw{Pass}{pass}
    \SetKw{Break}{break}
    \SetKwProg{Fn}{Function}{:}{end}
    \Fn{\textsc{UpdateFillingClasses}}{

    \Input{Set of filling classes $\Psi = \{Z_{\psi}, \Pi_{\psi}\}_{\psi}$, Path partition $\Pi$, Set of terminated PGTs $\tau$.}
    \Output{Updated set of filling classes $\Phi = \{Z_{\phi}, \Pi_{\phi} \}_{\phi} $, Mapping of associated resources $\xi$.}
    
    Set $\Gamma^{*} \leftarrow \emptyset$ \;
    Set $\Phi, \xi \leftarrow \textsc{BuildFillingClasses}(Z=\emptyset, \Pi)$ \;
    \For{$\psi \in \Psi$}{
        \For{$\gamma \in Z_{\psi}$}{
            \If{$ \gamma \in \tau$ \emph{\textbf{ or }} $    \texpiry_{\gamma} \leq \emph{\texttt{CurrentTime}}$}{
             \Continue \;
            }
            \uIf{$\exists \phi \text{ s.t. } \pi_{\gamma} \in \Pi_{\phi}$}{
                Set $\Phi \leftarrow \textsc{AssignFillingClass}(\gamma, \Phi)$\;
                }
            \Else{
            Set $\bm{\Gamma}^{*} \leftarrow \bm{\Gamma}^{*}\cup \gamma$ \;
            }
        }
    } 
    \For{$\gamma \in \bm{\Gamma}^{*}$}{
        Set $\Phi \leftarrow \textsc{ApplyRuleMissingPath}(\gamma, \Phi)$ \;
    }
    
}
\Return $\Phi, \xi$\;

\caption{Algorithm for Update Filling Classes. The functions \textsc{BuildFillingClasses}, \textsc{AssignFillingClass}, and \textsc{ApplyRuleMissingPath} may be re-defined by an implementation, or they may follow the default implementation of Section~\ref{sec: Implementation}.}
\label{alg: Update Filling Classes}
\end{algorithm}

\begin{algorithm*}
    \SetKwInOut{Input}{Input}
    \SetKwInOut{Output}{Output}
    \SetKwProg{try}{try}{:}{}
    \SetKwProg{except}{except}{:}{end}
    \SetKw{Return}{return}
    \SetKw{Continue}{continue}
    \SetKw{Break}{break}
    \SetKwProg{Fn}{Function}{:}{end}
    \Fn{\textsc{AdmitTasks}}{
    \Input{Task Intake Object $\bm{\Gamma}$, Set of filling classes $\Phi = \{Z_{\phi}, \Pi_{\phi}\}_{\phi}$, Set of live network resources $\mathfrak{R}$.}
    \Output{Set of Filling classes $\Phi$, List of accepted PGTs $\bm{A}$.}
        Set $\timeresourcer \leftarrow T^{SI}$ for $r \in \mathfrak{R}$ \;
        \For{$\phi \in \Phi$}{
            $R(Z_{\phi}) = \textsc{CalculateRequiredTime}(Z_{\phi})$ \;
        \For{$r \in \xi(\Pi_{\phi})$}{
        Set $\timeresourcer \leftarrow \timeresourcer - R(Z_{\phi})$ \;
        
        }
    }
    \While{$\bm{\Gamma} \neq \cancel{0}$}{
        Set $\Gamma_d \leftarrow \bm{\Gamma}[0]$ \;
        \While{$\Gamma_d \neq \cancel{0}$}{
            Set $\gamma \leftarrow \Gamma_d[0]$ \; 
            \try{}{
                Set $\phi \leftarrow \textsc{GetFillingClass}(\gamma)$ \;
                }
            \except{\emph{\texttt{FillingClassError}}}{
                Set $\Gamma_d \leftarrow \Gamma_d \setminus \gamma$ \;
                \Continue \;
                }
            Set  $\tilde{Z}_{\phi} \leftarrow Z_{\phi} + \gamma$\;
            Set $\texttt{Accept} \leftarrow \text{True}$ \;
            Set $\timeresourcer \leftarrow \timeresourcer + R(Z_{\phi})$ for $r \in \xi(\Pi_{\phi})$ \;
            Set $R(\tilde{Z_{\phi}}) \leftarrow \textsc{CalculateRequiredTime}(\tilde{Z_{\phi}})$ \;
            \For{$r \in \xi(\Pi_{\phi})$}{
                \If{$\timeresourcer - R(\tilde{Z_{\phi}}) < 0$}{
                    Set $\texttt{Accept} \leftarrow \text{False}$ \;
                    Set $\Gamma_d \leftarrow \Gamma_d \setminus \gamma$ \;
                    \If{$\Gamma_d = \cancel{0}$ }{
                        Set $\bm{\Gamma} \leftarrow \bm{\Gamma} \setminus \Gamma_d$ \;
                    }
                    \Break \;
                    }
                }
            \If{\emph{\texttt{Accept}}}{
                Set $\bm{A} \leftarrow \bm{A} + \gamma$ \;
                Set $Z_{\phi} \leftarrow \tilde{Z_{\phi}}$ \;
                Set $R(Z_{\phi}) \leftarrow R(\tilde{Z_{\phi}})$ \;
                Set $\timeresourcer \leftarrow \timeresourcer - R(Z_{\phi})$ for $r \in \xi(\Pi_{\phi})$ \;
                Set $\bm{\Gamma} \leftarrow \bm{\Gamma} \setminus \Gamma_d$ \;
                \Break \;
            }
        }
    }
\Return $\Phi$, $\bm{A}$}
    \caption[Admission Control]{The Admit Tasks process of the Network Scheduler is implemented by the function \textsc{AdmitTasks}. The helper function \textsc{GetFillingClass} returns the filling class of a task if it exists, or else it raises a \texttt{FillingClassError.} The function
    \textsc{CalculateRequiredTime} is defined in Algorithm \ref{alg: Calculate Required Time}.}\label{alg: Admit Tasks}
\end{algorithm*}

\begin{algorithm}
    \SetKwInOut{Input}{Input}
    \SetKwInOut{Output}{Output}
    \SetKw{Return}{return}
    \SetKw{Continue}{continue}
    \SetKw{Break}{break}
    \SetKwProg{Fn}{Function}{:}{end}

    \Fn{\textsc{DirectAllocation}}{
    \Input{Filling class $\fillclass = \big{(} Z_{\phi}, \Pi_{\phi} \big{)}$ with $Z_{\phi} = (\gamma_0, ..., \gamma_{M-1})$, network schedule $S$, start time $t_0$.}
    \Output{Updated network schedule $S$.}

    Order PGTs $\gamma \in Z_{\phi}$ by a set of indices $X=\{1, \cdots, |Z_{\phi}|\}$ such that $\minimumallocation_x\leq \minimumallocation_{x+1}$ 

    Set $M\leftarrow|Z_{\fillclass}|$

    Set $c_x \leftarrow \max\left( \underset{y \in \{x, \cdots M\}}{\max} \big{(} E_y + \tminsep_y \big{)} , \sum_{y=x}^ME_y \right)$\;
    Set $n_x \leftarrow \minimumallocation_x - \minimumallocation_{x-1}$ for $x\geq 1$\;
    Set $n_0 \leftarrow\minimumallocation_0  -1 $\;

    Set $t_{\text{start}}\leftarrow t_0$\;

    \For{$m = 0, ..., M-1$}{
        \For{$k = 0, ..., n_m - 1$}{
            \For{$x = m, ..., M-1$}{
                Set $t_{\text{\text{offset}}} \leftarrow \sum_{y=m}^{x}E_y$\;
                \For{$r\in\pi_{\gamma_x}$}{
                    Add PGA $\pga(\gamma_x, t_{\text{start}} + k c_m + t_{\text{offset}}; r)$ to $S[r]$\;
                }
                
            }
        }
        \For{$r\in\pi_{\gamma_x}$}{
            Add PGA $\pga(\gamma_m, t_{\text{start}} + n_mc_m; r)$ to $S[r]$\;
        }
        
        Set $t_{\text{start}} \leftarrow t_{\text{start}} + n_mc_m+E_m$\;
        
    }

    }
    \caption[Scheduler for minimal allocation by direct allocation of start times.]{Scheduler for direct allocation of PGA start times for all PGTs in a single filling class $\phi$. $\minimumallocation_x$ is the number of PGAs, each with duration $E_x$, of PGT $\gamma_x$ that need to be scheduled. 
    }
    \label{alg: direct allocation subroutine}

\end{algorithm}

\begin{algorithm}[t]
    \SetKwInOut{Input}{Input}
    \SetKwInOut{Output}{Output}
    \SetKw{Return}{return}
    \SetKw{Continue}{continue}
    \SetKw{Break}{break}
    \SetKwProg{Fn}{Function}{:}{end}

    \Fn{\textsc{RoundRobinBonus}}{

    \Input{Filling class $\phi$, network schedule $S$.}
    \Output{Updated network schedule $S$.}

        Set $T^{\text{rel}}\leftarrow \{\}$\;
        Set $\test \leftarrow $ \textsc{nextStartTime}$(S, T^{\text{rel}}, 0)$\;
        Set $k\leftarrow 0$\;
        \While{$t^{\texttt{\emph{est}}} < T^{SI}$}{
          \For{$i\in\mathbb{Z}_{|\phi|,k}$}{
            Set $C_{\text{resource}}\leftarrow\textsc{allResourcesAvailable}(\gamma_i, \test; S)$\\
            Set $V^{MS}_L \leftarrow \textsc{getLeftMinsepViolations}(\gamma_i, \test; S)$\\
            Set $V^{MS}_R \leftarrow \textsc{getRightMinsepViolations}(\gamma_i, \test; S)$\\
            Set $C_{\text{minsep}} \leftarrow V^{MS}_L\cup V^{MS}_R = \emptyset$\;
            \uIf {$C_{\text{\emph{resource}}} \wedge C_{\text{\emph{minsep}}} \wedge t^{\texttt{\emph{est}}}+E_{\gamma_i} < T^{SI}$}{
              \For{$r\in\pi_{\gamma_i}$}{
                Add PGA $\pga(\gamma_i, \test; r)$ to $S_r$\;
              }
              Set $k\leftarrow i$\;
            }\ElseIf{$V^{MS}_L\neq \emptyset$}{
              Append $\max_{\pga\in V^{MS}_L}(t^{\texttt{end}}_{\pga^*} + \tminsep_{\gamma_i})$ to $T^{\text{rel}}$
            }
          }
          Set $\test\leftarrow\textsc{nextStartTime}(S, T^{\text{rel}},\test)$\;
        }
    }

\caption{Algorithm for the inner loop of the bonus allocation phase of Compute Schedule.
${\mathbb{Z}_{|\phi|, k} = \{ k, k+1, ...,|\phi|-1, 0, 1, ..., k-1) \} }$. $\pga(\gamma, t; r)$ is a PGA for a PGT~$\gamma$ scheduled at time~$t$ on resource~$r$. 
The subroutines \textsc{nextStartTime}, \textsc{allResourcesAvailable} and \textsc{getLeft[Right]MinsepViolations} are specified in Algorithm~\ref{alg: auxiliary algoroithms used in bonus round robin}.
}
\label{alg: bonus round robin resource schedule}
\end{algorithm}

\begin{algorithm*}[!h]
    \SetKwInOut{Input}{Input}
    \SetKwInOut{Output}{Output}
    \SetKw{Return}{return}
    \SetKw{Continue}{continue}
    \SetKw{Break}{break}
    \SetKwProg{Fn}{Function}{:}{end}

    \Fn{\textsc{nextStartTime}}{
        \Input{Network Schedule $S$, list of release times $\hat{T}^{\text{rel}}$, current time $t$}
        \Output{next possible start time $\test$, updated set of release times $T^{\text{rel}}$}

        Set $S' = \bigcup_{r\in\mathfrak{R}}\{\pga\in S_r : t^{\texttt{end}}_\pga > t\}$\;

        \If{$S'=\emptyset \wedge \hat{T}^{\text{\emph{rel}}} = \emptyset$}{
            \Return $\infty$
        }
    
        Set $t_s = \min\{t^{\texttt{end}}_\pga : \pga\in S'\}$\;
        Set $t_r = \min \hat{T}^{\text{rel}}$\;
        \If{$t_r\leq t_s$}{
            Set $T^{\text{rel}} \leftarrow \hat{T}^{\text{rel}} \setminus t_r$\;
        }
        \Return $\min\{t_s, t_r\}, T^{\text{rel}}$\;
    }

    \Fn{\textsc{allResourcesAvailable}}{
        \Input{PGT $\gamma$, time $t$, network schedule $S$}
        \Output{boolean}

        \For{$r\in\pi_\gamma\cap\mathfrak{R}$}{
            \If{$r$ is not available in $[t, t+E_{\gamma})$}{\Return \textit{False}}\;
        }
        \Return \textit{True}\;
    }

    \Fn{\textsc{getLeftMinsepViolations}}{
        \Input{PGT $\gamma$, time $t$, network schedule $S$}
        \Output{set of conflicting PGAs, $V^{MS}_L$}
        Set $V^{MS}_L \leftarrow \emptyset$\;

        \For{$r\in\pi_\gamma\cap\mathfrak{R}$}{
            Set $S'\leftarrow \{\delta\in S_r : [t^\texttt{start}_\delta,t^\texttt{end}_\delta)\cap(t - \tminsep_\gamma, t] \neq \emptyset\}$\;
            Set $V^{MS}_L\leftarrow V^{MS}_L\cup\{\delta\in S' : \textsc{demandID}(\delta) = \textsc{demandID}(\gamma)\}$\;
        }

        \Return $V^{MS}_L$\;
    }
    \Fn{\textsc{getRightMinsepViolations}}{
        \Input{PGT $\gamma$, time $t$, network schedule $S$}
        \Output{set of conflicting PGAs, $V^{MS}_R$}
        Set $V^{MS}_L \leftarrow \emptyset$\;

        \For{$r\in\pi_\gamma\cap\mathfrak{R}$}{
            Set $S'\leftarrow \{\delta\in S_r : [t^\texttt{start}_\delta,t^\texttt{end}_\delta)\cap(t, t + \tminsep_\gamma) \neq \emptyset\}$\;
            Set $V^{MS}_R\leftarrow V^{MS}_R\cup\{\delta\in S' : \textsc{demandID}(\delta) = \textsc{demandID}(\gamma)\}$\;
        }

        \Return $V^{MS}_R$\;
    }
    
\caption{Auxiliary algorithms used in \textsc{BonusRoundRobin}, Algorithm~\ref{alg: bonus round robin resource schedule}.}
\label{alg: auxiliary algoroithms used in bonus round robin}
\end{algorithm*}
\clearpage
\section{Performance Analysis}\label{app: performance analysis proofs}
The main goal of this Appendix is to prove Theorem \ref{thm: min alloc guaranteed}, which is a statement about the performance of the Network Scheduler control application. To prove this theorem, we first develop a series of intermediate results which support the final analysis. 
This section is organized as follows: \begin{enumerate}
    \item We define a method of comparing filling classes and a \textit{well-behaved} property for an abstract set of filling classes. Subsequent results for the performance of the Admit Tasks and Compute Schedule processes require an abstract set of filling classes to have this property. 
    \item We define sequentially valid schedules, define conditions under which a schedule for a single filling class may be guaranteed to be sequentially valid, and we calculate the required duration of sequentially valid schedules for a single filling class. 
    \item We define valid schedules and prove that the compute schedule process produces valid schedules.
    \item We prove that the admission control process is sound, meaning that newly accepted demands do not disrupt existing service agreements \ref{consideration: AC prevents overload}.
    \item We prove Theorem \ref{thm: min alloc guaranteed}.
\end{enumerate}


\subsection{Filling Classes}\label{ssapp: filling classes results}
In this section we first prove that each PGT is only in a single filling class, as long as the path partition $\Pi$ which defines $\Phi$ is disjoint. This result is used later in the proof of \ref{prop: timeReqSingleFC}.
We then build up towards defining a well-behaved property for an abstract set of filling classes, required by Theorem~\ref{thm: min alloc guaranteed}. As a pre-curser, we define a method of comparing filling classes, $\leq_{\xi}$. 
Finally in this section we prove that a well-behaved set of filling classes in an internally connected network has a unique greatest element under the ordering $\leq_{\xi}$. This result is used later in the proof of Theorem~\ref{thm: good accounting applies to partially ordered FCs}.


\begin{lemma}\label{lemma: one fill class per PGT}
Let $\Pi = \{\Pi_{\phi} \}$ be a disjoint path partition and $\Phi$ the set of filling classes defined by $\Pi$.
Let $Z$ be a set of PGTs such that $\forall \gamma \in Z, \ \exists \phi \in \Phi \text{ s.t. } \pi_{\gamma} \in \Pi_{\phi}$. Then, $\{Z_\phi\}$ is a disjoint partition of $Z$.
\end{lemma}
\begin{proof}
 (Disjointedness) Suppose $\exists\gamma\in Z $ such that $\gamma\in Z_\phi$ and $\gamma\in Z_\psi$.
    This implies that $\pi_\gamma\in\Pi_\phi$ and $\pi_\gamma\in\Pi_\psi$, which implies $\Pi_\phi\cap\Pi_\psi \neq\emptyset$. 
    This contradicts that $\Pi$ is a disjoint partition and so no such $\gamma$ may exist. 
    \bigskip
    \newline (Completeness) By the hypothesis, ${\forall \gamma \in Z, \ \exists \phi \in \Phi}$ such that ${\pi_{\gamma} \in \Pi_{\phi}}$. Hence every $\gamma \in Z$ is a member of some $Z_{\phi}$.
\end{proof}





\begin{definition}[Comparison of filling classes]\label{def: comparison of filling classes}
    We define a partial order on $\Phi$ by $\phi>_\xi\psi \Leftrightarrow \xi(\Pi_{\psi})\subset\xi(\Pi_{\phi})$, with equality iff $\xi(\Pi_{\psi}) = \xi(\Pi_{\phi})$.
\end{definition}

As the ordering $\geq_\xi$ is only a partial order, there may exist non-comparable elements of $\Phi$ using $\geq_\xi$.
In particular, if $\xi(\Pi_{\psi})\cap\xi(\Pi_{\phi})=\emptyset$ then $\psi$ and $\phi$ do not have a determinate order under $\geq_\xi$. 

Recall that the network resource graph $\mathcal{G}$, from which the set of valid paths $\mathcal{P}_{\text{valid}}$ and the path partition $\Pi$ of $\mathcal{P}_{\text{valid}}$ are derived, may have the property of being internally connected. 
In that case we write that the path partition $\Pi$ describes an internally connected network. 

\begin{restatable}[Well-behaved set of filling classes]{definition}{WellBehavedFCs}\label{def: well-behaved set of filling classes}
    Let $\Phi =\{\phi_0, \phi_1, ..., \phi_{n-1}\}$ be a set of filling classes defined by the path partition $\Pi$. Let $\xi$ be the mapping of between a set of paths and their associated resources.
    We say that $\Phi$ is \emph{\textbf{well-behaved}} if the following conditions are met: 
    \begin{enumerate}
        \item $\forall i\neq j$, $\xi(\Pi_{\phi_i})\neq\xi(\Pi_{\phi_j})$
        \item $\forall i\neq j$, $\xi(\Pi_{\phi_i})\cap\xi(\Pi_{\phi_j})\neq\emptyset \implies \xi(\Pi_{\phi_i})\subset\xi(\Pi_{\phi_j}) \ \vee \ \xi(\Pi_{\phi_j})\subset\xi(\Pi_{\phi_i})$.
    \end{enumerate}
\end{restatable}


With the ordering $\leq_{\xi}$, condition $2.$ may be re-stated as: 
\begin{enumerate}[label=2'.]
    \item $\forall i\neq j, \ \xi(\Pi_{\phi_i})\cap\xi(\Pi_{\phi_j})\neq\emptyset \implies \phi_{j} >_{\xi} \phi_{i} \ \vee \ \phi_{i} >_{\xi} \phi_{j}$.
\end{enumerate}

These conditions ensure that firstly no two filling classes can be associated with exactly the same set of resources and secondly that if there is any overlap of associated resources between filling classes then the filling classes are comparable. 

\begin{lemma}\label{lemma: unique greatest element of Phi}
    Let the path partition $\Pi$ describe an internally connected network.
    Let ${\Phi = \{\phi_{n-1}, \phi_{n-2}, \cdots, \phi_1, \phi_{0}\}}$ be a well-behaved set of filling classes defined by $\Pi$.
    Then under $\leq_\xi$ as in Definition~\ref{def: comparison of filling classes}, there is a unique greatest element of $\Phi$. 
\end{lemma}
\begin{proof}

The existence of a path partition $\Pi$ with at least one cell in the partition defines a set of filling classes with at least one filling class. In the case that the partition has only a single cell $\Pi_0$, then there is a single filling class $\phi_0$, and $\phi_0$ is a unique greatest element of $\Phi$. 

Suppose therefore that the path partition has ${n>1}$ cells, and therefore defines a set of ${n>1}$ filling classes $\Phi$. For the sake of contradiction, suppose $\Phi$ has no unique greatest element.

Then there exists $i,j$ such that there is no $k$ satisfying $\phi_k >_{\xi} \phi_i$ and $\phi_k >_{\xi} \phi_j$. Such $\phi_i, \phi_j$ are greatest elements of $\Phi$.  
As $\phi$ is well-behaved, we must have that $\phi_i$ and $\phi_j$ are incomparable under $\geq_\xi$ and so $\xi(\Pi_{\phi_i})\cap\xi(\Pi_{\phi_j})=\emptyset$. 
    
    As $\Pi$ describes an internally connected network, there exists some path $\pi^*$ which passes through nodes in $\xi(\Pi_{\phi_i})$ and nodes in $\xi(\Pi_{\phi_j})$. 
    As $\Pi$ is a disjoint partition, there exists a unique $k^*$ such that $\pi^*\in\Pi_{\phi_{k^*}}$.
    As $\pi^*$ passes through nodes in $\xi(\Pi_{\phi_i})$, we must have that $\xi(\Pi_{\phi_{k^*}})\cap\xi(\Pi_{\phi_i})\neq \emptyset$. 
    As $\Phi$ is well-behaved, this implies that either $\xi(\Pi_{\phi_i})\subset\xi(\Pi_{\phi_{k^*}})$ or $\xi(\Pi_{\phi_{k^*}})\subset\xi(\Pi_{\phi_{i}})$. 
    If the former case holds, then $\phi_{k^{*}}>_\xi\phi_{i}$, contradicting that $\phi_i$ is a greatest element of $\Phi$. 
    If the latter case holds, then $\xi(\{\pi^*\})\subseteq\xi(\Pi_{\phi_{k^*}})\subset\xi(\Pi_{\phi_i})$.
    However, as $\xi(\{\pi^*\})\cap\xi(\Pi_{\phi_j})\neq\emptyset$, then $\xi(\Pi_{\phi_i})\cap\xi(\Pi_{\phi_j})\neq\emptyset$ contradicting that $\phi_i$ and $\phi_j$ are incomparable under $\geq_\xi$. 

    The same logic applies symmetrically to $\phi_j$, and therefore no such $i,j$ can exist and the result holds. 
\end{proof}

\subsection{Schedules for a Single Filling Class}
In the following, we first introduce a property called a \textit{minsep violation}. We then make use of this property to define a sequentially valid schedule for a single filling class. We then prove that a scheduling method that adheres to three simple conditions produces sequentially valid schedules for a single filling class. 

 \begin{definition}[Minsep violation]
    Let $S$ be a schedule for filling class $\phi$ with ${Z_{\phi} = \{\gamma_0, \cdots, \gamma_{M-1}\}}$ with start times $s_{x,i}$ and end times $e_{x,i}$ for the $i$th scheduled PGA of PGT $\gamma_x$, where ${x \in \{0, \cdots, M-1\}}$. 
    We say a \emph{\textbf{minsep violation}} occurs in $S$ if ${\exists x, i}$ such that ${s_{x,i+1} - e_{x,i}<\tminsep_{\gamma_x}}$.
\end{definition}

\begin{definition}[Sequentially valid]
    A schedule is said to be \textbf{sequentially valid} if it has the following two properties: \begin{enumerate}
    \item No two PGAs are scheduled simultaneously, and
    \item No minsep violations occur.
\end{enumerate}
\end{definition}

We now define two simple conditions on the start and end times of PGAs in a network schedule that are compatible with sequentially valid schedules.
\begin{definition}[First and Second Sequential Scheduling Conditions]
Let $S$ be a schedule for an arbitrary filling class $\phi = (Z_{\phi}, \Pi_{\phi})$, populated by PGTs $Z_{\phi} = \{\gamma_x\}_x$. Let $s_{x,i}$ and $e_{x,i}$ denote the start and end times, respectively, for the $i$th PGA from PGT $\gamma_x$. Let $E_{\gamma_x}$ denote the execution time of a PGA for PGT $\gamma_x$.
The \textbf{sequential scheduling conditions} are,
  \newline (Condition $c_1$): \begin{equation}\label{eq: sched condition c1}
        e_{x,i} = s_{x,i} + E_{\gamma_{x}}, \ \forall x \geq 0,
  \end{equation} and, 
  \newline (Condition $c_2$):\begin{equation}\label{eq: sched condition c2}
      s_{x,i} = e_{x-1, i}, \ \forall x\geq 1.
  \end{equation}
\end{definition}

\subsubsection{Calculating the required duration of sequentially valid schedules}
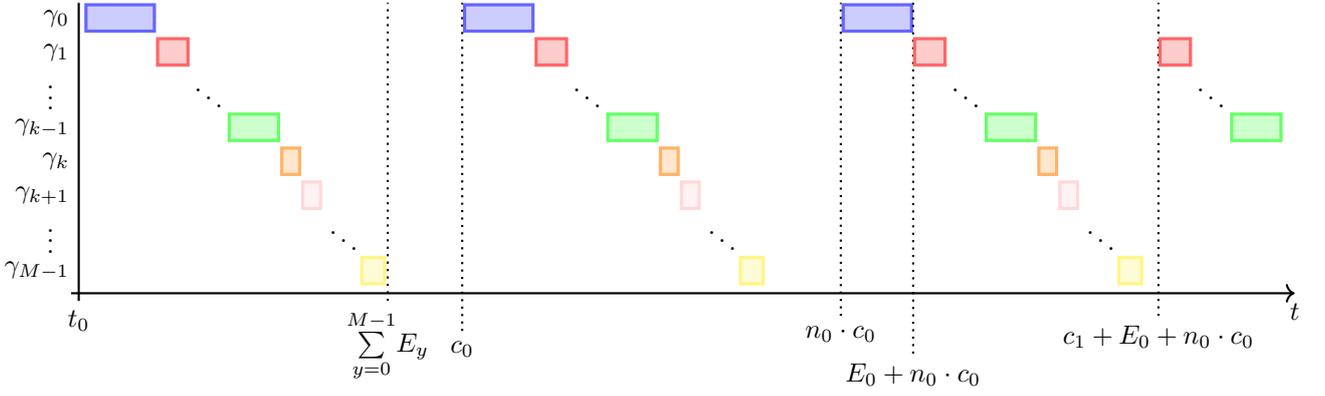
\begin{figure*}[t]
    \centering

\begin{tikzpicture}[node distance = 1cm,
        background rectangle/.style={fill=white},
        d0/.style={rectangle,
                        draw=blue!60, 
                        fill=blue!20, 
                        very thick, 
                        align=center, 
                        minimum width=0.9cm, 
                        minimum height = 0.35cm,
                        },
        d1/.style={rectangle,
                        draw=red!60, 
                        fill=red!20, 
                        very thick, 
                        align=center, 
                        minimum width=0.4cm, 
                        minimum height = 0.35cm,
                        },
        dk-1/.style={rectangle,
                        draw=green!60, 
                        fill=green!20, 
                        very thick, 
                        align=center, 
                        minimum width=0.65cm, 
                        minimum height = 0.35cm,
                        },
        dk/.style={rectangle,
                        draw=orange!60, 
                        fill=orange!20, 
                        very thick, 
                        align=center, 
                        minimum width=0.15cm, 
                        minimum height = 0.35cm,
                        },
        dk+1/.style={rectangle,
                        draw=pink!60, 
                        fill=pink!20, 
                        very thick, 
                        align=center, 
                        minimum width=0.2cm, 
                        minimum height = 0.35cm,
                        },
        dm/.style={rectangle,
                        draw=yellow!60, 
                        fill=yellow!20, 
                        very thick, 
                        align=center, 
                        minimum width=0.3cm, 
                        minimum height = 0.35cm,
                        },
    ]

    \coordinate (root) at (0,0);
                    
    \draw[thick, ->] (-0.1,0) -- (16,0) node [below] {$t$};
    \draw[thick] (0, -0.1) node [below] {$t_{0}$} -- (0, 3.85);

    \node (d0 anchor) at (-0.05,3.65) {};

    \node[d0] (d0-1) [right=0cm of d0 anchor] {};
    \node[d1] (d1-1) [right=0cm of d0-1, yshift=-0.45cm] {};
    \node[dk-1] (d2-1) [right=0.5cm of d1-1, yshift=-1.0cm] {};
    \node[dk] (dk-1) [right=0cm of d2-1, yshift=-0.45cm] {};
    \node[dk+1] (d3-1) [right=0cm of dk-1, yshift=-0.45cm] {};
    \node[dm] (dm-1) [right=0.5cm of d3-1, yshift=-1.0cm] {};

    \node (ddots1) at ($(d2-1.north west)!0.5!(d1-1.south east)$) {$\ddots$}; 
    \node (ddots2) at ($(d3-1.north west)!0.5!(dm-1.south east)$) {$\ddots$}; 

    \coordinate (d0label) at ($(d0-1.east -| root)$);
    \coordinate (d1label) at ($(d1-1.east -| root)$);
    \coordinate (d2label) at ($(d2-1.east -| root)$);
    \coordinate (dklabel) at ($(dk-1.east -| root)$);
    \coordinate (d3label) at ($(d3-1.east -| root)$);
    \coordinate (dmlabel) at ($(dm-1.east -| root)$);
    \coordinate (dots1label) at ($(ddots1.east -| root)$);
    \coordinate (dots2label) at ($(ddots2.east -| root)$);

    \node [left=0cm of d0label] {$\gamma_0$};
    \node [left=0cm of d1label] {$\gamma_1$};
    \node [left=0cm of d2label] {$\gamma_{k-1}$};
    \node [left=0cm of dklabel] {$\gamma_{k}$};
    \node [left=0cm of d3label] {$\gamma_{k+1}$};
    \node [left=0cm of dmlabel] {$\gamma_{M-1}$};
    \node [left=0.2cm of dots1label] {$\vdots$};
    \node [left=0.2cm of dots2label] {$\vdots$};

    \coordinate (t0 mark) at ($(root -| dm-1.east)$);

    \draw[thick, dotted] ($(t0 mark) + (0.02,-0.1)$) node [below] {$\overset{M-1}{\underset{y=0}{\sum}} E_y$} -- ($(t0 mark) + (0.02, 3.85)$);

    \draw[thick, dotted] ($(t0 mark) + (1.0,-0.5)$) node [below] {$ c_0$} -- ($(t0 mark) + (1.0, 3.85)$);

    \coordinate (d0 anchor) at ($(dm-1.north east |- d0-1.east)$);
    \node[d0] (d0-1) [right=1.01cm of d0 anchor] {};
    \node[d1] (d1-1) [right=0cm of d0-1, yshift=-0.45cm] {};
    \node[dk-1] (d2-1) [right=0.5cm of d1-1, yshift=-1.0cm] {};
    \node[dk] (dk-1) [right=0cm of d2-1, yshift=-0.45cm] {};
    \node[dk+1] (d3-1) [right=0cm of dk-1, yshift=-0.45cm] {};
    \node[dm] (dm-1) [right=0.5cm of d3-1, yshift=-1.0cm] {};

    \node at ($(d2-1.north west)!0.5!(d1-1.south east)$) {$\ddots$}; 
    \node at ($(d3-1.north west)!0.5!(dm-1.south east)$) {$\ddots$};

    \coordinate (t1 mark) at ($(root -| dm-1.east)$);


    \draw[thick, dotted] ($(t1 mark) + (1.0,-0.3)$) node [below] {$n_0 \cdot c_0$} -- ($(t1 mark) + (1.0, 3.85)$);

    \coordinate (d0 anchor) at ($(dm-1.north east |- d0-1.east)$);
    \node[d0] (d0-1) [right=1.01cm of d0 anchor] {};

    \coordinate (t2 mark) at ($(root -| d0-1.east)$);
    \draw[thick, dotted] ($(t2 mark) + (0.0,-0.8)$) node [below] {$E_0 + n_0 \cdot c_0$} -- ($(t2 mark) + (0.0, 3.85)$);

    \coordinate (d1 anchor) at ($(dm-1.north east |- d1-1.east)$);
    \node[d1] (d1-1) [right=1.95cm of d1 anchor] {};
    \node[dk-1] (d2-1) [right=0.5cm of d1-1, yshift=-1.0cm] {};
    \node[dk] (dk-1) [right=0cm of d2-1, yshift=-0.45cm] {};
    \node[dk+1] (d3-1) [right=0cm of dk-1, yshift=-0.45cm] {};
    \node[dm] (dm-1) [right=0.5cm of d3-1, yshift=-1.0cm] {};

    \node at ($(d2-1.north west)!0.5!(d1-1.south east)$) {$\ddots$};
    \node at ($(d3-1.north west)!0.5!(dm-1.south east)$) {$\ddots$};

    \coordinate (t3 mark) at ($(root -| dm-1.east)$);
    \draw[thick, dotted] ($(t3 mark) + (0.2,-0.3)$) node [below] {$c_1 + E_0 + n_0 \cdot c_0 $} -- ($(t3 mark) + (0.2, 3.85)$);

    \coordinate (d1 anchor) at ($(dm-1.north east |- d1-1.east)$);
    \node[d1] (d1-1) [right=0.2cm of d1 anchor] {};
    \node[dk-1] (d2-1) [right=0.5cm of d1-1, yshift=-1.0cm] {};

    \node at ($(d2-1.north west)!0.5!(d1-1.south east)$) {$\ddots$};

\end{tikzpicture}

    \caption{Example \textsc{DirectAllocation} schedule for a filling class $\phi$ populated by PGTs ${Z_{\phi} = \{\gamma_0, \gamma_1, \cdots, \gamma_{M-1} \}}$. In this example, ${\minimumallocation_0 =3}$, hence task $\gamma_0$ drops out after ${n_0=2}$ executions of cycle~$c_0$. Then, one final PGA for $\gamma_0$ is scheduled. All other PGTs $\gamma_x$ in the example satisfy
    ${\minimumallocation_x > 4}$, hence they occur in at least two executions of cycle~$c_1$. Here also it is illustrated that in cycle~$c_0$, task $\gamma_0$ is the task which satisfies the maximum in \eqref{eq:cycleDuration}, hence the amount of idle time inserted in $c_1$ differs from that inserted in $c_0$.}
    \label{fig: direct allocation schedule}
\end{figure*}

\begin{lemma}\label{lemma: direct allocation schedule properties}
  Let ${\phi = (Z_{\phi}, \Pi_{\phi})}$ with ${Z_{\phi} = \{\gamma_0,...,\gamma_{M-1}\}}$ be a filling class such that ${\forall \gamma\in Z_{\phi}}$, the minimum allocation satisfies ${\minimumallocation_\gamma = \overline{N}}$. 
  Let $S$ be a schedule for $\phi$ satisfying the sequential scheduling conditions, \eqref{eq: sched condition c1} and \eqref{eq: sched condition c2}. 
  The schedule $S$ is sequentially valid if and only if ${\forall i = 0, ..., \overline{N}-2}$: 
  \begin{equation}\label{eq: condition start times PGT 0}
    s_{0, i+1} - s_{0, i} \geq \max\left(\max_{\gamma\in Z_{\phi}}\left(E_\gamma+\tminsep_\gamma\right), \sum_{\gamma\in Z_{\phi}}E_\gamma\right).
  \end{equation}
\end{lemma}

Once the start time of the very first PGA of PGT $\gamma_0$ is known, conditions \eqref{eq: sched condition c1} and \eqref{eq: sched condition c2} define a method of adding the $i^{\text{th}}$ PGA of PGTs ${\gamma_x= \{\gamma_1, \cdots, \gamma_{M-1}\}}$ to a schedule, but these conditions do not specify how to add the $(i+1)^{\text{th}}$ PGA of PGT $\gamma_0$ to the schedule. For that, \eqref{eq: condition start times PGT 0} is required to ensure the schedule produced is sequentially valid. The very first PGA of PGT $\gamma_0$ can be added at an arbitrary start time $t_0$. 

\begin{proof}
By the hypothesis, the schedule $S$ contains exactly $\overline{N}$ PGAs of each ${\gamma \in Z_{\phi}}$. Suppose ${j \leq \overline{N}-2}$. Beginning from $s_{0,j}$, the start time of the $j^{\text{th}}$ PGA of PGT $\gamma_0$, repeated application of conditions $c_1$ and $c_2$ generate the sequence of start and end times of the $j^{\text{th}}$ PGAs of each PGT $\gamma_x$,
\begin{align}
    s_{x, j} &= s_{0,j} + \sum_{y=0}^{x-1} E_{\gamma_y} \label{eq: start time PGT x}, \ \forall x \geq 1,\\
    e_{x,j} &= s_{0,j} + \sum_{y=0}^{x}E_{\gamma_y}, \ \forall x \geq 0. \label{eq: end time PGT x}
\end{align}

    ($\Rightarrow$) Suppose the schedule $S$ has properties $c_1$ and $c_2$ of a sequentially valid schedule. By~\eqref{eq: sched condition c1}, the start time for the $(j+1)^{\text{th}}$ PGA of PGT $\gamma_0$ must satisfy ${s_{0, j+1} \geq e_{M-1,j}}$. By~\eqref{eq: end time PGT x}, that is \begin{equation}
        s_{0, j+1}  \geq s_{0, j} + \sum_{y=0}^{M-1} E_{\gamma_y}.\label{eq: sum of execution times requirement}
    \end{equation} 
    By~\eqref{eq: sched condition c2}, ${\forall x \geq 0}$, \begin{equation}\label{eq: no minsep violations occus}
        s_{x, j+1} - s_{x, j} \geq E_{\gamma_x} + \tminsep_{\gamma_x}. 
    \end{equation}
    Combining (\ref{eq: start time PGT x}) and (\ref{eq: no minsep violations occus}), $\forall x \geq 1, $ \begin{align}
        s_{0, j+1} + \sum_{y=0}^{x-1}E_{\gamma_x} - (s_{0, j} + \sum_{y=0}^{x-1}E_{\gamma_x}) &\geq  E_{\gamma_x} +\tminsep_{\gamma_{x}} \nonumber \\
        s_{0, j+1} - s_{0, j} &\geq E_{\gamma_x} + \tminsep_{\gamma_{x}} \label{eq: implies max over minseps and executions}
    \end{align}
    Combining (\ref{eq: no minsep violations occus}), which holds for ${x\geq 0}$ and (\ref{eq: implies max over minseps and executions}), which holds for $x\geq 1$, \begin{equation}
        s_{0, j+1} - s_{0, j} \geq \max_{x} (E_{\gamma_x} + \tminsep_{\gamma_x}).\label{eq: set max over sum of minseps and execution times}
    \end{equation}
     The combined requirement from (\ref{eq: sum of execution times requirement}) and (\ref{eq: set max over sum of minseps and execution times}) is, \small \begin{equation*}
         s_{0, j+1} - s_{0, j} \geq \max \bigg{(}  \max_{x} (E_{\gamma_x} + \tminsep_{\gamma_x}) , \sum_{y=0}^{M-1}E_{\gamma_{y}} \bigg{)}.
     \end{equation*} \normalsize

    $(\Leftarrow)$ Suppose the start times ${s_{0, i}, \ s_{0, i+1}}$ ${ \forall i \in \{0, \cdots, \overline{N}-2 \}}$ satisfy (\ref{eq: condition start times PGT 0}). 
  
   It follows from  (\ref{eq: condition start times PGT 0}) and (\ref{eq: end time PGT x}),
   \begin{align}\label{eq: execution time based earliest start time PGA j of PGT 0}
       s_{0, j+1} &\geq \sum_{y = 0}^{M-1} E_{\gamma_x} + s_{0,j} \nonumber \\
       s_{0, j+1} & \geq e_{M-1, j}.
   \end{align}
   Condition $c_2$ guarantees that the $j^{\text{th}}$ PGAs of any PGT are non-overlapping. Then, \eqref{eq: execution time based earliest start time PGA j of PGT 0} guarantees that the $(j+1)^{\text{th}}$ and $j^{\text{th}}$ PGAs of any PGT are non-overlapping. Since $j$ is arbitrary, it follows that $S$ has property $c_1$. 
   
    In consideration of property $c_2$, it follows from (\ref{eq: condition start times PGT 0}) that, \begin{equation}\label{eq: minsep EST PGA j+1 PGT 0}
        s_{0, j+1} \geq \max_{y \geq 0} (E_{\gamma_y} + \tminsep_{\gamma_y}) + s_{0,j}.
    \end{equation}
 
    Combination of (\ref{eq: start time PGT x}) for $s_{x, j+1}$ with (\ref{eq: minsep EST PGA j+1 PGT 0}) gives,
    \begin{align}\label{eq: start times in terms of minseps}
        s_{x, j+1} &\geq \max_{y \geq 0}(E_{\gamma_y} + \tminsep_{\gamma_y}) + s_{0,j} + \sum_{y=0}^{x-1}E_{\gamma_y},
    \end{align} 
    which holds for $x \geq 1$.
    Subtracting (\ref{eq: end time PGT x}) from (\ref{eq: start times in terms of minseps}), for $x\geq 1$ gives \begin{align}
        s_{x, j+1} - e_{x,j} &\geq \max_{y \geq 0}(E_{\gamma_y} + \tminsep_{\gamma_y}) - E_{\gamma_x} \nonumber \\
        &\geq \tminsep_{\gamma_x}.
    \end{align}
    Similarly, for $x=0$, subtracting (\ref{eq: end time PGT x}) for $e_{0,j}$ from (\ref{eq: minsep EST PGA j+1 PGT 0}) gives, \begin{align}
        s_{0,j+1} - e_{0,j} &\geq \max_{y \geq 0} (E_{\gamma_{y}} + \tminsep_{\gamma_y}) - E_{\gamma_0} \nonumber \\
        & \geq \tminsep_{\gamma_0}.
    \end{align}
    Since $j$ is arbitrary, it follows that $S$ has property $c_2$, no minsep violations occur in $S$.
\end{proof}

\begin{proposition}\label{prop: timeReqSingleFC}
    Let $\phi = (Z_{\phi}, \Pi_{\phi})$ be a filling class. Let $R(Z_{\phi}) = \textsc{CalculateRequiredTime}(Z_{\phi})$ denote the output of Algorithm \ref{alg: Calculate Required Time}. Then, $R(Z_{\phi})$ is an upper bound for the minimum time required to execute a sequentially valid schedule which contains $\minimumallocation_{\gamma}$ PGAs of every PGT $\gamma \in Z_{\phi}$.
\end{proposition}

Wherever it is relevant to do so, throughout the proof of Proposition \ref{prop: timeReqSingleFC} we refer to Figure \ref{fig: direct allocation schedule}. Note that the Figure is not required for any of the arguments in the proof. The Figure simply provides a visual illustration of the concepts introduced in the proof, easing their comprehension. 
\begin{proof}
    To upper bound the minimum amount of time required to execute $\minimumallocation_{\gamma}$ PGAs of every PGT ${\gamma \in Z_{\phi}}$ such that no two PGAs are scheduled simultaneously and no minsep violations occur, it is sufficient to account for the time required, $R$, to execute a specific schedule that meets these criteria. 
    We first describe a method of scheduling PGAs based on the conditions in Lemma~\ref{lemma: direct allocation schedule properties} and which leverages the possibility of inserting idle time into a schedule. 
    We then demonstrate that this method of scheduling results in a particularly simple calculation of the required time $R$.

    Let ${M = |Z_{\phi}|}$ and let $X$ be an ordering of $Z_{\phi}$ such that ${\minimumallocation_{\gamma_x} \leq \minimumallocation_{\gamma_{x+1}}, \ \forall x \in \{0, \cdots, M-1 \}}$.
    As in Lemma~\ref{lemma: direct allocation schedule properties}, let $s_{x, i}$ and $e_{x, i}$ identify the start and end times of the $i^{\text{th}}$ PGA of PGT $\gamma_x$ added to the schedule. When the $i^{\text{th}}$ PGA of PGT $\gamma_x$ is added to the schedule, the start and end times are specified by condition $c_1$ of Lemma~\ref{lemma: direct allocation schedule properties}, \begin{equation}
        e_{x,i} = s_{x,i} + E_{\gamma_x}, \ \forall x \in \{0, \cdots, M-1\}. \label{eq: execution times give ends}
    \end{equation}
    Let the term \textit{cycle~0} refer to  a schedule consisting of the $i^{\text{th}}$ PGA from each PGT $\gamma_x \in \phi$ ($i \leq \minimumallocation_{\gamma_0}$) according to (\ref{eq: execution times give ends}) and the second sequential scheduling condition \eqref{eq: sched condition c2}, 
    \begin{equation}
        s_{x, i} = e_{x-1, i}, \ \forall x \in \{0, \cdots, M-1\}. \label{eq: back to back PGAs}
    \end{equation}
    Let cycle~0 also include an idle time of duration $I_0$ following the $i^{\text{th}}$ PGA of PGT $\gamma_{M-1}$, so that \begin{equation}
        s_{0,i+1} =  s_{0, i} + \sum_{y=0}^{M-1} E_y + I_0. 
    \end{equation}
    Note that cycle~0 has the property that no two PGAs are scheduled simultaneously. 
    By the hypothesis, supposing that the required network resources to execute a PGA are available, it is possible to schedule a PGA as long as doing so does not introduce a minsep violation.
    Since cycle~0 meets the conditions that no two PGAs are scheduled simultaneously and no minsep violations occur, it is possible to commence the schedule with an iteration of cycle~0.
    In Figure \ref{fig: direct allocation schedule}, an iteration of cycle~0 with idle time~$I_0=0$ is illustrated between time~$t_0$ and the first dotted line.
    If ${\minimumallocation_{\gamma_0} >2}$, then after the first iteration of cycle~0 all PGTs ${\gamma_x \in \phi}$ have at least 2 PGAs which remain to be scheduled. Let us consider a scheduling method where the objective is to schedule an identical repetition of cycle~0. The earliest possible start time for this repetition of cycle~0 is the earliest possibility for $s_{0,1}$. The duration of the idle time~$I_0$ included in cycle~0 is set such that $s_{0,1}$ occurs as soon as is possible, without violating the two requirements.  
    By Lemma~\ref{lemma: direct allocation schedule properties}, with ${\overline{N}=\minimumallocation_{\gamma_0}}$, to ensure the schedule meets the requirements that no two PGAs are scheduled simultaneously and no minsep violations occur, two cases can arise for $s_{0,1}$ and $I_0$. 

        \textbf{Case 1:} If \begin{equation}
            \sum_{y =0 | y \neq x}^{M-1} E_{\gamma_y} \geq \tminsep_{\gamma_x} \ \forall x \in \{0, \cdots, M-1 \}, \label{eq:noIdleTimeCondition}
        \end{equation} then,
        \begin{align}\label{eq:cycleZeroExec}
            s_{0,1} &= \sum_{y=0}^{M-1} E_{\gamma_y}, \\
            I_0 &= 0. \nonumber
        \end{align} 
      
        \textbf{Case 2:} If \eqref{eq:noIdleTimeCondition} doesn't hold, there exists at least one task $x$ such that \begin{equation}\label{eq:cycleZeroMinSep}
             \sum_{y =0 | y \neq x}^{M-1} E_{\gamma_y} < \tminsep_{\gamma_x}.
        \end{equation} Then, \begin{align}\label{eq:cycleZeroMinSepEST}
            s_{0,1} &= \max_{y \in \{0, \cdots, M-1 \}} \big{(} E_{\gamma_y} + \tminsep_{\gamma_y} \big{)}, \\
            I_0 &= \max_{y \in \{0, \cdots, M-1 \}} (E_{\gamma_y} + \tminsep_{\gamma_y}) - \sum_{y=0}^{M-1}E_{\gamma_y}. \nonumber 
        \end{align} 
        
        Figure~$\ref{fig: direct allocation schedule}$ illustrates a scenario where case~2 applies. The first dotted line in the figure indicates the time given by (\ref{eq:cycleZeroExec}), and the second dotted line indicates the time given by (\ref{eq:cycleZeroMinSepEST}).

        Let the term \textit{cycle~m} refer to a schedule consisting of the $i^{\text{th}}$ PGA from each PGT $\gamma_x \in \phi$ such that \begin{align}
            e_{x, i} &= s_{x,i} + E_{\gamma_x}, \ \forall x \in \{m, \cdots, M-1 \}\label{eq: cycle y PGA end times} \\
            s_{x,i} &= e_{x-1, i}, \ \forall x \in \{m, \cdots, M-1 \},\label{eq: cycle y PGA start times}
        \end{align} and the $i^{\text{th}}$ PGA of PGT $\gamma_{M-1}$ is followed by idle time $I_m$ so that \begin{equation}
            s_{m, i+1} = s_{m, i} + \sum_{y=m}^{M-1} E_{\gamma_y} + I_m. \label{eq: cycle m includes idle time Im}
        \end{equation} 
        Note that cycle~m no longer includes PGAs for the PGTs $\{\gamma_0, \cdots, \gamma_m\}$.

        The basis of the scheduling method is to continue scheduling cycles of 
        PGA executions according to (\ref{eq: cycle y PGA end times}), (\ref{eq: cycle y PGA start times}) and (\ref{eq: cycle m includes idle time Im}), which preserve the ordering $X$. However, we only need to schedule $\minimumallocation_{\gamma_x}$ PGAs for PGT $\gamma_x$. Therefore, it is necessary to account for when a PGT $\gamma_x$ \textit{drops out} of a cycle, meaning that $\minimumallocation_{\gamma_x} -1$ PGAs of the PGT have already been scheduled by previous cycles. 
        After a drop out point, the PGT that is dropping out needs to be scheduled a final time. 
        Earlier we assumed $\minimumallocation_{\gamma_0} > 2$. In accounting for when PGTs drop-out, we address the other possibilities, $\minimumallocation_{\gamma_0} = 1, 2$. 
       Due to the ordering $X$, it will always be the case that the PGT or PGTs dropping out at a drop out point are the first PGTs of the cycle. 
        Figure \ref{fig: direct allocation schedule} illustrates a scenario wherein $\minimumallocation_{\gamma_0}=3$, hence it is indicated that cycle~$c_0$ is repeated $n_0=2$ times, after which point $\gamma_0$ drops out, marking a transition in the schedule from cycle~0 to cycle~1. The final execution of $\gamma_0$ is scheduled before cycle~1 begins.

        It follows from (\ref{eq: cycle y PGA end times}), (\ref{eq: cycle y PGA start times}), (\ref{eq: cycle m includes idle time Im}) and Lemma \ref{lemma: direct allocation schedule properties} (with $\overline{N}=\minimumallocation_{\gamma_m}$) that if the duration $c_m$ of cycle~m is  
        \begin{equation}\label{eq:cycleDuration}
            c_m = \max \Big{(} \underset{y \in \{m, \cdots M -1\} }{\max} \big{(} E_{\gamma_y} + \tminsep_{\gamma_y} \big{)},  \sum_{y=m}^{M-1} E_{\gamma_y}\Big{)},
        \end{equation}
        then the schedule has the properties that no two PGAs are scheduled simultaneously, and no minsep violations occur.

        Exactly $n_m$ iterations of cycle~m, with duration $c_m$, occur before the subsequent drop out point. The first task (task zero) will drop out after \begin{equation}
            n_0 = \minimumallocation_{\gamma_0} - 1
        \end{equation} cycles. In general, $n_{m}$ cycles of duration $c_m$ will be scheduled by our scheduling method, with $n_{m}$ given by,
        \begin{equation}\label{eq:numCycles}
            n_{m} = \minimumallocation_{\gamma_m} - \minimumallocation_{\gamma_{m-1}}.
        \end{equation}

        After $\sum_{y=0}^{x} n_y$ cycles, $\minimumallocation_x - 1$ PGAs for PGT $\gamma_x$ have been added to the schedule, since \begin{align*}
            \sum_{y=0}^{x} n_y &= n_0 + \sum_{y=1}^{x} n_m \\
            &= \minimumallocation_{\gamma_0} - 1 + \sum_{y=1}^{x} (\minimumallocation_{\gamma_y} - \minimumallocation_{\gamma_{y-1}}) \\
            &= \sum_{y=0}^{x} \minimumallocation_{\gamma_y} - \sum_{y=0}^{x-1} \minimumallocation_{\gamma_y} - 1 \\
            &= \minimumallocation_{\gamma_x} - 1.
        \end{align*} 
        After it's drop out point, each task is scheduled one final time. In total these individual executions contribute an amount of time $\underset{y=0}{\overset{M-1}{\sum}} E_{\gamma_y}$ to the required time.
        In total, following our scheduling method, the time required to execute $\minimumallocation_{\gamma}$ PGAs of every task $\gamma \in Z_{\phi} $ is \begin{align}
            R(Z_{\phi}) &= \sum_{x=0}^{M-1} ( n_{x} \cdot c_{x} + E_{\gamma_x})\\
            &= \textsc{CalculateRequiredTime}(Z_{\phi}). \label{eq:RphiIsCalcReqTime}
        \end{align}
\end{proof}

\subsection{The Compute Schedule process produces valid schedules}
In this section we first extend the definition of sequentially valid schedules, which is only suitable for schedules covering a single filling class, to a definition of valid schedules for multiple filling classes. To do so, we first define the notion of a resource conflict.  
Then, we analyze the Compute Schedule process (Algorithm \ref{alg: Compute Schedule}). This process includes a minimal allocation phase that relies on the \textsc{DirectAllocation} (Algorithm \ref{alg: direct allocation subroutine}) subroutine and a bonus allocation phase that relies on the \textsc{RoundRobinBonus} (Algorithm \ref{alg: bonus round robin resource schedule}) subroutine. We prove that the schedules produced by Compute Schedule are valid and always include a minimum allocation of PGAs for every active PGT.

\begin{definition}[Resource Conflict]\label{def: resource conflict}
Let $\Phi$ be a set of filling classes and $Z_{\Phi}$ the set of active PGTs in $\Phi$. Let $S$ be a schedule for the PGTs $Z_{\Phi}$ with start times $s_{\gamma, i}$ and end times $e_{\gamma, i}$ for the $i^{\text{\emph{th}}}$ scheduled PGA of PGT ${\gamma, \ \forall \gamma \in Z_{\Phi}}$. We say a \emph{\textbf{resource conflict}} occurs in $S$ if 
${\exists \gamma, \gamma', \ \exists i, j}$ such that 
\begin{equation*}
    \xi(\pi_\gamma) \cap \xi(\pi_\gamma') \neq \emptyset,
\end{equation*} and \begin{equation*}
    s_{\gamma,i} \leq s_{\gamma', j} < e_{\gamma, i}.
\end{equation*}
\end{definition}

\begin{definition}[Valid schedule]
    A schedule is said to be \textbf{valid} if it has the following two properties: \begin{enumerate}
        \item No resource conflicts occur, and
        \item No minsep violations occur.
    \end{enumerate}
\end{definition}

Note that the first condition of a valid schedule, prohibiting resource conflicts, is a weaker condition than the first condition of a sequentially valid schedule, prohibiting any two PGAs from being scheduled simultaneously. The later condition guarantees that there are no resource conflicts within the same filling class, but it would also prevent simultaneously scheduling PGAs from different filling classes, even if they do not have overlapping resource requirements. 

\begin{corollary}\label{corr: DirectAllocReqTime}
  The \textsc{DirectAllocation} scheduling algorithm (Algorithm \ref{alg: direct allocation subroutine}) implements the specific scheduling method described in the proof of Proposition~\ref{prop: timeReqSingleFC}.
\end{corollary}

\begin{proof}
    The result follows by direct comparison of Algorithm~\ref{alg: direct allocation subroutine} and the proof of Proposition~\ref{prop: timeReqSingleFC}.
\end{proof}

By Corollary~\ref{corr: DirectAllocReqTime} and Proposition~\ref{prop: timeReqSingleFC}, $\textsc{DirectAllocation}$ produces a valid schedule for a filling class ${\phi=(Z_{\phi}, \Pi_{\phi})}$ which requires an execution time $R(Z_{\phi}) = \textsc{CalculateRequiredTime}(Z_{\phi})$, as given by the output of Algorithm~\ref{alg: Calculate Required Time}. 
We subsequently refer to this method of scheduling simply as direct allocation.

Proposition \ref{prop: timeReqSingleFC} makes it possible to account for the duration of a valid schedule for a single filling class. To progress towards proving Theorem~\ref{thm: min alloc guaranteed}, it is necessary to account for the duration of a valid schedule for a set of filling classes. The next results build towards the required accounting.

It is possible for a set of filling classes ${\Theta = \{\theta_0, \theta_{1}, \cdots, \theta_{M-1} \}}$ to have the property of being totally ordered, according to some ordering $\geq$. If $\Theta$ is totally ordered by an arbitrary ordering $\geq$, then ${\theta_{i-1} \geq \theta_{i} \geq \theta_{i+1} \ \forall i \in \{0, 1, \cdots, M-1 \}}$. In general a set of filling classes $\Phi$ may only be partially ordered by an arbitrary ordering $\geq$, in which case there exist incomparable elements $\phi$ and $\psi$ such that neither ${\phi \geq \psi}$, nor ${\psi \geq \phi}$.

\begin{proposition}\label{prop: totally ordered good accounting}
    Let $\Theta$ be a totally ordered set of filling classes conforming to Definition~\ref{def: filling classes}. Let $Z_{\Theta}$ denote the set of active PGTs in $\Theta$. The minimum time required to execute a valid schedule in which there are $\minimumallocation_{\gamma}$ PGAs of every PGT ${\gamma \in Z_{\Theta}}$ is upper bounded by the quantity
    \begin{equation}\label{eq:minimalAllocReqtime}
        \tilde{\mathcal{R}}(Z_{\Theta}) = \sum_{\theta_l \in \Theta} R(Z_{\theta_l}),
    \end{equation}
    where $R(Z_{\theta_l}) = \textsc{CalculateRequiredTime}(Z_{\theta_l})$ is the output of Algorithm~\ref{alg: Calculate Required Time}.
\end{proposition}

\begin{proof}
     To upper bound the minimum amount of time required to execute ${\minimumallocation_{\gamma}}$ PGAs of every PGT ${\gamma \in Z_{\Theta}}$, we build upon the proof of Proposition~\ref{prop: timeReqSingleFC} and calculate the time required to execute a specific schedule for $\Theta$ including ${\minimumallocation_{\gamma}}$ executions of each PGT ${\gamma \in Z_{\Theta}}$.
     
     Let $M$ be the number of filling classes in $\Theta$. 
     By the hypothesis, $\Theta$ is totally ordered, according to some particular ordering $\geq$, such that $\theta_{0} \geq \theta_1 \geq \cdots \geq \theta_{M-1}$. 
     The only restrictions which can disallow scheduling a PGA at a certain time are that it would introduce a resource conflict or a minsep violation.
     By Lemma~\ref{lemma: one fill class per PGT}, every PGT $\gamma$ is associated with a single filling class. Since minsep violations may only occur between PGAs of the same PGT, no minsep violations may occur between PGTs in different filling classes. Hence, in building a schedule for $\Theta$ it is always possible to schedule back-to-back executions of a sequentially valid schedule for filling class $\theta_l$ and a sequentially valid schedule for filling class $\theta_{l+1}$. 
     With the given ordering, the greatest element is scheduled first, and back-to-back executions mean that the start time of the schedule for filling class $\theta_{l+1}$ is equal to the end time of the schedule for filling class $\theta_{l}$. 
     In this scheme no two PGAs are ever scheduled simultaneously, hence no resource conflicts can occur.

     To construct a specific schedule for $\Theta$, it is thus sufficient to proceed by using direct allocation scheduling to schedule $\minimumallocation_\gamma$ executions of all PGTs $\gamma \in Z_{\theta_0}$. By Proposition \ref{prop: timeReqSingleFC}, this schedule has duration $R(Z_{\theta_0})$. To continue construction of the schedule for $\Theta$, use direct allocation scheduling to
     construct a schedule for each filling class ${\theta_l, \ l \in \{2, \cdots, M-1 \}}$ and append this schedule to the end of the schedule for filling class $\theta_{l-1}$. The total time required to execute this specific schedule for $\Theta$ is \begin{equation*}
         \tilde{\mathcal{R}}(Z_{\Theta}) = \sum_{l =0}^{M-1} R(Z_{\theta_l}).
    \end{equation*}
\end{proof}

The ordering $\geq_{\xi}$ introduced in Definition~\ref{def: comparison of filling classes} serves as an example of a specific ordering which may totally order the filling classes $\Theta$ in Proposition~\ref{prop: totally ordered good accounting}.
Two filling classes are incomparable under $\geq_{\xi}$ whenever ${\xi(\phi) \cap \xi(\psi) = \emptyset}$.
The following theorem extends the result of Proposition~\ref{prop: totally ordered good accounting} from a set of totally ordered filling classes to a set of well-behaved filling classes, partially ordered by $\geq_{\xi}$, as in Definition~\ref{def: comparison of filling classes}. 
The extension is non-trivial, and requires accounting for how the totally ordered subsets of $\Phi$ can be consistently added to a schedule without introducing resource conflicts.

A set $\Phi$ that is partially ordered by an order $\leq$ contains totally ordered subsets. If ${\Theta = \{\theta_0, \theta_1, \theta_2 \}}$ is a totally ordered subset of $\Phi$, then $\{\theta_0 \}$, $\{\theta_1 \}$, $\{\theta_2 \}$,  $\{\theta_0, \theta_1\}$, $\{\theta_1, \theta_2 \}$, and $\{\theta_0, \theta_2 \}$ are also totally ordered subsets of $\Phi$. In this example, only $\Theta$ is a candidate for a maximally sized totally ordered subset of $\Phi$, as all of the other totally ordered subsets listed are subsets of $\Theta$. 

\begin{theorem}(Good Accounting)
    \label{thm: good accounting applies to partially ordered FCs} 
    Let $\Phi$ be a set of well-behaved filling classes, partially ordered by $\geq_{\xi}$. Let $Z_{\Phi}$ denote the set of active PGTs in $\Phi$.
    Let ${\Omega = \{\Theta_0, \Theta_1, \cdots, \Theta_n \}}$ be the set of all totally ordered subsets of $\Phi$, of which there are a number $n$, for some ${n \in \mathbb{N}, \ n \geq 0}$.
    The minimum time required to execute a valid schedule in which there are $\minimumallocation_{\gamma}$ PGAs of every PGT ${\gamma \in Z_{\Phi}}$ is upper bounded by 
    \begin{equation}\label{eq: partiallyOrderedminimalAllocReqTime}
        \mathcal{R}(Z_{\Phi}) = \max_{\Theta_i \in \Omega} \tilde{\mathcal{R}}(Z_{\Theta_i}), 
    \end{equation} where $\tilde{\mathcal{R}}(Z_{\Theta_i})$ is given by \emph{\eqref{eq:minimalAllocReqtime}}.
\end{theorem}

\begin{proof}(Good Accounting.)
Suppose that path partition $\Pi$ describes an internally connected network.
Let $\Omega^{\max}$ be the set of maximally sized elements of $\Omega$, \begin{equation}
    \Omega^{\max} := \{\Theta_i \in \Omega \ | \ \forall i' \neq i, \ \Theta_i \ \cancel{\subset} \ \Theta_{i'} \}, 
\end{equation}
where the inclusion $\subset$ is the usual set inclusion. 
Whenever ${\Omega \neq \emptyset}$, it follows by definition that ${\Omega^{\max} \neq \emptyset}$. Note that ${\forall \phi \in \Phi}$, ${\exists \omega \in \Omega^{\max}}$ such that $
{\phi \in \omega}$.
Therefore, $\Omega^{\max}$ contains a total number of elements $m$ such that ${0 \leq m \leq n}$. 
Denote the elements as ${\Omega^{\max} = \{\omega_0, \omega_1, \cdots, \omega_{m-1} \}}$. 
Consider elements ${\omega, \omega' \in \Omega^{\max}}$, where ${\omega = \{\phi_{0}, \dots, \phi_{k-1}\}}$ consists of $k$ filling classes ordered by $\geq_{\xi}$, and ${\omega' = \{\phi'_{0}, \cdots, \phi'_{k'-1}\}}$ consists of $k'$ filling classes ordered by $\geq_{\xi}$. The ordering implies that $\phi_0$ is the greatest element of $\omega$ under $\geq_{\xi}$ and $\phi'_0$ is the greatest element of $\omega'$ under $\geq_{\xi}$. As a preliminary, we prove the following claim: 

\bigskip

\textbf{Claim 1:} Whenever ${\omega \cap \omega' \neq \emptyset}$, then for some index ${0 \leq j \leq \min(k, k')}$, $\omega$ and $\omega'$ are identical up to this $j^{\text{th}}$ element. That is, ${\phi_{0}=\phi'_0, \ \phi_1 = \phi'_1, \cdots, \ \phi_{j-1} = \phi'_{j-1}}$.

\bigskip
 We proceed recursively and by contradiction. In the base case, suppose for the sake of contradiction that $\phi_0 \neq \phi'_{0}$. 
 Since $\omega \cap \omega' \neq \emptyset$, $\exists \phi $ such that $\phi \in \omega$ and $\phi \in \omega'$. Since $\phi_0, \ \phi'_0$ are respectively the greatest elements of $\omega$ and $\omega'$ under the ordering $\geq_{\xi}$, 
 \begin{align}
     \xi(\Pi_\phi) \subset \xi(\Pi_{\phi_0}), \label{eq: phi zero inclusion} \\
     \xi(\Pi_\phi) \subset \xi(\Pi_{\phi'_0}).\label{eq: phi prime zero inclusion}
 \end{align}
 Therefore, ${\xi(\Pi_{\phi_0}) \cap \xi(\Pi_{\phi'_0}) \neq \emptyset}$. There are three possible cases. In the first case, $\xi(\Pi_{\phi_0}) = \xi(\Pi_{\phi'_0})$, and we arrive at the contradiction that ${\phi_0 = \phi'_0}$, since $\Phi$ is well-behaved.
 In the second case, \begin{equation}
     \xi(\Pi_{\phi_0}) \subset \xi(\Pi_{\phi'_0}). \label{eq: phi zero inclusion in phi prime zero}
 \end{equation} Then, \begin{equation}
     \omega'' = \{\phi'_0, \phi_0, \phi_1, \cdots, \phi_{k-1} \} \label{eq: omega '' second case}
 \end{equation} is a totally ordered subset of $\Phi$ such that $\omega \subset \omega'',$ and we arrive at the contradiction that $\omega \ \cancel{\in} \ \Omega^{\max}$.
 In the third case, $\xi(\Pi_{\phi'_0}) \subset \xi(\Pi_{\phi_0})$ and the argument follows the same logic as in the second case. 

If ${(\omega \setminus \phi_0) \cap (\omega' \setminus \phi'_0) = \emptyset}$, then $\omega, \ \omega'$ are only identical up to the $1^{st}$ element. Otherwise, we can repeat the argument from the base case with $\omega \setminus \phi_0$ and $\omega \setminus \phi'_0 $ as the basis. This recursive argumentation can be repeated until we arrive at an element $j\leq \min(k, k')$ such that ${(\omega \setminus \phi_{j-1}) \cap (\omega' \setminus \phi'_{j-1}) = \emptyset}$.

\bigskip
As a remark, note that  a singleton $\omega = \phi \in \Phi$ can be a maximal element of $\Omega$ as long as ${\xi(\Pi_{\phi}) \cap \xi(\Pi_{\phi'}) = \emptyset, \ \forall \phi' \neq \phi \in \Phi}$.

\bigskip
We now proceed to construct a schedule $S$ for $\Phi$, with start time $t_0$, by constructing schedules $\forall \omega \in \Omega^{\max}$. 
For every $\omega \in \Omega^{\max}$, identify $\alpha_{\omega}$, the unique greatest element of $\omega$ according to the order $\geq_{\xi}$. 
The existence of this unique greatest element is guaranteed by Lemma \ref{lemma: unique greatest element of Phi}. Let $S_{\omega}$ be a direct allocation schedule for $\alpha_{\omega}$ with start time ${s_{\alpha_{\omega}} = t_0}$. 
By Corollary~\ref{corr: DirectAllocReqTime}, the end time of the schedule $S_{\omega}$ is ${e_{\alpha_{\omega}} = t_0 + R(\alpha_{\omega})}$. Define $\tilde{\omega} = \omega \setminus \alpha_{\omega}$, a temporary update of $\omega$ excluding the greatest element. 
Now, by Lemma~\ref{lemma: unique greatest element of Phi}, $\tilde{\omega}$ has a unique greatest element $\beta_{\omega}$, according to the order $\geq_{\xi}$. 
Append to schedule $S_{\omega}$ a direct allocation schedule for $\beta_{\omega}$ with start time ${s_{\beta_{\omega}} = e_{\alpha_{\omega}}}$. 
By Corollary~\ref{corr: DirectAllocReqTime}, the end time of the schedule $S_{\omega}$ is ${e_{\beta_{\omega}} = t_0 + R(\alpha_{\omega}) + R(\beta_{\omega})}$. 
Update ${\tilde{\omega} = \omega \setminus \{\alpha_{\omega}, \beta_{\omega} \}}$. 
Continue in this manner until there are no remaining filling classes in $\tilde{\omega}$, i.e. until ${\tilde{\omega} = \emptyset}$. 
Then, $S_{\omega}$ is a schedule for $\omega$ constructed in the manner of the proof of Proposition~\ref{prop: totally ordered good accounting} which has start time $t_0$ and end time ${e_{\omega} = t_0 + \underset{\phi \in \omega}{\sum} R(\phi_{\omega})}$. Hence the required execution time of $S_{\omega}$ is $\tilde{R}(\omega)$, as given by \eqref{eq:minimalAllocReqtime}.

Claim 1 ensures that $\forall \omega, \omega' \in \Omega^{\max}$ the schedules $S_{\omega}$, $S_{\omega'}$ are consistent, meaning if $ \exists \phi \in \omega \cap \omega'$, then the start time $s_{\phi_{\omega}}$ assigned to $\phi$ in schedule $S_{\omega}$ is the same as the start time $s_{\phi_{\omega'}}$ assigned to $\phi$ in schedule $S_{\omega'}$, \begin{equation}
    s_{\phi_{\omega}} = s_{\phi_{\omega'}}. \label{eq: consistent start times}
\end{equation}
It follows immediately that \begin{equation}
    e_{\phi_{\omega}} = e_{\phi_{\omega'}},
\end{equation} since in both $S_{\omega}$ and $S_{\omega'}$, the schedule for $\phi$ is a direct allocation schedule with start times obeying \eqref{eq: consistent start times}.

Therefore, $\forall \omega, \omega' \in \Omega^{\text{max}}$, the schedules $S_{\omega}, \ S_{\omega'}$ are either independent, which happens when $\omega \cap \omega' = \emptyset$, or else the schedules are consistent. 
To construct the full schedule $S$,  it is therefore sufficient to loop through the schedules $S_{\omega}, \ \forall \omega \in \Omega^{\text{max}}$. For every schedule $S_{\omega}$, $\forall \phi \in \omega$, if $S$ does not already include the start time $s_{\phi}$ and end time $e_{\phi}$ for $\phi$, these start and end times are added to $S$. This method guarantess that no resource conflicts occur in $S$. Furthermore, since $\forall \phi \in \Phi$, $\exists \omega \in \Omega^{\max}$ such that $
\phi \in \omega$, the schedule $S$ constructed in this manner includes a direct allocation schedule $\forall \phi \in \Phi$. Hence, the schedule $S$ includes $\minimumallocation_{\gamma}$ PGAs of every PGT $\gamma \in Z_{\Phi}$ such that no minsep violations occur.

Let $\omega^{*}\in \Omega^{\max}$ indicate a set with schedule $S_{\omega^{*}}$ which is a solution to  \begin{equation}\label{eq: max schedule duration}
    S_{\omega^{*}} = \max_{\omega \in \Omega^{\max}} \tilde{R}(Z_{\omega}),
\end{equation}
where $Z_{\omega}$ denotes the set of PGTs which occupy the set of filling classes $\omega$. 
Whenever $\Omega^{\max} \neq \emptyset$, there exists such an $\omega^{*}, \ S_{\omega^{*}}$, which need not be unique. The amount of time required to execute the schedule $S$ is the amount of time required to execute $S_{\omega^{*}}$. This concludes the proof in the case that the path partition $\Pi$ describes an internally connected network. 

\bigskip
In the case that the path partition $\Pi$ describes a network that is disconnected, partition $\Pi$ into it's internally connected sub-networks. The preceding proof then applies to each internally connected sub-network. Let $\{S\}$ denote the set of schedules produced for the set of internally connected sub-networks.  The schedules $S \in \{S\}$ are independent and have no resource conflicts since the schedules are for sub-networks of the path partition, which have no resources in common. Each schedule $S \in \{S\}$ has start time $t_0$ and some end time $e_{S} \geq t_0$. Therefore, the set of schedules $S \in \{S\}$ constructs a  schedule $\mathfrak{S}$ for $\Phi$ in which $\minimumallocation_{\gamma}$ PGAs of every PGT $\gamma \in \Phi$ are scheduled, no minsep conflicts occur, and no resource conflicts occur.
Moreover, ${\exists S^{*}}$ with end time $e_{S^{*}}$ which satisfies ${e_{S^{*}} = \underset{{S \in \{S\}}}{\max} \ e_{S}}$. The schedule $\mathfrak{S}$ has required execution time ${e_{S^{*}} = \mathcal{R}(\Phi)}$, given by \eqref{eq: partiallyOrderedminimalAllocReqTime}. 

\end{proof}


\begin{corollary}\label{corr:minimumAllocTimeReq}
The \textsc{MinimalAllocation} function (Algorithm \ref{alg: Compute Schedule}) implements the specific scheduling method described in the proof of Theorem~\ref{thm: good accounting applies to partially ordered FCs} (Good Accounting). 
\end{corollary}

\begin{proof}
    The result follows by direct comparison of Algorithm~\ref{alg: Compute Schedule} and the proof of Theorem~\ref{thm: good accounting applies to partially ordered FCs}. In particular, note that the \textsc{MinimalAllocation} function in \textsc{ComputeSchedule} makes use of the partial ordering of the filling classes in building 
    the schedule $S$ with start time $t_0$. 
    To construct $S$, ${\forall \phi \in \Phi}$, a direct allocation schedule for $\phi$ is inserted into $S$ with start time $s_{\phi}$ and end time ${e_{\phi} = s_{\phi} + R(\phi)}$.
    The start time $s_\phi$ is set according to \begin{equation}\label{eq: set start phi}
        s_{\phi} = t_0 + \sum_{\phi' >_{\xi} \phi} R(\phi').
    \end{equation}
    The identification of ${\phi' >_{\xi} \phi}$, implicitly describes that \textsc{MinimalAllocation} identifies the subset ${\omega\subset \Omega}$ of all totally ordered subsets of $\Phi$ for which ${\omega \cap \phi \neq 0}$. By~\eqref{eq: set start phi}, ${\forall \phi' \geq_{\xi} \phi}$, ${s_{\phi'} < s_{\phi}}$ and moreover 
    \begin{equation}
        s_{\phi} = \max_{\phi' >_{\xi} \phi} e_{\phi'},
    \end{equation}
    which is exactly how $s_\phi$ is assigned in the proof of Theorem~\ref{thm: good accounting applies to partially ordered FCs}.
\end{proof}

As a consequence of Theorem~\ref{thm: good accounting applies to partially ordered FCs}, given the set of filling classes $\Phi$ as input, the minimal allocation phase of Compute Schedule (Algorithm \ref{alg: Compute Schedule}) produces a valid schedule which includes $\minimumallocation_{\gamma}$ PGAs of every PGT $\gamma \in \Phi$. The schedule requires execution time $\mathcal{R}(\Phi)$, as given by (\ref{eq:minimalAllocReqtime}). 
We subsequently refer to this method of scheduling simply as minimal allocation scheduling.

\subsubsection{Admit Tasks is Sound}\label{ssapp: Admin Control is Sound}

For the Admit Tasks process to be sound, it must guarantee that newly accepted demands do not disrupt existing service agreements (\ref{consideration: AC prevents overload}). 
Before proceeding, it is useful to introduce definitions describing the availability of resources and how much time has already been reserved on a network resource. 

\begin{definition}[Resource Availability]
    A resource is \emph{available} between times $t_1$ and $t_2$ if there are no scheduled PGAs for any PGTs at any time between $t_1$ and $t_2$. The total amount of \emph{\textbf{time available}} on each resource ${r \in \mathfrak{R}}$ is a non-negative number \begin{equation}\label{eq:AvailabelTimeResr}
        T_{\xi_r} \in [0, T^{SI}].
    \end{equation}
\end{definition}

\begin{definition}[Time Reserved on a Resource]\label{def:timeResResource}
Let ${r \in \mathfrak{R}}$ and let $T_{\xi_r}$ be the total amount of time available on $r$. A duration of time $T$ is \emph{\textbf{reserved on resource}} $r$ by decrementing the available time by $T$,
\begin{equation}\label{eq:reserveTimeT}
    T_{\xi_{r}} \leftarrow T_{\xi_r} - T.
\end{equation}
We say it is \emph{possible} to reserve time $T$ on resource $r$ if \emph{(\ref{eq:reserveTimeT})} is positive.
A reservation of time with duration $T$ may be \emph{\textbf{canceled}} by incrementing the available time on resource $r$ by $T$,
\begin{equation*}
    T_{\xi_r} \leftarrow T_{\xi_r} + T.
\end{equation*}
\end{definition}

Note that in the proofs of Propositions~\ref{prop: timeReqSingleFC} and~\ref{prop: totally ordered good accounting} it was assumed that the required network resources to execute a PGT are available. It remains to be shown that Algorithm~\ref{alg: Admit Tasks} guarantees that a PGT is only accepted if it is possible to reserve sufficient time on the required network resources.

\begin{theorem}[Admit Tasks is Sound]\label{thm: admin control is sound} 
Let $\gamma' \in \bm{\Gamma}$, let $\Phi$ be a well-behaved set of filling classes with the set of active PGTs $Z_{\Phi}$, and suppose ${\exists \phi \in \Phi}$ such that $\pi_{\gamma'} \in \Pi_{\phi}$. 
Then, $\gamma'$ is accepted by Admit Tasks as implemented by Algorithm~\ref{alg: Admit Tasks} if the time required to execute a minimal allocation schedule for $Z_{\Phi} + \gamma'$ is less than or equal to the scheduling interval. That is, \begin{equation}\label{eq:reqLessThanSI}
    \mathcal{R}(Z_{\Phi} + \gamma') \leq T^{SI},
\end{equation} where $\mathcal{R}(\cdot)$ is as in \eqref{eq: partiallyOrderedminimalAllocReqTime}.
\end{theorem}

\begin{proof}
    Throughout, let $\mathcal{R}(\Theta)$ denote the time required to execute a minimal allocation schedule for a set of filling classes $\Theta$, as given by \eqref{eq: partiallyOrderedminimalAllocReqTime}, and let $R(\theta)$ denote the time required to execute a direct allocation schedule for a single filling class $\theta$, as given by~\eqref{eq:RphiIsCalcReqTime} .

    The set of active PGTs $Z_{\Phi}$ satisfies $\mathcal{R}(Z_{\Phi}) \leq T^{SI}$, as every active PGT was previously accepted by Admit Tasks. 
    In the special case where there are no active PGTs, i.e. $Z_{\Phi}=\emptyset$ and $N=|Z_{\Phi}|$ is zero, this still holds as \begin{equation*}
        \mathcal{R}(Z_{\Phi}=\emptyset)=0 \Rightarrow \mathcal{R}(Z_{\Phi}) \leq T^{SI}.
    \end{equation*}
    In that case also \begin{equation*}
        Z_{\phi} = \emptyset  \wedge R(Z_{\phi}) = 0 \ \forall \phi \in \Phi.
    \end{equation*}
     This situation must occur at least the very first time Admit Tasks executes, and may occur in some scheduling interval $k$ if all active PGTs expire or are terminated by Update Filling Classes, which precedes Admit Tasks. 
     
    In the first step of \textsc{AdmitTasks}, the available time $T_{\xi_r}$ on each resource $r \in \mathfrak{R}$ is reset to the duration of the scheduling interval, $T^{SI}$. 
    For each filling class ${\phi \in \Phi}$, the amount of time $R(Z_{\phi})$ required to execute a direct allocation schedule is calculated. Then, ${\forall r \in \xi(\Pi_{\phi})}$ the time reserved on $r$ is incremented by $R(Z_{\phi})$. 
    By Corollary \ref{corr:minimumAllocTimeReq}, sequentially reserving the time required to execute a direct allocation schedule on a resource ${r \in \xi(\Pi_{\phi})}, \ \forall \phi$, is equivalent to reserving time on the the resource according to a minimal allocation schedule.
    
    Let $M^{Z_{\Phi}}$ denote the maximum amount of time reserved on any resource for the set of active PGTs $Z_{\Phi}$, \begin{equation*}
        M^{Z_{\Phi}} := \max_{r \in \mathfrak{R}} (T^{SI} - T_{\xi_r}).
    \end{equation*}
    By Corollary~\ref{corr:minimumAllocTimeReq} this is equal to the duration of the minimal allocation schedule for $Z_{\Phi}$. Overall,
    \begin{equation*}
        M^{Z_{\Phi}} = \mathcal{R}(Z_{\Phi}) \leq T^{SI}.
    \end{equation*}
    At this point, the task $\gamma'$ is considered by \textsc{AdmitTasks}. 
    By the hypothesis, ${\exists \phi \in \Phi}$ such that ${\pi_{\gamma'} \in \Pi_{\phi}}$, so \textsc{GetFillingClass} returns the filling class $\phi$. 
    The next step of \textsc{AdmitTasks} is to add the PGT $\gamma'$ to $\tilde{Z}_{\phi} = Z_{\phi} + \gamma'$, a temporary update of the set of active PGTs in filling class $\phi$. 
    The \texttt{Accept} variable is initialized as True. 
    The time reservation of amount $R(Z_{\phi})$ is then canceled on every resource $r \in \xi(\Pi_{\phi})$. 
    Next, the time $R(\tilde{Z}_{\phi})$ required to execute a direct allocation schedule for $\tilde{Z}_{\phi}$ is calculated.
    
    By the proof of Proposition \ref{prop: timeReqSingleFC} and Corollary \ref{corr: DirectAllocReqTime}, the direct allocation schedule for $\phi$ may include idle time based on the values of $\tminsep_{\gamma}, \ \forall \gamma \in Z_{\phi}$. For this reason, it is possible that the time required to execute a direct allocation schedule of $\tilde{Z}_{\phi}$, $R(\tilde{Z}_{\phi})$ is not substantially greater than $R(Z_{\phi})$, and it may be that $R(\tilde{Z}_{\phi}) = R(Z_{\phi})$.
    Canceling the time reservation of $R(Z_{\phi})$ allows for correct comparison of $R(\tilde{Z}_{\phi})$ to the available time on all resources $r \in \xi(\Pi_{\phi})$. Overall, $R(Z_{\phi}) \leq R(\tilde{Z_{\phi}})$.

    For every resource $r \in \xi(\Pi_{\phi})$, the time $T_{\xi_r}$ available on resource $r$ is compared to the required time $R(\tilde{Z}_{\phi})$.
    If $\exists r \in \xi(\Pi_{\phi})$ such that \begin{equation}
        T_{\xi_r} - R(\tilde{Z}_{\phi}) < 0,\label{eq: test required time}
    \end{equation}
     then the \texttt{Accept} variable is set to false, the PGT $\gamma$ is removed from $\bm{\Gamma}$, and $\gamma'$ receives a reject decision.
    Otherwise, the value of \texttt{Accept} remains true, the set of active PGTs $Z_{\phi}$ is updated to $\tilde{Z}_{\phi}$, and the available time on each associated resource $r \in \xi(\tilde{Z}_{\phi})$ is decremented by the required time for the filling class, $R(\tilde{Z}_{\phi})$. 
    Define $M^{Z_{\Phi} + \gamma'}$ and $r^{*}$ respectively as the maximum amount of time reserved on any resource due to the PGTs in $Z_{\Phi} + \gamma'$ and the resource on which the maximum reservation occurs.
    By Corollary \ref{thm: good accounting applies to partially ordered FCs}, \begin{equation}\label{eq:MaxReservedIsMinimumAllocRequired}
        M^{Z_{\Phi} + \gamma'} = \mathcal{R}(Z_{\Phi} + \gamma').
    \end{equation}
    Since we condition on all comparisons \eqref{eq: test required time} passing, \begin{equation}
        T_{\xi_r} - R(\tilde{Z}_{\phi}) \geq 0, \ \forall r \in \xi(\Pi_{\phi}).
    \end{equation}
    Hence it is guaranteed that
    \begin{equation}
        T_{\xi_{r^{*}}} = T^{SI} - M^{Z_{\Phi}+\gamma'} \geq 0.\label{eq:MaxReservedLeqSchedulingInterval} 
    \end{equation}
    Combining (\ref{eq:MaxReservedIsMinimumAllocRequired}) and (\ref{eq:MaxReservedLeqSchedulingInterval}), \begin{equation}
        \mathcal{R}(Z_{\Phi}+\gamma') \leq T^{SI}.
    \end{equation}
\end{proof}

Theorem \ref{thm: admin control is sound} (Admit Tasks is Sound) conditioned on inclusion of the path $\pi_{\gamma'} \in \Pi_{\phi}$ for some $\phi \in \Phi$. The following Lemma addresses PGTs $\gamma' \in \bm{\Gamma}$ for which ${\cancel{\exists} \ \phi}$ such that $ {\pi_{\gamma'} \in \ \Pi_{\phi}} $.
 This can happen in a network with a dynamically evolving network resource graph. For example, if a resource on the path $\pi_{\gamma}$ becomes unresponsive between the time at which the demand was registered and the time at which $\bm{\Gamma}$ is retreived by the Network Scheduler. 
\begin{lemma}\label{lemma: Admin sound when no FC}
    Suppose $\gamma' \in \bm{\Gamma}$ and ${\cancel{\exists} \ \phi}$ such that $ {\pi_{\gamma'} \in \ \Pi_{\phi}}$. 
    Then, $\gamma'$ is rejected by \textsc{AdmitTasks}.
\end{lemma}

\begin{proof}
    Let $d$ be the demand associated with $\gamma'$. Let $\bm{\Gamma}_d \subseteq \bm{\Gamma}$ denote set PGTs in $\bm{\Gamma}$ associated with the demand $d$. 
    The \textsc{AdmitTasks} function calls the \textsc{GetFillingClass} function for the PGT $\gamma'$. If ${\cancel{\exists} \ \phi}$ such that $ {\pi_{\gamma'} \in \ \Pi_{\phi}}$, then \textsc{GetFillingClass} raises a \texttt{FillingClassError}, the PGT $\gamma'$ is removed from $\bm{\Gamma}_{d}$, and the PGT $\gamma'$ is rejected by \textsc{AdmitTasks}.  
\end{proof}

\begin{remark}
    When a PGT $\gamma'$ which realizes a demand $d$ is accepted, the final step of \textsc{AdmitTasks} is to remove the set of alternative PGTs $\bm{\Gamma}_d \subseteq \bm{\Gamma}$ associated with the same demand. Hence $\bm{\Gamma}$ is updated to $\bm{\Gamma} \setminus \bm{\Gamma_d}$. This ensures that at most one PGT is admitted per demand.
\end{remark}

\subsection{Minimal Allocation is Guaranteed}

It is now possible to prove one of our main results, Theorem~\ref{thm: min alloc guaranteed} (Minimal Allocation is Guaranteed).
This theorem underpins the proof of Theorem~\ref{thm: QoSAcceptedTasks} (Deterministic Satisfaction of Service Agreements) in Section~\ref{sec: Performance Analysis}.

The following theorem assumes that the set of filling classes is well-behaved. 
In our implementation of Arqon (Section~\ref{sec: Implementation}) we specify a path partition and the set of filling classes it induces. Then, we prove that this specific set of filling classes is well-behaved.
This is sufficient to demonstrate the existence of well-behaved sets of filling classes. 

\MinAllocGuaranteed*

\begin{proof}
    By Theorem~\ref{thm: admin control is sound}, the time required to execute a minimum allocation schedule does not exceed the scheduling interval $T^{SI}$. 
    
    It remains to show that given a schedule $S_k$ as input, the bonus allocation phase of Compute Schedule cannot extend the required execution time of the schedule $S_k$ to beyond $T^{SI}$.
    The \textsc{BonusAllocation} function in Algorithm \ref{alg: Compute Schedule} for Compute Schedule takes the schedule $S_k$ as input, and ${\forall \phi \in \Phi}$, it returns an update of the schedule, $S_k = \textsc{RoundRobinBonus}(\phi, S_k)$. \textsc{RoundRobinBonus}, given by Algorithm~\ref{alg: bonus round robin resource schedule} may insert PGAs of a PGT $\gamma$ into the schedule $S_k$. Before any PGA may be added, the following condition is checked: \begin{equation}\label{eq: condition to add in RRB}
        \test + E_{\gamma} < T^{SI}.
    \end{equation}
    Condition \eqref{eq: condition to add in RRB} ensures that a PGA is only added to schedule $S_k$ if the time required to execute the updated schedule $S_k$ does not exceed $T^{SI}$. Hence, after all updates are made to $S_k$ in the bonus allocation phase, the time required to execute $S_k$ does not exceed $T^{SI}$.
\end{proof}

\section{Complexity Analysis of the Network Scheduler}
\label{app: complexity proofs}

We now develop the complexity analysis of the Network Scheduler, from which we derive and prove Theorem~\ref{thm: overall complexity} in Section~\ref{sec: Complexity}. 
We calculate the complexity in terms of the number of active PGTs ${N=|Z_{\Phi}|}$ in the set of filling classes $\Phi$, the number of active PGTs ${N_{\phi}=|Z_{\phi}|}$ in a single filling class $\phi$,
the number of PGTs ${k=|\bm{\Gamma}|}$ in the task intake object, and the number of internal resources ${R=|\mathfrak{R}|}$. 

To bound the complexity of a process, it must eventually halt on any input. 
Unlike the other processed of the Network Scheduler, which are based on loops with a finite number of iterations, the bonus allocation phase is based on \textsc{RoundRobinBonus}, which contains a \texttt{while} loop.
To bound the operational complexity of \textsc{RoundRobinBonus}, it is necessary to demonstrate that this algorithm halts on any input. 

In this Appendix we first prove that \textsc{RoundRobinBonus} halts on any input, and then we bound the complexity of each process of \textsc{NetworkScheduler}.

\newcommand{\trel}{T^{\text{rel}}}
\newcommand{\htrel}{\hat{T}^{\text{rel}}}
\newcommand{\tend}{t^\texttt{end}}
\newcommand{\tstart}{t^\texttt{start}}
\newcommand{\tminimal}{t^\texttt{minimal}}
\newcommand{\sinfty}{S_\infty}

\newcommand{\nst}{\textsc{nextStartTime}}

\subsection{Round Robin Bonus Halts}
To prove that \textsc{RoundRobinBonus}, Algorithm~\ref{alg: bonus round robin resource schedule}, halts on any input we need to show that in finite time the earliest start time $\test$ is updated by \textsc{RoundRobinBonus} to a value that exceeds the scheduling interval, $\test > T^{SI}$. 

Let us simplify the representation of the algorithm in the following manner:

\begin{algorithm}[h]
    \SetKwInOut{Input}{Input}
    \SetKwInOut{Output}{Output}
    \SetKw{Return}{return}
    \SetKw{Continue}{continue}
    \SetKw{Break}{break}
    \SetKwProg{Fn}{Function}{:}{end}

    \Fn{\textsc{RoundRobinBonus}}{

    \Input{Filling class $\phi$, network schedule $S^{\text{input}}$.}
    \Output{Updated network schedule $S$.}

        Set $T^{\text{rel}}\leftarrow \{\}$\;
        Set $\test \leftarrow 0$\;
        Set $k\leftarrow 0$\;
        \While{$t^{\texttt{\emph{est}}} < T^{SI}$}{
          $S, \htrel \leftarrow f(S, T^{\text{rel}}, \test)$\;
          Set $\test, T^{\text{rel}} \leftarrow\textsc{nextStartTime}(S, \htrel,\test)$\;
        }
    }
\caption{Simplified representation of \textsc{RoundRobinBonus}.}
\label{alg: bonus round robin resource schedule rewrite}
\end{algorithm}
In Algorithm~\ref{alg: bonus round robin resource schedule rewrite} the function $f$ represents the action of the for loop in Algorithm~\ref{alg: bonus round robin resource schedule}. 

To prove that the algorithm halts, we need the sequences $S_i, \htrel_i, \trel_i$ and $\test_i$,  where \begin{align}
        S_{i+1}, \htrel_{i+1} &= f(S_i, \trel_i, \test_i),\label{eq: sequence f}\\
        \test_{i+1}, \trel_{i+1} &= \textsc{nextStartTime}(S_{i+1}, \htrel_{i+1}, \test_i), \label{eq: sequence nst}
    \end{align}
    and 
    \begin{align}
        \test_0 &= 0, & \trel_0 &= \emptyset, & \htrel_0 &= \emptyset, & S_0 &= S^{\text{input}}.
    \end{align}

Note that \textsc{NextStartTime} removes a release time from the set of release times to consider if it is the next start time, so in general $\htrel_{k} \neq \trel_{k}$.

Similarly, if a PGA is scheduled, then $f$ will add a release time to the set of release times to consider, so $\trel_{i}\neq\htrel_{i+1}$ in general. 


We first show that the next earliest start time $\test_k$ is smaller than any remaining release times to consider. 
\begin{lemma}
    $t^{\texttt{\emph{est}}}_k < \min T^{\text{\emph{rel}}}_k.$
    \label{lemma: test-k < min trel-k}
\end{lemma}

\begin{proof}
Recall that $\delta = (\gamma, t; r) \in \tilde{S}$ is a PGA for PGT $\gamma$ scheduled at start time $t$ on resource $r$ in a schedule $\tilde{S}$. The end time of such a PGA $\delta$ is ${t^{\text{end}}_{\delta} = t + E_{\gamma}}$.
We examine the affect of \nst{}.
There are two cases to consider. 
If $\min\{t_\delta^{end} : \delta\in S'_k\} < \min \htrel_k$, then \begin{equation*}
    \test_k = \min\{t_\delta^{end} : \delta\in S'\} < \min \htrel_k = \min \trel_k.
\end{equation*} 
    Otherwise, if $\min \htrel_k \leq \min\{t_\delta^{end} : \delta\in S'_k\}$, then \begin{equation*}
        \test = \min\htrel_k, \text{ and } \trel_k  = \htrel_k\setminus \{\test_k\} \implies \test_k < \min \trel_k.
    \end{equation*} 
\end{proof}

Next we show that new release times are only added when $\test$ is the end time of a PGA in the schedule.
\begin{lemma}\label{lemma: only add to T-hat at end times}
    Suppose that for some $t^{\text{\emph{est}}}_k$, $\nexists \delta\in S_k$ such that $t^{\text{\emph{est}}}_k = t^{\text{\emph{end}}}_\delta$. 
    Then, following the application of $f$ as in~\eqref{eq: sequence f},   $\hat{T}^{\text{\emph{rel}}}_{k+1} = T^{\text{\emph{rel}}}_k$, i.e. no release times are added by $f$.  
\end{lemma}
\begin{proof}
    Suppose not.
    Then there exists some time $t^*$ such that ${t^*\in\htrel_{k+1}}$ but ${t^*\not\in\trel_k}$. 
    This implies there exists a PGA $\delta^*$ from a PGT $\gamma^*$ such that ${t^* = \tend_{\delta^*} + \tminsep_{\gamma^*}}$ and ${\test_{k-1} < \tend_{\delta^*} \leq \test_{k}}$. 
    
    As ${\test_{k-1} < \tend_{\delta^{*}}}$, by the action of \textsc{nextStartTime} we must have that $\test_{k} \leq \tend_{\delta^*}$.
     This implies that ${\test_k = \tend_{\delta^*}}$. 
    But this contradicts the assumption that $\test_k$ is not an end time of any scheduled PGA and so no such $\delta^*, t^{*}$ exist.
\end{proof}

We now show that $\test_i$, which is updated in each loop of \textsc{RoundRobinBonus} by \textsc{nextStartTime} is strictly increasing.
\begin{lemma}
    For all $i>0$, ${t^{\text{\emph{est}}}_{i+1} - t^{\text{\emph{est}}}_{i} > 0}$. 
    \label{lemma: nextStartTime is non-decreasing.}
\end{lemma}
\begin{proof}
    We proceed by induction on the iteration index $i$. 

    For the base case $i=0$. 
    Consider the action of $f$ on $S_0, \trel_0, \test_0$.
    The only way for a new time to be added to $\trel$ is from a past minsep violation. 
    As there are no PGAs scheduled before time $t=0$, there cannot be any past minsep violations at $\test_0$. 
    Therefore, there are no new release times added to $\trel_0$ by $f$ and so ${\htrel_1 = \trel_0 = \emptyset}$.

    Now consider the action of \nst{} on $S_1, \htrel_1, \test_0$. 
    As $\htrel = \emptyset$ and \nst{} only removes elements from $\htrel$, we must have ${\trel_1 = \emptyset}$. 
    Moreover, there are no release times to consider, and so the next earliest start time must come from the earliest end time of a PGA in $S_1$, i.e. \begin{equation}
        \test_1 = \min_{\delta\in S_1} \tend_\delta \geq \min_{\gamma\in Z_{\Phi}} E_\gamma > 0,
    \end{equation} and the base case follows. 

    Suppose now that the result holds for all $i\leq k$. 
    We need to consider two cases, whether $\test_{k+1}$ is an end time of a PGA in $S_{k+1}$ or whether it is a release time in $\htrel_{k+1}$. 
    In the first case, we have that \nst{} explicitly only considers PGAs with end times greater than $\test_k$ and so we must have $\test_{k+1}>\test_k$. 
    In the latter case, suppose that $\exists t^*\in\htrel_{k+1}$ such that $t^*<\test_k$, i.e. \nst{} would choose $\test_{k+1} = t^*$. 
    
    There are two ways for $t^*$ to get into $\htrel_{k+1}$, either carried forward from~\eqref{eq: sequence nst} in iteration $k$ or from~\eqref{eq: sequence f} in iteration $k+1$. 
    By Lemma~\ref{lemma: test-k < min trel-k} the time returned by \nst{} is less that the minimum value in $\trel_k$, so it must be that $t^*$ was added by $f$ as in~\eqref{eq: sequence f}, in iteration $k+1$. 
    In particular, $\exists \delta^*\in S_k$ such that $\delta^*$ is a left-minsep violation for some PGT $\gamma^*$ at time $\test_k$, i.e. $t^* = \tend_{\delta^*} + \tminsep_{\gamma^*} > \test$. 
    This contradicts that $t^* < \test$, and so no such $t^*$ can exist. 

    Therefore, $\test_{k+1} > \test_k$ and the result follows. 
\end{proof}

We now show that the sequence $\test_k$ visits the end time of every scheduled PGA in a finite number of steps. 
To do so, we introduce the notation \begin{equation}\label{eq: def S infinity}
    S_\infty := \lim_{k\to\infty}S_k.
\end{equation}
As $S_k$ is strictly increasing, in the sense that PGAs can only be added to a schedule, this limit is well-defined.  
Note that in the limit $k\to\infty$ we don't impose a cutoff at the end of a scheduling interval, and instead assume the functions $f$ and $\textsc{NextStartTime}$ are applied as long as the algorithm does not halt. 

\begin{corollary}[Finite number of steps between end times]
    Let $t{^{\text{\emph{est}}}_k < \max_{\delta\in S_{\infty}}t^{\text{\emph{end}}}_\delta}$ be the end time of a PGA in $S_\infty$. 
    Then $\exists l < \infty$ such that $t^{\text{\emph{est}}}_{k+l}$ is also the end time of a PGA in $\sinfty$. 
\end{corollary}
\begin{proof}
    Consider $\test_i$ for $i>k$.
    If $\test_i \neq \tend_{\delta'}$, for some $\delta' \in \sinfty$, then $\test_i$ must be a release time. By Lemma~\ref{lemma: only add to T-hat at end times}, the number of release times in the set of release times cannot increase until ${\test_l =\tend_{\delta'}}$ for some $\delta' \in \sinfty$. Hence for ${k+1 \leq i \leq l}$, ${|\htrel_{i}| \leq |\htrel_{k+1}|}$.
    Furthermore, $|\htrel_{k+1}| < |Z_{\Phi}| < \infty$ as each PGT can only contribute one release time. 
    Therefore, after at most $l=|Z_{\Phi}|$ steps, $\test_{k+l}=\tend_{\delta'}$ for some $\delta' \in \sinfty$. 
\end{proof}
\begin{corollary}[Visit every end time]
    Let $\delta$ be a PGA in $\sinfty$ with end time $t^{\text{\emph{end}}}_\delta$. Then there exists some $k<\infty$ such that $t^{\text{\emph{est}}}_k = t^{\text{\emph{end}}}_\delta$. 
    \label{coro: vist every end time}
\end{corollary}
\begin{proof}
    As PGAs can only be scheduled at the current value of $\test$, any new end times are in the future and will therefore not be skipped over. 
    The number of PGAs in $\sinfty$ is countable, therefore let us enumerate them by $\delta_0, \delta_1, ...$ such that $\tend_{\delta_i} \leq \tend_{\delta_{i+1}}$. Then by the previous corollary we have that there exists some $l < k|Z_{\Phi}|$ such that $\test_l = \tend_{\delta_k}$.  
\end{proof}

Finally, we show that for every time $c> 0$, $\test_k$ exceeds $c$ in a finite number of steps.

\begin{lemma}
     For each filling class ${\phi\in \Phi}$, \textsc{BonusRoundRobin}$(\phi, S^{\text{\emph{input}}})$ generates the sequences~\eqref{eq: sequence f} and~\eqref{eq: sequence nst} such that $ {\forall c\geq 0}$ there exists $k<\infty$ for which ${t^{\text{\emph{est}}}_k \geq c}$.
    \label{lemma: test hits in finite time}
\end{lemma}

\begin{proof}
    As every PGT in $Z_{\phi}$ has a finite expiry time, after which no further PGAs are scheduled, there are only a finite number of PGAs in $\sinfty$.
    We consider two cases, ${c<\max_{\delta\in\sinfty}\tend_\delta}$ and ${c\geq \max_{\delta\in\sinfty}\tend_\delta}$. 
    In the former case, there exists some end time $\tend_{\delta^*} > c$. 
    By Corollary~\ref{coro: vist every end time}, the sequence $(\test_i)$ reaches $\tend_{\delta^*}$ in finite time and so exceeds $c$ in finite time. 

    In the latter case, consider the values that $(\test_i)$ takes after reaching the last end time in $S_\infty$.  
    Let $k$ be such that $\test_k = \max_{\delta\in\sinfty}\tend_\delta$, and let $\htrel_{k+1}$ be the set of release times to consider after $\test_k$. 
    As each PGT can only contribute at most one release time to $\htrel_{k+1}$, we have $r = |\htrel_{k+1}| \leq |Z_{\phi}| < \infty$. 
    Furthermore, we have that $|\htrel_{k+i+1}| = |\htrel_{k+i}| - 1$ for $i = 1, ..., r$. 
    Therefore, we have that $\htrel_{k+r} = \emptyset$.
    Moreover, at $\test_{k+r}$ there are no future end times (as nothing can be added to schedule which ends after $\max_{\delta\in\sinfty}\tend_\delta < \test_{k+r}$) and no release times as $\htrel_{k+r} = \emptyset$. 
    Therefore, the action of \textsc{nextStartTime} is to set $\test_{k+r+1} = \infty > c$.
    Therefore $(\test_i)$ hits $\infty$ in a finite number of time steps and hence exceeds $c$ in a finite number of time steps. 
    
\end{proof}

We can now prove that \textsc{RoundRobinBonus}($\phi, \ S$) halts on all inputs. 
\begin{theorem}
    For all possible inputs~${\phi, \ S}$, \textsc{RoundRobinBonus}(${\phi, S}$) halts.
\end{theorem}
\begin{proof}
    Suppose not. 
    Then the algorithm must enter an infinite loop somewhere. 
    The \verb|For| loop has finite number of steps, $|Z_{\phi}|$, so will always exit.
    We define a `cycle' of the round-robin scheduler to be each time the \verb|While| condition is checked. 
    Note that each cycle the value of $\test$ is incremented by \textsc{nextStartTime} as per the sequences~\eqref{eq: sequence f},~\eqref{eq: sequence nst} and by Lemma~\ref{lemma: nextStartTime is non-decreasing.} $\test$ increases with each cycle. 
    The program halts when $\test >T^{SI}$.
    By Lemma~\ref{lemma: test hits in finite time} this occurs in a finite number of cycles and so the program will always halt. 
\end{proof}

\subsection{Complexity Analysis}

We determine the complexity of \textsc{NetworkScheduler} as in Algorithm~\ref{alg: Network Scheduler main} by examining each of the subroutines: \textsc{UpdateFillingClasses}, \textsc{AdmitTasks}, and \textsc{ComputeSchedule}. 

\subsubsection{Update Filling Classes}
First, we examine the complexity of \textsc{UpdateFillingClasses}, Algorithm~\ref{alg: Update Filling Classes}, as implemented in Listing~\ref{lst: update filling classes}. 

\begin{lemma}\label{lemma: UFC complexity}
  \textsc{UpdateFillingClasses} as implemented in listing~\ref{lst: update filling classes} has operational complexity $O((N+R)^2)$.
\end{lemma}

\begin{proof}
    The algorithm is split into two phases: building the set of filling classes from $\Pi$ and then removing terminated and expired demands from each filling class. 
    
    If the path partition is unchanged, the old set of filling classes can be brought forward and there are no required operations. 
    Otherwise, we build a set of filling classes from the new path partition. 
    We create one filling class per cell of the path partition, then sort them and finally create a dictionary holding the ordering. 
    Creating a filling class is an $O(1)$ operation, as it only requires value assignment. 
    
    To calculate how many filling classes there can be, consider the specific choice of filling classes given in Definition~\ref{def: specific path partition}. 
    Let $|B|$ be the number of long-distance backbones, $|J|$ the number of junction nodes and $|I|$ be the number of EGIs in the network, so $|B|+|J|+|I|=R$. 
    
    There is a single backbone partition $\Pi^B$. 
    Each set of junction nodes $J_i$ is disjoint, and so there are at most $|J|$ junction filling classes.
    Each EGI contributes a filling class, so there are $I$ EGI filling classes.
    Therefore there are at most $1+|J|+|I| < R$ filling classes, and so there are $O(R)$ filling classes. 

    Creating all of the filling classes is thus $O(R)$. 
    Sorting the set of filling classes is $O(R\log(R))$. 
    Computing the dictionary holding the ordering requires a double iteration of the list of filling classes, which takes $O(R^2)$ time. 
    The implementation of these processes is given in Listing~\ref{lst: set of filling classes}.

    The remaining assignments in lines 14-16 of Listing~\ref{lst: set of filling classes} take $O(1)$ time, because a None value is always assigned in \textsc{NetworkScheduler}. These properties are never initialized in this main loop and are only initialized in the numeric complexity evaluations of Section~\ref{sec: Evaluation}, resulting in Figures~\ref{fig: admission control complexity} and~\ref{fig: scheduling complexity}. 
    The overall complexity of constructing the set of filling classes is thus $O(R^2)$. 

    If the path partition changes, we attempt to re-add each PGT $\gamma$ to a filling class. 
    This requires an $O(R)$ lookup to find the correct filling class and then $O(1)$ time to add $\gamma$ to the filling class. 
    Note that re-calculating the minimum required time for the filling class is delayed until the end of this whole routine. 
    Therefore, re-adding PGTs takes $O(NR)$ time. 
    Subsequently, updating the set of removed tasks requires creating and assigning a list of length $O(N)$. 

    Therefore, building the set of filling classes takes $O(1)$ time if there is no change to the path partition and ${O(R^2+NR+N)}$ time otherwise. 

    To remove the terminated and expired demands, each filling class $\phi \in \Phi$ is treated in turn. 
    For a given filling class $\phi$ with $N_\phi$ active PGTs, each PGT is iterated over and if it is in the set of terminated demands it is removed, or if it has expired it is removed. 
    Checking if a PGT is in the terminated set is an $O(1)$ operation (set membership).
    Checking if a PGT has expired is also an $O(1)$ operation. 
    Therefore, removing the terminated or expired PGTs takes $O(N_\phi)$ time. 
    
    After the PGTs are removed, the required time is recalculated (\textsc{CalculateRequiredTime}($\phi$), implemented by lines 52-63 in Listing~\ref{lst: single filling class}), at a cost of $O(N_\phi^2)$ -- see Lemma~\ref{lemma: compute required time complexity}. 
    Therefore, updating each filling class takes $O(N_\phi^2)$ time, for a total of $O(N^2)$ time across all filling classes. 

    Updating the lists of removed and expired demands are $O(N_\phi)$ or $O(N)$ ((once we sum over all filling classes) operations, so are sub-leading. 

    Therefore, the overall complexity is $$O(R^2+NR+N) + O(N^2) = O((N+R)^2).$$ 
    
\end{proof}

\subsubsection{Admit Tasks}
We now show the complexity of Admit Tasks, Algorithm~\ref{alg: Admit Tasks}, as implemented by Listing~\ref{lst: admit new tasks}. It makes use of some of the methods of a filling class, implemented in Listing~\ref{lst: single filling class}. First we determine the operational complexity of \textsc{CalculateRequiredTime}, which Admit Tasks uses to reserve time for the minimum allocation of PGAs of active PGTs on all resources associated with a filling class.
\begin{lemma}
    \textsc{CalculateRequiredTime}$(\phi)$ for a filling class $\phi$ as implemented in Listing~\ref{lst: single filling class} has operational complexity $O(N_\phi^2)$.
    \label{lemma: compute required time complexity}
\end{lemma}
\begin{proof}
    Calculating the total required time requires calculating the vector of cycle times $\mathbf{c}^\phi = (c_i^\phi)$ and the vector of cycle numbers $\mathbf{n}^\phi = (n_i^\phi)$ as in the proof of Proposition~\ref{prop: timeReqSingleFC} in Appendix~\ref{app: performance analysis proofs}.
    Calculating each of the $n_i^\phi$ each requires a single value assignment and subtraction, an $O(1)$ operation, for a total of $O(N_\phi)$ operations. 
    
    Calculating $c_k^\phi$ has the following steps:
    We obtain the list of end times plus minseps for each PGT with index greater than or equal to $k$, which has complexity $O(N-k)$. 
    Obtaining the sum of the execution times is also $O(N-k)$.
    Finding the maximum of this list and the sum of the executions times is an $O(log(N-k))$ operation. 
    Combining these three steps, the total complexity for computing each of the $c_k^\phi$ is $O(N-k)$. 
    The total complexity for computing the $c_i^\phi$ is then $\sum_{k=0}^{N_\phi - 1}O(N_\phi - k) = O(N_\phi^2)$.
    
    Once we have the $n_i^\phi$ and the $c_i^\phi$ terms, we need to calculate $\mathbf{n}^\phi.\mathbf{c}^\phi$, which requires $O(N_\phi)$ multiplications, $O(N_\phi)$ additions and one value assignment, for an overall complexity of $O(N_\phi)$. 

    Combining these calculations, the total complexity is \begin{equation}
        O(N_\phi) + O(N_\phi)^2 + O(N_\phi) = O(N_\phi^2).
    \end{equation}
\end{proof}

\begin{theorem}\label{thm: admit new tasks complexity}
  \textsc{AdmitTasks} as implemented in Listing~\ref{lst: admit new tasks}, operating on a set of filling classes $\Phi$ as implemented in Listing~\ref{lst: set of filling classes}, has operational complexity \begin{equation}\label{eq: admit tasks complexity}
      O(k[(N+k)^2+R]2^R+R^2).
  \end{equation}  
\end{theorem}
\begin{proof}

    Initially \textsc{AdmitTasks} calculates and sets the current available time on each resource $r \in \mathfrak{R}$. 
    Setting the initial available time of a full scheduling interval requires $O(R)$ value assignments (each $O(1)$).
    In our implementation of the filling classes the total required time is stored as a property of a filling class, so we can get the total required time in $O(1)$ (amortized) time. 
    Each filling class is associated with $O(R)$ resources, so decrementing the time available ($O(1)$) for all resources associated with a particular filling class takes $O(R)$ time.
    There are $O(R)$ filling classes in total, so the initialization of the routine takes $O(R^2) + O(R) = O(R^2)$ time in total. 
    
   Next \textsc{AdmitTasks} determines whether to admit a particular PGT $\gamma\in\bm{\Gamma}$.
    To find the filling class of $\gamma$ we check each filling class $\phi$ sequentially to determine if the path is in the cell $\Pi_{\phi}$ of the path partition. 
    Each filling class can be obtained in $O(1)$ time (dictionary query), performing the inclusion check is $O(1)$ (set inclusion), and there are $O(R)$ filling classes, for a total of $O(R)$ operations. 

    Suppose we find a filling class $\phi$ for $\gamma$, and let $N_\phi = |Z_{\phi}|$. 
    First, we need to remove the contribution from the PGTs in $Z_{\phi}$ to the reserved time on each resource associated with $\phi$. 
    Updating each of the $O(R)$ resources takes $O(1)$ time (dictionary parameter query, subtraction, value assignment) for a total complexity of $O(R)$. 
    
    We can then add $\gamma$ to $Z_{\phi}$.
    To do so, first we need to determine at which index in the ordering it should be inserted to preserve the ordering by $\minimumallocation$. 
    In our implementation this takes $O(N_\phi)$ time as we examine each element in sequence.  
    Once the index $i$ is found, $\gamma$ can be inserted into the list, which takes $O(N_\phi - i)$ time. 
    Finally, adding $\gamma$ to $Z_{\phi}$ means the minimum required time needs to be recomputed, which by Lemma~\ref{lemma: compute required time complexity} takes $O(N_\phi^2)$ time.
    Therefore, the total operational complexity of adding $\gamma$ to $Z_{\phi}$ is $O(N_\phi) + O(N_\phi - i) + O(N_\phi^2) = O(N_\phi^2)$.

    Once the perspective PGT has been added to its filling class, we need to check if $R(Z_\phi) \leq T^{SI}$, in which case the PGT can be accepted. 
    Updating the reserved time on each resource associated with $\phi$ takes $O(1)$ time (dictionary query, addition, value assignment)
    Checking (float comparison, $O(1)$) that the available time for all $O(R)$ resources is positive takes a total of $O(R)$ time. 

    If this condition is met (i.e. the PGT is accepted), then updating the dictionary of accepted demands takes $O(1)$ time.
    Otherwise, if the PGT is rejected, then we need to undergo the same procedure as above but in reverse, which has at worst the same complexity. 

    Therefore, determining an accept/reject decision for a PGT $\gamma$ has total complexity \begin{equation*}
        O(R) + O(N_\phi^2)+O(R) + O(1) = O(N_\phi^2)+O(R).
    \end{equation*}
    Each demand in $\bm{\Gamma}$ contributes at most one PGT into $Z_{\phi}$, so we can bound $|N_\phi|$ above by $N+k$, and the total complexity of making a determination for $\gamma$ is $O((N+k)^2+R)$
    
    For each possible path $\pi$ along which PGAs for a demand $d$ can be realized, there is a constant finite number of entanglement generation protocols which can be employed.
    Therefore, the total number of possible PGTs for a given demand is $O(2^R)$.
    There are $k$ demands for which a decision needs to be made.
    Therefore the overall computational complexity is given by ${O(k[(N+k)^2+R]2^R+R^2)}$.     
\end{proof}

\subsubsection{Compute Schedule}
We determine the complexity of Compute Schedule by first determining the complexity of it's subprocesses. Fundamentally, computing schedules requires adding PGAs to the network schedule, which is implemented as in Listing~\ref{lst: resource schedule}. Compute Schedule has three phases: minimal allocation, bonus allocation, and schedule compilation. The minimal allocation phase computes direct allocation schedules (\textsc{DirectAllocation}, Algorithm~\ref{alg: direct allocation subroutine}, Listing~\ref{lst: directallocation}) for each filling class. The bonus allocation phase attempts to add PGAs to the direct allocation schedules for each filling class, according to \textsc{BonusRoundRobin}, Algorithm~\ref{alg: bonus round robin resource schedule}, which is implemented in Listing~\ref{lst: bonusRR}. Schedule compilation (\textsc{CompileSchedule}) is implemented in Listing~\ref{lst: resource schedule}.

\begin{lemma}
    Let the format of the network schedule be as implemented in Listing \ref{lst: resource schedule}.
    Then, appending a PGA to the schedule has complexity $O(R)$. Inserting a PGA anywhere else in the schedule has complexity $O(RN)$. 
    \label{lemma: add PGA to schedule}
\end{lemma}
\begin{proof}
    The process of scheduling a PGA for a PGT $\gamma$ can be decomposed into three parts.
    Firstly, we get the resource schedules for the individual resources on the path $\pi_{\gamma}$. 
    The network schedule is implemented as a dictionary of resource schedules keyed by their resource identifier, so we simply loop through the path $\pi_\gamma$ and retrieve each schedule in turn (dictionary query, $O(1)$). 
    Therefore, the complexity of getting all the schedules to update is $O(|\pi_\gamma|) = O(R)$. 
    
    Secondly, for each resource in $\pi_{\gamma}$, we need to configure the PGA(s) to be scheduled, as in lines 53-56 of Listing~\ref{lst: resource schedule}. This takes $O(1)$ time. 
    
    Finally, we need to add the constructed PGA(s) to the resource schedule $S_r$.
    If a PGA can be appended to the end of the resource schedule whilst maintaining time ordering of all scheduled PGAs, adding the PGA has the same (amortized) complexity as appending to a list, which is $O(1)$. 
    If it is not known at which index a new PGA needs to be inserted to maintain time ordering, then it can be obtained using binary search in time $O(\log |S_{r}|)$.
    As $S_r$ is a list of PGAs, the total number of which is $O(N)$, we have that finding the index required is $O(\log N)$. 
    Let $q$ be the index at which the new PGA should be inserted. 
    Then performing the insertion has complexity $O(|S_r|-q) = O(N)$. 

    The total complexity of parts two and three combined is then $O(1)$ when appending and $O(N)$ otherwise. 
    Since these parts need to be repeated for each resource on the path $\pi_{\gamma}$, in total the complexity is $O(R)$ when appending and $O(RN)$ otherwise. 
\end{proof}

\begin{lemma}\label{lemma: single FC direct allocation}
  \textsc{DirectAllocation}($\phi$), as implemented in Listing~\ref{lst: directallocation} for a filling class $\phi$, operating on a network schedule as implemented in Listing~\ref{lst: resource schedule}, has computational complexity \begin{equation}\label{eq: direct allocation complexity}
      O(N_\phi  R) + O(N_\phi^2).
  \end{equation}
\end{lemma}

\begin{proof}
  The total number of PGAs to be scheduled from filling class $\phi$ is $\sum_{\gamma\in\phi}\minimumallocation_\gamma = O(N_\phi )$.
This holds since $\forall \gamma, \ \minimumallocation_{\gamma}$ is a positive integer.
    We proceed through Listing~\ref{lst: directallocation} line-by-line. 
    In initializing the routine, lines 11-15, we obtain the logger, set \verb|cycle_start_time| and set up one pointer for each resource. These are either value assignment or dictionary query operations, which are all $O(1)$. 
    
    In lines 17-45 we construct the schedule. The PGTs in a filling class are ordered by a set of indices $X$ such that ${\minimumallocation_{x} \leq \minimumallocation_{x+1}}$.  
    For each $\gamma_x$ in increasing order ${x=\{0, 1, \cdots, |Z_{\phi}|\}}$, we add all $\minimumallocation_{\gamma_x}$ PGAs for $\gamma_x$ to the schedule. 
    
    In lines 19-21 we retrieve the new cycle time (dictionary query, $O(1)$) and reset the cycle counter $k$ (value assignment, $O(1)$). This portion has total complexity $O(1)$.

    The two \texttt{for} loops, in lines 23-31 add all the PGAs for the PGT $\gamma_x$ to the schedule, for each cycle in turn.  
    For each of these PGAs, the offset of the start time from the start of the cycle is calculated in line 25.
    This requires summing the elements of a list of length order $O(N_\phi)$, so it is an $O(N_\phi)$ operation. 
    Then the PGA needs to be added to the schedule. 
    By Lemma~\ref{lemma: add PGA to schedule}, this has complexity $O(R)$ as, by construction, \textsc{DirectAllocation} only appends PGAs to the end of resource schedules. 
    Therefore, the total complexity incurred from adding a PGA to the schedule is ${O(R) + O(N_\phi)}$.
    Lines 33-42 simply schedule one last PGA for each PGT following the end of it's last cycle, as discussed in the proof of Proposition~\ref{prop: timeReqSingleFC}. 
    
    The total complexity is thus given by \begin{align}
        O(1) + O(N_\phi) + O(N_\phi)&(O(R)+O(N_\phi))\notag\\ &= O(N_\phi^2) + O(N_\phi R).
    \end{align}
\end{proof}

\begin{lemma}
  \textsc{RoundRobinBonus}($\phi$) as implemented in Listing~\ref{lst: bonusRR} for a filling class $\phi$ operating on a network schedule as implemented in Listing~\ref{lst: resource schedule} has operational complexity \begin{equation}\label{eq: rr bonus complexity}
      O(N^2N_\phi R)).
  \end{equation}
  \label{lemma: single FC Bonus RR}
\end{lemma}
\begin{proof}
    The overall complexity depends on the complexity of the subroutines which are detailed in Algorithm~\ref{alg: auxiliary algoroithms used in bonus round robin} and implemented in Listing~\ref{lst: bonusRR}.
    
    \verb|_get_next_start_time|: 
    With the implementation of the network schedule in Listing~\ref{lst: resource schedule}, getting the end times to consider relies on getting each resource schedule $\forall r \in \mathfrak{R}$ (dictionary query, $O(1)$) and then retrieving the end time of a specified PGA (list element retrieval at known index, $O(1)$). 
    Therefore, getting the PGA end times to consider has a total complexity of $O(R)$. 
    The minimal end time is determined by minimizing over the set of end times to consider, an $O(R)$ operation.
    Therefore, the total complexity to get the next end time is ${O(R) + O(R) = O(R)}$.
    The next release time is determined by minimizing over the list of release times supplied as an input to the subroutine. 
    Each PGT in $\phi$ can contribute at most one release time to this list for a total of $O(N_\phi)$ elements. 
    Therefore, finding the next release time is $O(N_\phi)$. 
    The next start time is the minimum of the next end time and the next release time. Calculating it takes $O(1)$ time (comparing two values). 
    
    If the next start time is the next end time, the set of schedule pointers is updated. 
    To update the pointers, for each resource in $\mathfrak{R}$, we check if the pointed PGA ends before or at the next start time (dictionary query, $O(1)$, and list item retrieval at known index, $O(1)$). 
    If so, the pointer is incremented by two ($O(1)$).
    Therefore, the total complexity for updating the pointers is $O(R)$.     
    Alternatively, if the next start time is the next release time, we remove the next start time from the set of release times, which is an $O(1)$ operation. 

    Therefore, the contributions from getting the next end time, the next release time, finding the minimum of these and then updating the release times/schedule pointers gives a total complexity of \begin{align}
        O(R) + O(N_\phi) + &O(1) + O(R) \notag\\ &= O(R) + O(N_\phi).
    \end{align}

    \verb|_resources_available|:
    To check if all the resources required to schedule a PGA from PGT $\gamma$ are available, we perform the following steps: For each resource in $\pi_\gamma$, we check if the pointed PGT is active at time $\test$. 
    This consists of a dictionary query to get the resource schedule $O(1)$, list retrieval at a known index ($O(1)$) and a constant number of comparisons ($O(1)$). 
    Therefore, checking if one resource is available at time $\test$ is an $O(1)$ operation. The total from checking all resources is then $O(|\pi_\gamma|) = O(R)$. 

    \verb|_future_minsep_violations|:
    To check if there are minium separation violations in the future for a PGA from PGT $\gamma$, it suffices to check only the resource schedule from the first internal resource in $\pi_\gamma$. 
    Getting the relevant resource and its resource schedule is an $O(1)$ operation (list retrieval from known index, dictionary query).
    Initializing the offset variable is also $O(1)$. 
    Each check which is performed is an $O(1)$ operation (comparison, property retrieval, value assignment).
    The number of scheduled PGAs to check is only bounded by the size of the schedule, and so there are $O(N)$ checks to perform. 
    Combining the contributions from obtaining the resource schedule, initializing the counter and then performing the checks gives a total complexity of $O(N)$. 

    \verb|_past_minsep_violations|:
    The process is the same as for \verb|_future_minsep_violations|, so the complexity is $O(N)$. 

    We now move to the main loop of the algorithm. 
    The maximum number of cycles is the same as the maximum number of possible start times. 
    Each possible start time corresponds to the end of some PGA in the schedule, so there are at most \begin{equation}\label{eq: max possible start times}
        \sum_{\phi\in\Phi}\sum_{\gamma\in Z_{\phi}}\frac{T^{SI}}{E_\gamma + \tminsep_\gamma} = O(N)
    \end{equation} possible start times. The expression \eqref{eq: max possible start times} dictates that every PGT can contribute at most a number of PGAs based on how many times the PGA duration plus minsep fits within the scheduling interval.
    Thus the number of cycles is $O(N)$. 
    
    In each cycle there are $N_\phi $ rounds, during each of which \verb|_resources_available|, \verb|_future_minsep_violations|, \verb|_past_minsep_violations| are called. The algorithm may also add a PGA to the schedule, which by Lemma~\ref{lemma: add PGA to schedule} takes $O(RN)$ time as these PGAs typically need to be inserted intro the schedule rather than appended to the end. 
    At the end of each cycle, \verb|_get_next_start_time| is called. 

    Overall, the complexity is \vspace{-2mm}\begin{align}
        O(N)\bigg{(}O(N_\phi ) \big{(}O(R) + O(N) &+ O(N) + O(NR) \big{)} \notag\\+ O(R) + O(N_\phi)\bigg{)}&= O(N^2N_\phi R).
    \end{align}\normalsize
\end{proof}

\begin{lemma}
  Given that the network schedule is implemented as in Listing~\ref{lst: resource schedule}, \textsc{CompileSchedule} has operational complexity $O(NR)$.
  \label{lemma: compile schedule complexity}
\end{lemma}
\begin{proof}
    The schedule is implemented as a dictionary of lists of PGAs, one per internal resource in the network, called a resource schedule. 
    Therefore, the only operation required is to convert each PGA into its transmission format. 
    This takes constant time for each PGA. 
    There are $O(N)$ PGAs in each resource schedule and $O(R)$ resource schedules, and so the overall complexity for compiling the schedule is $O(NR)$. 
\end{proof}

\begin{theorem}
  The program \textsc{ComputeSchedule} has operational complexity \begin{equation}\label{eq: compute schedule complexity}
      O(N^3R).
  \end{equation}
  \label{thm: compute NS complexity}
\end{theorem}
\begin{proof}
  By Lemmas~\ref{lemma: single FC direct allocation}, \ref{lemma: single FC Bonus RR} and \ref{lemma: compile schedule complexity} the total complexity is given by \begin{align}
    &\sum_{\phi\in\Phi}O(N_\phi R) + \sum_{\phi\in\Phi}O(N^2N_\phi R) + O(NR)\notag\\
    &= O(NR) + O(N^3R) = O(N^3R).
  \end{align}
\end{proof}

\subsection{Numeric Evaluation}
In Section~\ref{sec: Evaluation} we discuss the simulation of randomly generated network topologies, parametrized by the number of long-distance backbones, local areas, and end nodes in the network. Table~\ref{tab: random topos computation times} compiles the average schedule computation times for each of these parametrizations. The time to compute the schedule is broken down by the processes Update Filling Classes, Admit Tasks, and Compute Schedule (without the Compile Schedule process). The total time to produce the network schedule includes the time for each of these processed and Compile Schedule. 
\renewcommand{\arraystretch}{1.2}
\begin{table*}[t]
    \centering
    \begin{tabular}{|c|c|c|c|c|c|c|}
    \hline
    \thead{\textbf{Long-distance}\\ \textbf{Backbones}} & \thead{\textbf{Local} \\ \textbf{Areas}} & \thead{\textbf{End} \\ \textbf{Nodes}} & \thead{\textbf{Time to Update}\\ \textbf{Filling Classes (s)}} & \thead{\textbf{Time to Admit}\\ \textbf{PGTs (s)}} & \thead{\textbf{Time to Compute}\\ \textbf{Schedule (s)} \\ \textbf{(without Compilation)}} & \thead{\textbf{Total Time to Produce}\\\textbf{Network Schedule (s)}}  \\ \hline
    1 & 2 & 15 & 0.00011 & 0.00017 & 0.00683 & 0.05268 \\
    2 & 2 & 15 & 0.00011 & 0.00019 & 0.00707 & 0.05246 \\
    2 & 2 & 50 & 0.00025 & 0.00180 & 0.03461 & 0.09126 \\
    2 & 3 & 30 & 0.00017 & 0.00051 & 0.02276 & 0.07343 \\
    2 & 3 & 50 & 0.00027 & 0.00147 & 0.04741 & 0.10465 \\
    3 & 3 & 30 & 0.00018 & 0.00061 & 0.02406 & 0.07529 \\
    5 & 4 & 40 & 0.00031 & 0.00133 & 0.06215 & 0.11797 \\
    6 & 3 & 35 & 0.00022 & 0.00121 & 0.03738 & 0.09196 \\
    7 & 5 & 50 & 0.00059 & 0.00291 & 0.13799 & 0.20227 \\
    12 & 4 & 40	& 0.00034 & 0.00182 & 0.08051 & 0.14175 \\
    \hline
    \end{tabular}
    \caption{Average computation times for each process of the Network Scheduler control application, from the simulations of randomly generated network topologies parametrized by Table~\ref{tab: random topology parameters}. These values are based on the simulations with default Arqon, where the bonus allocation phase of the Compute Schedule process is enabled.}
    \label{tab: random topos computation times}
\end{table*}

\section{Implementation: Supporting Results and Formal Definitions}\label{app: Implementation}

In this section we provide formal definitions for the set of allowed edges in a network resource graph and for the set of allowed entanglement generation paths that we consider in our implementation of Arqon. Then, we prove the supporting results Lemma~\ref{lemma: specific path partition disjoint} and Lemma~\ref{lemma: specific FCs well behaved} describing the properties of the specific choice of path partition and the set of filling classes it induces. Finally, we discuss our methods of setting the duration of a PGA for a given $\ppacket$.

\subsection{Allowed Edges in a Network Resource Graph}

\begin{definition}[Set of Allowed Edges, $\mathcal{E}^{*}$]\label{def: allowed edges}
Let ${v_i, v_{j} \in \mathcal{V}}$ be two vertices in the network resource graph. An edge ${e = (v_i, v_{j})}$ is an \emph{\textbf{allowed edge}} if the following conditions are satisfied:
\begin{enumerate}
    \item If ${v_{i} \in E}$ then ${v_{j} \in I}$. Reflexively, if ${v_{j} \in E}$, then ${v_{i} \in I}$. 
    \item If ${v_i \in B}$, then ${v_{j} \in J}$. Reflexively, if ${v_{j} \in B}$, then ${v_{i} \in J}$.
    \item If ${v_i \in J \wedge v_{j} \ \cancel{\in} \ B}$, then ${v_{j} \in I}$. Reflexively, if ${v_{j} \in J \wedge v_{i} \ \cancel{\in} \ B}$, then ${v_{i} \in I}$.
\end{enumerate}
The \emph{\textbf{set of allowed edges}} is \begin{equation}\label{eq: set of allowed edges}
    \mathcal{E}^{*} := \{(v_i, v_j) \ | \ (1. \vee 2. \vee 3.) \wedge v_i, \ v_j \in \mathcal{V} \}.
\end{equation}
\end{definition}

\subsection{The set of allowed entanglement generation paths}

\begin{definition}[Allowed entanglement generation paths]\label{def: allowed entanglement generation paths}
Let $\mathcal{G} = (\mathcal{V}, \mathcal{E})$ be an internally connected network resource graph with ${\mathcal{V} = E\sqcup I \sqcup J \sqcup B}$ and ${\mathcal{E}\subseteq \mathcal{E}^{*}}$. Let $\mathcal{L}$ be the set of local areas of $\mathcal{G}$. 
Define \begin{equation}
    \pi_{e,e'} = \big{\{} \pi=(e, v_1, \cdots, v_k, e') \in \mathcal{P}_{\text{\emph{valid}}} \big{\}},
\end{equation}
the set of all valid entanglement generation paths between two specific end nodes ${e, e' \in E}$, of variable length ${k=|\pi_{e,e'}| - 2}$.
The set $\mathcal{P}_{\text{\emph{allowed}}} \subset \mathcal{P}_{\text{\emph{valid}}}$ of \emph{\textbf{allowed entanglement generation paths}} is a restriction of the set of valid entanglement generation paths, as follows 
\begin{align}
    \mathcal{P}_{\text{\emph{allowed}}} = \bigg{\{}& \pi_{e,e'} \ \forall e, e' \in E \ \bigg{|} \pi_{e,e'} \in B^{*}(e, e') \ \wedge \nonumber \\ 
     &\forall \mathcal{L}_{y} \in \mathcal{L},  \pi_{e,e'} \in L^{*}_{y}(e, e', b^{*}) \bigg{\}},
\end{align}
where \begin{equation}
    B^{*}(e,e') = \{ \pi_{e,e'} \ \big{|} \ \sum_{x=1}^{k=|\pi_{e,e'}| - 2} \mathbbm{1}_{(v_{x} \in B)} = b^{*(e,e')}\},
\end{equation}
\begin{align}
    L^{*}_{y}(e,e',&b^{*}) = \big{\{} \pi_{e,e'} \in B^{*}(e,e') \ \big{|} \nonumber \\
    &\sum_{x=1}^{k=|\pi_{e,e'}| - 2} \mathbbm{1}_{(v_{x} \in \mathcal{L}_y)} \mathbbm{1}_{(v_{x} \in I)} = l^{*}_{y}(e,e',b^{*}) \big{\}},
\end{align}
and 
\begin{equation}
    b^{*}(e,e') = \min_{\pi_{e,e'}} \sum_{x=1}^{k=|\pi_{e,e'}| - 2} \mathbbm{1}_{(v_{x} \in B)},
\end{equation}
\begin{equation}
     l^{*}_{y}(e, e', b^{*}) = \min_{\pi_{e,e'}} \sum_{x=1}^{k=|\pi_{e,e'}| - 2} \mathbbm{1}_{(v_{x} \in \mathcal{L}_y)} \mathbbm{1}_{(v_{x} \in I)},
\end{equation}
and $\mathbbm{1}_{(A)}$ is an indicator function that takes value 1 if $A$ is true and value 0 if $A$ is false. 
\end{definition}
Allowed paths are restricted to those that minimize the number of backbones traversed in the path between two end nodes. Of those paths, they are also shortest paths within each local area.

\subsection{Specific Choice of Path Partition}

\SpecificPathPartitionDisjoint*

\begin{proof}[Proof of Lemma~\ref{lemma: specific path partition disjoint}]\label{proof: implementation: specific path partition disjoint}
    \textbf{Disjointness}: 
    By construction, if a path is in $\Pi^B$ it cannot be in any $\Pi_j^J$ nor any $\Pi_i^I$, and likewise paths in any $\Pi_j^J$ cannot be in any $\Pi^I_i$. 
    Therefore, it suffices to show that any path can be in at most one $\Pi_j^J$ or at most one $\Pi_i^I$.
    
    Suppose that there exists some $\pi, i, i'$ with $i\neq i'$ such that $\pi\in \Pi_i^I$ and $\pi\in\Pi_{i'}^I$. 
    By definition of the set of allowed edges $\mathcal{E}^{*}$,  $(i,i')\notin\mathcal{E}^{*}$. 
    Since $\mathcal E \subseteq \mathcal{E}^{*}$,  $(i,i')\notin\mathcal{E}$, thus $p = (u,...,i, \bullet,..., i',..., v)$ for some end nodes $u,v\in E$. 
    The only allowed edges $(i,\bullet)$ are $(i,w)$ for some $w\in E$, or $(i,j)$ for some $j\in J$. 
    In the former case, this implies $p\notin \mathcal{P}_{\text{valid}}$ and hence also $p\notin \mathcal{P}_{\text{allowed}}$. 
    In the latter case, $\exists k~s.t.~j\in J_k\implies p\in\Pi_k^J$.
    In both cases, this contradicts that $p\in\Pi_i^I$.
    Therefore, no such $p,i,i'$ can exist. 

    Suppose instead that there exists $p,j,j'$ with $j\neq j'$ such that $p\in\Pi^J_j$ and $p\in\Pi_{j'}^J$. 
    By the definitions of $\mathcal{E}^{*}$ and $\mathcal{P}_{\text{allowed}}$, the only paths from nodes in $J_j$ to nodes in $J_{j'}$ is via some $b\in B$. 
    Therefore, there exists some $k$ such that $p_k = b$ and thus $p\in\Pi^B$. 
    This contradicts that $p\in\Pi_i^J$ and so no such $p,j,j'$ can exist

    \textbf{Completeness}: 
    We proceed by construction.
    Consider some path $p = (e, \bullet, ...)$, for $e\in E$. By definition of $\mathcal{E}^{*}$, and since $\mathcal{E} \subseteq \mathcal{E}^{*}$, the only possible edges $(e,\bullet)$ are of the form $(e,i)$ for some $i\in I$. Then, $p = (e, i, \bullet, \cdots), $ for some $i \in I$. 
    By definition of $\mathcal{E}^{*}$, and since $\mathcal{E} \subseteq \mathcal{E}^{*}$, the only allowed edges are of the form $(i,e')$ for some $e' \in E$, or $(i, j)$ for some $j \in J$. 
    In the former case, $p \in \Pi_i$. 
    In the later case, since the $J_k$ form a disjoint partition of $J$, $\exists k$ such that $j\in J_k$. Therefore, if $p\notin\Pi^B$, then $p\in\Pi_k^J$.
    Otherwise, $p \in \Pi^B$. 
    Therefore, all paths in $\mathcal{P}_{\text{allowed}}$ are in some $\Pi$, and so $\Pi$ forms a disjoint partition of $\mathcal P_{\text{allowed}}$.

\end{proof}

\subsection{Specific Choice of Filling Classes}
Here we prove that the specific choice of filling classes used in our implementation is well-behaved. 

\begin{restatable}[The Specific Choice of Filling Classes is Well-Behaved]{lemma}{SpecificFCsWB}\label{lemma: specific FCs well behaved} 
    Let $\mathcal G$ be an internally connected network resource graph. The specific choice of filling classes $\Phi$ defined in~\eqref{eq: specific FC partition} is a well-behaved set of filling classes.
\end{restatable}

\begin{proof}
Let $\mathcal{L} = \{L_i = (\mathcal V_i, \mathcal E_i) \}$ be the set of local areas of $\mathcal{G}$. Let $J_i = J \cap \mathcal{V}_i$ be the junction nodes in sub-graph $L_i$.  
We need to show that $\Phi$ satisfies both conditions of Definition \ref{def: well-behaved set of filling classes}.
\begin{enumerate}
    \item Distinct Resource Sets: \newline
    \textbf{Claim:} $\forall \phi, \phi' \in \Phi$ with $\phi \neq \phi'$, we have ${\xi(\Pi_{\phi}) \neq \xi(\Pi_{\phi'})}$.

    \textbf{Proof:} 
    We consider all possible pairs of distinct filling classes:
    \begin{enumerate}
        \item $\phi = \phi^{B}$ and $\phi' = \phi^J_j$ for some $j \in J$.

        Since every path $\pi \in \Pi^B$ contains at least one backbone node $b \in B$, we have ${B \cap \xi(\Pi^B) \neq \emptyset}$.
        By definition, $\Pi^J_j \subseteq \mathcal{P}_{\text{allowed}} \setminus \Pi^B$, so paths in $\Pi^J_j$ contain no backbone nodes.
        Therefore ${B \cap \xi(\Pi^J_j) = \emptyset}$. It follows that $\xi(\Pi^B) \neq \xi(\Pi^J_j)$.
        
        \item $\phi = \phi^B$ and $\phi' = \phi^I_i$ for some $i \in I$.

        Similarly to case (a), ${\Pi^I_i \subseteq \mathcal{P}_{\text{allowed}} \setminus \Pi^B}$, so $B \cap \xi(\Pi^I_i) = \emptyset$. But ${B \cap \xi(\Pi^B) \neq \emptyset}$, so ${\xi(\Pi^B) \neq \xi(\Pi^I_i)}$.

        \item $\phi = \phi^J_j$ and $\phi' = \phi^J_{j'}$ for $j \neq j' \in \{ 1, \cdots,  |\mathcal{L}|\}$.

        Paths $\pi \in \Pi^J_j$ contain at least one junction node from $J_j$ and paths $ \pi' \in \Pi^J_{j'}$ contain at least one junction node from $J_{j'}$.
        Since, with the exclusion of end nodes, local areas have disjoint vertex sets (by Definition \ref{def: local area} of local areas), $J_j \cap J_{j'} = \emptyset$.
        Therefore, $\xi(\Pi^J_j) \neq \xi(\Pi^J_{j'})$.

        \item $\phi = \phi^J_j$ and $\phi' = \phi^I_i$ for some $j  \in J$ and $i \in I$.

        By definition, paths in $\Pi^{I}_i$ contain no junction nodes. 
        Therefore, $\xi(\Pi^{I}_i) \cap J_{j'} = \emptyset, \ \forall j'$. Paths in $\Pi^J_j$ on the other hand contain at least one junction node $j' \in J_j$, hence $\xi(\Pi^J_j) \cap J_j \neq \emptyset$. 
        Therefore, $\xi(\Pi^J_j) \neq \xi(\Pi^I_i)$.

        \item $\phi = \phi^I_i$ and $\phi' = \phi^I_{i'}$ for $i \neq i' \in I$.

        By definition, paths in $\Pi^I_i$ must contain the entanglement generation interface $i$, similarly paths in  $\Pi^I_{i'}$ must contain the entanglement generation interface $i'$. 
        Based on Definition~\ref{def: allowed edges} of allowed edges, $\mathcal{E^{*}}$, each path in $\Pi^I_i \ (\Pi^I_{i'})$ can contain at most one entanglement generation interface. 
        It follows that $\xi(\Pi^I_i) \neq \xi(\Pi^I_{i'})$, and moreover $\xi(\Pi^I_i) \cap \xi(\Pi^I_{i'}) = \emptyset$.

    \end{enumerate}

    \item Intersection of Associated Resources: \newline
    \textbf{Claim:} $\forall \phi, \phi' \in \Phi$ with $\phi \neq \phi'$, if $\xi(\Pi_\phi) \cap \xi(\Pi_{\phi'}) \neq \emptyset$, then either $\xi(\Pi_{\phi}) \subset \xi(\Pi_{\phi'})$ or $\xi(\Pi_{\phi'}) \subset \xi(\Pi_{\phi})$.

    \textbf{Proof:} In consideration of the structure of the specific choice of path partition, we analyze each case where sets of associated resources can intersect:
    \begin{enumerate}
        \item $\phi = \phi^B$ and $\phi' \in \{\phi^J_j : j \in J, \phi^I_i: i \in I\}$. 

        Consider arbitrary $r \in \xi(\Pi_{\phi'})$.
        Then, $\exists \pi \in \Pi_{\phi'}$ such that $r \in \pi$.
        By Definition \ref{def: specific path partition}, $\pi$ is a path in $\mathcal{P}_{\text{allowed}}$ that does not contain backbone vertices.
        However, since $\mathcal{G}$ is internally connected (Definition \ref{def: internally connected}), there exists a path $\pi_b \in \mathcal{P}_{\text{allowed}}$ that: \begin{enumerate}
            \item contains at least one backbone vertex (so $\pi_b \in \Pi^B$),
            \item and traverses through resource $r$ (since $\mathcal{G}$ is internally connected and $\mathcal{P}_{\text{allowed}}$ does not decrease the network connectivity, but only restricts to a particular type of shortest path between any two nodes).
        \end{enumerate}

Therefore, $r \in \xi(\Pi^B)$. Since r was arbitrary, $\xi(\Pi_{\phi'}) \subseteq \xi(\Pi^{B})$.
Moreover, the inclusion is proper, as $B \cap \xi(\Pi_b) \neq \emptyset$ while $B \cap \xi(\Pi_{\phi'}) = \emptyset$ (from Condition 1).
Hence, $\xi(\Pi_{\phi'}) \subset \xi(\Pi^{B})$.

\item $\phi = \phi^J_j$ and $\phi' = \phi^J_{j'}$ for $j \neq j'  \in J$.

In case 1. (c), we noted that, with the exclusion of end nodes, local areas have disjoint vertex sets. Since paths in $\Pi^J_j$, respectively $\Pi^J_{j'}$, are entirely contained in the local areas $L_j$, respectively  $L_{j'}$, $\xi(\Pi^J_j) \cap \xi(\Pi^J_{j'}) = \emptyset$ and this case does not satisfy the premise of condition 2.

\item $\phi = \phi^J_j$ and $\phi' = \phi^I_i$ for some $j  \in J$ and $i \in I$.

There are two sub-cases. If ${i \notin \mathcal{V}_j}$, then ${\xi(\Pi^J_j) \cap \xi(\Pi^I_i) = \emptyset}$, since distinct local areas have distinct vertex sets. 
Hence this sub-case does not satisfy the premise of condition 2. 
Alternatively, ${i \in \mathcal{V}_j}$ and ${\xi(\Pi^J_j) \cap \xi(\Pi^I_i) \neq \emptyset}$.
In this sub-case, consider arbitrary ${r \in \xi(\Pi^I_i)}$. Then ${\exists \pi_i \in \Pi^I_i}$ such that ${r \in \pi_i}$.
By definition of $\Pi^I_i$ and the set of local areas, all vertices in $\pi_i$ are in $\mathcal{V}_j$.  
By the internal connectivity of local area $L_j$, there exists an allowed path $\pi_j \in L_j$ that: \begin{enumerate}
    \item Contains at least one junction $j' \in J_j$ (so $\pi_j \in \Pi^J_j$), 
    \item and traverses resource $r$.
\end{enumerate}

Therefore, $r \in \xi(\Pi^J_j)$. Since $r$ was arbitrary, $\xi(\Pi^I_i) \subseteq \xi(\Pi^J_j)$.
Moreover, the inclusion is proper since, from 1. (d),  $J_j \cap \xi(\Pi^J_j) \neq \emptyset$ while $J_j \cap \xi(\Pi^I_i) = \emptyset$.
Hence, $\xi(\Pi^I_i) \subset \xi(\Pi^J_j)$.

\item $\phi = \phi^I_i$ and $\phi' = \phi^I_{i'}$ for $i \neq i' \in I$.

In 1. (e) we showed that $\xi(\Pi^I_i) \cap \xi(\Pi^I_{i'}) = \emptyset$, and thus this case does not satisfy the premise of condition 2.
    \end{enumerate}
\end{enumerate}

\end{proof}

\subsection{Determining the PGA Execution Time}
We need to determine the duration of time $E_d$ s.t. 
\begin{equation}\label{eq: set E for ppacket} \mathbb P\left[\parbox{14em}{packet generated before time $E_d$}\right] \geq \ppacket.\end{equation}
For packets $\mathfrak{p}$ with $s=1$, this can be done with the geometric/exponential distribution, but for packets where the number of requested pairs $s>1$, the exact calculation of \eqref{eq: set E for ppacket} is not known. 
There are however results from \textit{scan statistics} which give very accurate estimations of this probability for a known number of attempts to generate entanglement $m$ and probability of each attempt succeeding. Such results can be found in, e.g.~\cite{alm_distributions_1983, naus_approximations_1982, glaz_scan_2001, glaz_scan_1999}. 
We translate between a number of attempts $m$ of some concrete duration and each with a probability of success $p_{\text{succ}}$ and the rate of successful link generation $R_{\pi}$ on an end-to-end path $\pi$ for a total duration of time $E_d$. For full details of our treatment of this problem, see~\cite{NetArx}, Appendix~C.
Every PGA from a demand $d$ is ultimately assigned an execution time $E_d$ that is long enough that each PGA succeeds with probability $\ppacket$.

\onecolumn


\section{Code Listings}
\label{app: code listings}

\definecolor{mygreen}{rgb}{0,0.6,0}
\definecolor{mygray}{rgb}{0.5,0.5,0.5}
\definecolor{mymauve}{rgb}{0.58,0,0.82}

\lstset{ %
  backgroundcolor=\color{white},   
  basicstyle=\footnotesize,        
  breaklines=true,                 
  captionpos=b,                    
  commentstyle=\color{mygreen},    
  escapeinside={\%*}{*)},          
  keywordstyle=\color{blue},       
  stringstyle=\color{mymauve},     
  numbers=left,
  numberstyle=\small,
  frame=single,
}
\subsection{Compute Network Schedule}
\subsubsection{Minimal Allocation}

\begin{lstlisting}[firstline=18, lastline=64,language=python, caption={Implementation of the minimal allocation phase of Compute Schedule, via the \textsc{DirectAllocation} scheduler, Algorithm~\ref{alg: direct allocation subroutine}. Note that rather than returning an updated network schedule, this subroutine simply updates the provided \texttt{NetworkSchedule} object.}, label={lst: directallocation}]
def minimal_allocation_scheduler(
        filling_class: FillingClass,
        network_schedule: NetworkSchedule,
        start_time: float
):
    if not filling_class:
        return

    network_schedule.add_tasks_minimal(filling_class)

    # logger = LogManager.get_resource_scheduler_logger()

    cycle_start_time = start_time

    ponters = {r: -100 for r in network_schedule.schedule.keys()}

    for i in range(len(filling_class)):

        cycle_time = filling_class.cycle_times[i]

        k = 0

        for k in range(filling_class.numbers_of_cycles[i]):
            for j, task in enumerate(filling_class[i:]):
                offset = sum(t.execution_time for t in filling_class[i:i + j])
                _start_to_schedule = time.process_time()
                network_schedule.schedule_packet_generation_attempt(
                    task, cycle_start_time + cycle_time * k + offset,
                    True,
                    pga_id=sum(filling_class.numbers_of_cycles[:i]) + k,
                    schedule_pointers=ponters)

        _start_to_schedule = time.process_time()
        network_schedule.schedule_packet_generation_attempt(
            filling_class[i],
            cycle_start_time + cycle_time * (
                filling_class.numbers_of_cycles[
                    i]), True,
            pga_id=sum(
                filling_class.numbers_of_cycles[
                    :i]) + k,
            schedule_pointers=ponters)

        cycle_start_time += filling_class[i].execution_time + (
            filling_class.numbers_of_cycles[i]) * cycle_time

\end{lstlisting}

\subsubsection{Bonus Round Robin}
\begin{lstlisting}[firstline=65, lastline=261,language=python, caption={Implementation of Algorithm~\ref{alg: bonus round robin resource schedule} (\textsc{BonusRoundRobin}) in Python. Note that rather than returning an updated network schedule, this subroutine simply updates the provided \texttt{NetworkSchedule} object.}, label={lst: bonusRR}]
def bonus_phase_round_robin_scheduler(
        filling_class: FillingClass,
        network_schedule: NetworkSchedule,
):
    # logger = LogManager.get_resource_scheduler_logger()

    if not filling_class:
        return

    _resources = [r for r in network_schedule.schedule.keys() if
                  not r & 1 << get_component_id_size() - 1]

    schedule_pointers = {r: 0 for r in _resources}

    release_times_to_consider = set()

    round_robin_index = 0

    next_start_time = 0

    network_schedule.add_tasks_bonus(filling_class)

    def _get_next_start_time() -> float | None:
        nonlocal next_start_time

        end_times = {
            r: network_schedule.schedule[r][schedule_pointers[r]].end_time for r
            in _resources if
            schedule_pointers[r] < len(network_schedule.schedule[r])}

        # logger.debug(f'considering end times: {end_times}')

        _next_end_time = min(end_times.values()) if end_times.values() else inf
        # logger.debug(f"next end time: {_next_end_time}")

        _minimising_resources = [
            r for r in end_times.keys() if
            math.isclose(end_times[r],
                         _next_end_time)] if _next_end_time != inf else []

        _index, _next_release_time = min(
            enumerate(release_times_to_consider),
            key=lambda t: t[1]
        ) if release_times_to_consider else (0, inf)

        # logger.debug(f"considering release times: {release_times_to_consider}")
        # logger.debug(f"next release time: {_next_release_time}")

        if _next_end_time == inf and _next_release_time == inf:
            return None

        if _next_end_time <= _next_release_time:
            for r in _minimising_resources:
                schedule_pointers[r] += 2

            if _next_end_time in release_times_to_consider:
                release_times_to_consider.remove(_next_end_time)
                # logger.debug("Also removed end time from release_times_to_consider")
            return _next_end_time

        else:

            release_times_to_consider.remove(_next_release_time)
            return _next_release_time

    def _task_indices_to_inspect(start: int) -> Generator[int, None, None]:
        for i in range(len(filling_class)):
            yield (start + i) % len(filling_class)

    def _resources_available(_tts: PacketGenerationTask) -> bool:
        for resource in _tts.path[1:-1]:
            if schedule_pointers[resource] >= len(
                    network_schedule.schedule[resource]):
                continue
            if network_schedule.schedule[resource][schedule_pointers[
                resource]].start_time < next_start_time + candidate_task.execution_time:
                return False
        return True

    def _future_minsep_violation(_tts: PacketGenerationTask) -> bool:
        _reference_resource = _tts.path[1]
        _schedule_except = network_schedule.schedule[_reference_resource]
        offset = 0
        while (
                schedule_pointers[_reference_resource] + offset < len(
            _schedule_except)
                and
                _schedule_except[
                    schedule_pointers[_reference_resource] + offset
                ].start_time < next_start_time + _tts.execution_time + _tts.minsep
        ):
            if _schedule_except[
                schedule_pointers[_reference_resource] + offset
            ].task_id == _tts.task_id:
                return True
            offset += 1
        return False

    def _past_minsep_violation(_tts: PacketGenerationTask) -> Tuple[
        bool, float | None]:
        _reference_resource = _tts.path[1]
        _schedule_excerpt = network_schedule.schedule[_reference_resource]
        offset = 1 if schedule_pointers[_reference_resource] < len(
            _schedule_excerpt) else 2

        while (
                schedule_pointers[_reference_resource] - offset >= 0
                and (
                        _schedule_excerpt[
                            schedule_pointers[
                                _reference_resource
                            ] - offset
                        ].end_time > next_start_time - _tts.minsep
                        and not
                        np.isclose(
                            _schedule_excerpt[
                                schedule_pointers[
                                    _reference_resource
                                ] - offset
                            ].end_time,
                            next_start_time - _tts.minsep
                        )
                )
        ):

            if _schedule_excerpt[
                schedule_pointers[
                    _reference_resource
                ] - offset
            ].task_id == _tts.task_id:
                return True, _schedule_excerpt[
                                 schedule_pointers[
                                     _reference_resource
                                 ] - offset
                                 ].end_time + _tts.minsep
            offset += 1
            if schedule_pointers[_reference_resource] - offset < 0:
                break

        return False, None

    if all(network_schedule.schedule[r][0].start_time > 0 for r in
           schedule_pointers.keys() if network_schedule.schedule[r]):
        next_start_time = min(network_schedule.schedule[r][0].end_time for r in
                              schedule_pointers.keys())

    while next_start_time < network_schedule.length:

        cycle_start_time = time.process_time()

        for index in _task_indices_to_inspect(round_robin_index):
            candidate_task = filling_class[index]
            # logger.debug(f'candidate_task: {candidate_task}')

            if (
                    not _resources_available(candidate_task)
                    or _future_minsep_violation(candidate_task)
                    or next_start_time > candidate_task.expiry_time
                    or next_start_time + candidate_task.execution_time >= network_schedule.length
            ):
                continue

            past_minsep_violation, release_time = _past_minsep_violation(
                candidate_task)
            if past_minsep_violation:
                if release_time > next_start_time:
                    release_times_to_consider.add(release_time)
                continue

            network_schedule.schedule_packet_generation_attempt(
                candidate_task,
                next_start_time,
                False,
                schedule_pointers=schedule_pointers
            )

            round_robin_index = index + 1

        old_start_time = 0 + next_start_time

        next_start_time = _get_next_start_time()

        if next_start_time is None:
            return

        assert old_start_time < next_start_time

        # logger.debug(f"Completed cycle in time {time.process_time() - cycle_start_time:.5g}")

        if next_start_time is None:
            break

        # logger.debug(f"next start time is {next_start_time}")

        pass

\end{lstlisting}

\subsubsection{Complete Scheduler}
\begin{lstlisting}[firstline=463, lastline=502,language=python, caption={Implementation of Algorithm~\ref{alg: Compute Schedule} in Python.}, label={lst: complete compute schedule}]
def complete_schedule_scheduler(
        filling_classes: SetOfFillingClasses,
        network_schedule: NetworkSchedule,
) -> Tuple[float, float]:
    logger = LogManager.get_resource_scheduler_logger()

    _start_minimal = time.process_time()
    logger.info("Computing Minimal Schedule")
    for filling_class in filling_classes.all_filling_classes:
        logger.debug(
            f"Starting to compute minimal schedule for filling class {filling_class}"
        )
        start_time = sum(fc.minimum_service_time for fc in
                         filling_classes.preceding_filling_classes[
                             filling_class])
        minimal_allocation_scheduler(filling_class, network_schedule,
                                     start_time)

    minimal_time = time.process_time() - _start_minimal

    logger.debug(
        f"Completed minimal phase in time {time.process_time() - _start_minimal:.5g}"
    )

    _start_bonus = time.process_time()

    logger.info("Computing Bonus Schedule")
    for fc in filling_classes.all_filling_classes:
        logger.debug(
            f"Starting to compute bonus schedule for filling class {fc}")
        bonus_phase_round_robin_scheduler(fc, network_schedule)

    bonus_time = time.process_time() - _start_bonus

    logger.debug(
        f"Completed bonus phase in time {time.process_time() - _start_bonus:.5g}"
    )

    return minimal_time, bonus_time
\end{lstlisting}

\subsection{Admit New Tasks}
\begin{lstlisting}[linerange={9-55,59-59},language=python, caption={Implementation of Algorithm~\ref{alg: Admit Tasks} in Python.}, label={lst: admit new tasks}]
def admit_new_tasks(filling_classes: SetOfFillingClasses, task_intake_object: Dict[int, List[PacketGenerationTask]], resources: List[int]) -> Dict[int, int]:
    """
    :param filling_classes: set of filling classes
    :param task_intake_object: task intake object in format {demand_id:[PGT]}
    :param resources: list of resources
    :return: list of accepted demands
    """

    logger = LogManager.get_admission_control_logger()

    available_time = {r : get_scheduling_interval_duration() for r in resources}
    for filling_class in filling_classes.all_filling_classes:
        for resource in filling_class.associated_resources:
            available_time[resource] -= filling_class.minimum_service_time

    accepted_demands = {}

    for demand in task_intake_object:
        logger.info(f"Considering demand {demand}")
        for pgt_to_admit in task_intake_object[demand]:
            logger.info(f"Considering PGT {pgt_to_admit}")
            try:
                filling_class = filling_classes.get_filling_class(pgt_to_admit)
            except FillingClassError:
                logger.info(f"Rejected PGT {pgt_to_admit} because no filling class was found")
                continue

            for resource in filling_class.associated_resources:
                available_time[resource] += filling_class.minimum_service_time

            filling_class.add_packet_generation_task(pgt_to_admit)
            if all(available_time[resource] - filling_class.minimum_service_time > 0 for resource in filling_class.associated_resources):


                accepted_demands[demand] = pgt_to_admit.task_id
                for resource in filling_class.associated_resources:
                    available_time[resource] -= filling_class.minimum_service_time

                logger.debug(f"Remaining available time: {available_time}")
                logger.info(f"accepted demand: {demand} with PGT {pgt_to_admit}")

                break
            else:
                logger.info(f"rejected PGT for insufficient time. Demand ID: {demand}, PGT: {pgt_to_admit}")
                filling_class.remove(pgt_to_admit)
                for resource in filling_class.associated_resources:
                    available_time[resource] -= filling_class.minimum_service_time
    return accepted_demands
\end{lstlisting}

\subsection{Update Filling Classes}
\clearpage
\begin{lstlisting}[linerange={8-160}, language=python, caption={Implementation of Algorithm~\ref{alg: Update Filling Classes} (\textsc{UpdateFillingClasses})}, label={lst: update filling classes}]
class CompletedDemands:
    expired_demands: List[int] = dataclasses.field(default_factory=list)
    removed_demands: List[int] = dataclasses.field(default_factory=list)
    terminated_demands: List[int] = dataclasses.field(default_factory=list)

    @property
    def all_removed_demands(self):
        return self.removed_demands + self.terminated_demands + self.expired_demands

def update_filling_classes(
        old_filling_classes: SetOfFillingClasses | None,
        path_partition: List[Set[Tuple[int, ...]]] | None,
        termination_buffer: Set[int],
        expiration_cutoff_time: float
) -> Tuple[
    SetOfFillingClasses,
    CompletedDemands,
]:

    completed_demands = CompletedDemands()

    if path_partition is not None:
        new_filling_classes = SetOfFillingClasses(path_partition)

        failed_tasks = new_filling_classes.add_packet_generation_tasks(old_filling_classes.all_tasks if old_filling_classes else [])
        completed_demands.removed_demands = [t.parent_demand_id for t in failed_tasks]

    else:
        new_filling_classes = old_filling_classes

    expired, terminated = new_filling_classes.update_live_tasks(termination_buffer, expiration_cutoff_time)

    completed_demands.expired_demands = [task.parent_demand_id for task in expired]
    completed_demands.terminated_demands = [task.parent_demand_id for task in terminated]

    return new_filling_classes, completed_demands

\end{lstlisting}

\clearpage
\begin{lstlisting}[linerange={8-160}, language=python, caption={Implementation of a single filling classes $\phi$. Any properties of a filling class (e.g. \texttt{minimum\_service\_time}) which are \texttt{None} when called are calculated and stored. This means that subsequent calls to get the property only require $O(1)$ time rather than recomputing the value each time.}, label={lst: single filling class}]
class FillingClass(List[PacketGenerationTask]):

    def __init__(self, path_partition: Set[Tuple[int, ...]]) -> None:
        super().__init__()
        self._path_partition: Set[Tuple[int, ...]] = path_partition
        self._associated_resources: Set[int] | None = None
        self._numbers_of_cycles: List[int] | None = None
        self._cycle_times: List[float] | None = None
        self._minimum_service_time: float | None = None

    def _calculate_associated_resources(self) -> None:
        self._associated_resources = set(
            [r for path in self._path_partition for r in path[1:-1]])

    @property
    def path_partition(self):
        return self._path_partition

    @property
    def associated_resources(self):
        if self._associated_resources is None:
            self._calculate_associated_resources()
        return self._associated_resources

    @property
    def associated_nodes(self) -> Set[int]:
        return {r for path in self._path_partition for r in path}

    def calculate_cycle_times(self):

        # self.sort(key=lambda task: task.minimum_allocation)

        self._cycle_times = [
            max([t.execution_time + t.minsep for t in self[i:]] + [
                sum(t.execution_time for t in self[i:])]) for i in
            range(len(self))]

        pass

    def calculate_numbers_of_cycles(self):

        # self.sort(key=lambda task: task.minimum_allocation)

        _list = [t.minimum_allocation for t in self]
        for i in range(len(self) - 1, 0, -1):
            _list[i] = _list[i] - _list[i - 1]

        _list[0] -= 1

        self._numbers_of_cycles = _list

    def calculate_minimum_service_time(self) -> None:
        self._minimum_service_time = sum(x * y for x, y in
                                         zip(self.numbers_of_cycles,
                                             self.cycle_times)) + sum(
            t.execution_time for t in self)

    def remove_expired_tasks(self, cutoff_time: float) -> List[
        PacketGenerationTask]:
        expired_tasks = [task for task in self if
                         task.expiry_time <= cutoff_time]
        self[:] = [task for task in self if task not in expired_tasks]
        return expired_tasks

    def remove_terminated_demands(self, terminated_demands: Iterable[int]) -> \
    List[PacketGenerationTask]:
        terminated_tasks = [task for task in self if
                            task.parent_demand_id in terminated_demands]
        self[:] = [task for task in self if task not in terminated_tasks]
        return terminated_tasks

    def update_live_tasks(
            self,
            terminated_demands: Set[int],
            current_time: float
    ) -> Tuple[
        List[PacketGenerationTask],
        List[PacketGenerationTask],
    ]:

        expired_tasks = self.remove_expired_tasks(current_time)
        terminated_tasks = self.remove_terminated_demands(terminated_demands)

        if self:
            self.calculate_numbers_of_cycles()
            self.calculate_cycle_times()
            self.calculate_minimum_service_time()
        else:
            self._numbers_of_cycles = None
            self._cycle_times = None
            self._minimum_service_time = None

        return expired_tasks, terminated_tasks

    @property
    def numbers_of_cycles(self):
        if self._numbers_of_cycles is None or len(
                self._numbers_of_cycles) != len(self):
            self.calculate_numbers_of_cycles()

        return self._numbers_of_cycles

    @property
    def cycle_times(self):
        if self._cycle_times is None or len(self._cycle_times) != len(self):
            self.calculate_cycle_times()
        return self._cycle_times

    @property
    def minimum_service_time(self):

        if not self:
            return 0

        if self._minimum_service_time is None:
            self.calculate_minimum_service_time()
        return self._minimum_service_time

    def add_packet_generation_task(self, task: PacketGenerationTask):
        if not isinstance(task, PacketGenerationTask):
            raise TypeError(
                f'Task must be of type PacketGenerationTask (provided type {type(task)})')
        _i = 0
        while _i < len(self) and self[
            _i].minimum_allocation < task.minimum_allocation:
            _i += 1
        self.insert(_i, task)
        self._numbers_of_cycles = None
        self._cycle_times = None
        self._minimum_service_time = None

    def clear(self):
        super().clear()
        self._numbers_of_cycles = None
        self._cycle_times = None
        self._minimum_service_time = None

    def append(self, task: PacketGenerationTask):
        if not isinstance(task, PacketGenerationTask):
            raise TypeError(
                f"Task must be of type PacketGenerationTask (provided type {type(task)})")
        else:
            super().append(task)
            self._numbers_of_cycles = None
            self._cycle_times = None
            self._minimum_service_time = None

    def remove(self, __value):
        super().remove(__value)
        self._numbers_of_cycles = None
        self._cycle_times = None
        self._minimum_service_time = None
\end{lstlisting}

\begin{lstlisting}[linerange={195-248}, language=python, caption={Implementation of the set of filling classes $\Phi$. Any properties of a filling class (e.g. \texttt{minimum\_service\_time}) which are \texttt{None} when called are calculated and stored. This means that subsequent calls to get the property only require $O(1)$ time rather than recomputing the value each time.}, label={lst: set of filling classes}]
class SetOfFillingClasses:

    def __init__(self, path_partition: List[Set[Tuple[int, ...]]]) -> None:

        self._all_filling_classes: List[FillingClass] = [FillingClass(partition)
                                                         for partition in
                                                         path_partition]
        self._all_filling_classes.sort(reverse=True)
        self._preceding_filling_classes: Dict[
            FillingClass, Set[FillingClass]] = {
            fc: {_fc for _fc in self._all_filling_classes if _fc > fc} if fc.path_partition else set() for fc in
            self._all_filling_classes}

        self._interface_filling_classes: List[FillingClass] | None = None
        self._local_area_filling_classes: List[FillingClass] | None = None
        self._backbone_filling_classes: List[FillingClass] | None = None

    def add_packet_generation_tasks(self, pgts_to_add: Iterable[
        PacketGenerationTask]) -> List[PacketGenerationTask]:
        failed_tasks: List[PacketGenerationTask] = []
        for pgt in pgts_to_add:
            try:
                self.get_filling_class(pgt).add_packet_generation_task(pgt)
            except FillingClassError:
                failed_tasks.append(pgt)

        return failed_tasks

    def update_live_tasks(
            self,
            terminated_demands: Set[int],
            current_time: float
    ) -> Tuple[
        List[PacketGenerationTask],
        List[PacketGenerationTask],
    ]:
        expired_tasks = []
        terminated_tasks = []

        for filling_class in self.all_filling_classes:
            _ex, _term = filling_class.update_live_tasks(terminated_demands,
                                                         current_time)
            expired_tasks.extend(_ex)
            terminated_tasks.extend(_term)

        return expired_tasks, terminated_tasks

    def get_filling_class(self, pgt: PacketGenerationTask) -> FillingClass:
        for filling_class in self._all_filling_classes:
            if pgt.path in filling_class.path_partition:
                return filling_class

        raise FillingClassError("No known filling class for pgt")
\end{lstlisting}

\subsection{Network Schedules}
\begin{lstlisting}[linerange={27-36, 79-121, 136-147}, language=python ,caption={Implementation of the network schedule.}, label={lst: resource schedule}]
class PacketGenerationAttempt(ArqonDataRecord):
    src: int
    dst: int
    start_time: float
    end_time: float
    task_id: int
    demand_id: Optional[int]
    pga_id: Optional[int]
    is_minimal: bool


class ResourceSchedule(List[PacketGenerationAttempt]):
    def __init__(self, resource_id: int, seq: tuple = ()):
        super().__init__(seq)
        self._resource_id = resource_id

    @property
    def resource_id(self):
        return self._resource_id



    def add_pga(self, __object: PacketGenerationAttempt, pointer: Optional[int] = None):
        """
        re-written append method to ensure time-ordering of PGAs.
        :param __object: PacketGenerationAttempt to add to resource schedule.
        :param pointer: index of PGA preceding requested time
        """
        if not isinstance(__object, PacketGenerationAttempt):
            raise TypeError(f"Only type PacketGenerationAttempt is allowed. (supplied {type(__object)})")

        else:
            if not self:
                super().append(__object)
                return

            if pointer is None:
                pointer = _bs(self, __object.end_time)

                if pointer > 0:
                    self.insert(pointer + 1, __object)
                elif self[pointer].end_time < __object.end_time:
                    self.insert(1, __object)
                else:
                    self.insert(0, __object)

            else:
                if pointer == 0:
                    self.insert(pointer, __object)
                elif pointer > 0:
                    self.insert(pointer, __object)
                elif pointer < 0:
                    super().append(__object)
        return {
            "node": self.resource_id,
            "packet_generation_attempts": [pga.to_dict() for pga in self]
        }



class NetworkSchedule:

    def __init__(self, start_time: float, length: float, resources: List[int], identifier: int | None = None):

        self._schedule: Dict[int, ResourceSchedule] = {r: ResourceSchedule(r) for r in resources}
\end{lstlisting}

\twocolumn
\end{document}